\documentclass[11pt,a4paper]{article}
\usepackage[margin=2cm]{geometry}
\usepackage[shortlabels]{enumitem} 
\setlist{noitemsep}

\usepackage [only, mapsfrom]{stmaryrd}

\usepackage{undertilde,comment}
\usepackage{amsfonts}
\usepackage{amssymb}
\usepackage{amsthm}
\usepackage{amsmath}
\usepackage{graphicx}
\usepackage{tikz}
\usetikzlibrary{matrix,arrows,patterns,cd}

\usepackage{scalerel}
\usetikzlibrary{svg.path}

\definecolor{orcidlogocol}{HTML}{A6CE39}
\tikzset{
	orcidlogo/.pic={
		\fill[orcidlogocol] svg{M256,128c0,70.7-57.3,128-128,128C57.3,256,0,198.7,0,128C0,57.3,57.3,0,128,0C198.7,0,256,57.3,256,128z};
		\fill[white] svg{M86.3,186.2H70.9V79.1h15.4v48.4V186.2z}
		svg{M108.9,79.1h41.6c39.6,0,57,28.3,57,53.6c0,27.5-21.5,53.6-56.8,53.6h-41.8V79.1z M124.3,172.4h24.5c34.9,0,42.9-26.5,42.9-39.7c0-21.5-13.7-39.7-43.7-39.7h-23.7V172.4z}
		svg{M88.7,56.8c0,5.5-4.5,10.1-10.1,10.1c-5.6,0-10.1-4.6-10.1-10.1c0-5.6,4.5-10.1,10.1-10.1C84.2,46.7,88.7,51.3,88.7,56.8z};
	}
}

\newcommand\orcidicon[1]{\href{https://orcid.org/#1}{\mbox{\scalerel*{
				\begin{tikzpicture}[yscale=-1,transform shape]
					\pic{orcidlogo};
				\end{tikzpicture}
			}{|}}}}

\usepackage{color}
\usepackage[mathscr]{euscript}

\usepackage{soul,todonotes}

\usepackage{leftidx,slashed}

\usepackage{mathrsfs}
\usepackage{datetime} 
\usepackage[pdfpagelabels]{hyperref}
\hypersetup{
	breaklinks=true,   
	colorlinks=true,   
	pdfusetitle=true,  
}

\numberwithin{equation}{section}

\newcommand{\id}{\mathrm{id}}
\DeclareMathOperator{\pr}{pr}

\newcommand{\bbLbrack}{[\kern-0.4em{[}\,}
\newcommand{\bbRbrack}{\,]\kern-0.4em{]}}

\newcommand{\ip}[2]{(#1,#2)}

\newcommand{\ii}{{\rm i}}
\newcommand{\ee}{{\rm e}}
\newcommand{\xu}{\underline{x}}
\newcommand{\ku}{\underline{k}}

\newcommand{\vol}{{\rm vol}}

\newcommand{\dd}{{\rm d}}
\newcommand{\dvol}{{\rm dvol}}
\DeclareMathOperator{\supp}{supp}

\DeclareMathOperator{\End}{End}
\DeclareMathOperator{\Hom}{Hom}
\DeclareMathOperator{\Lin}{Lin} 
\DeclareMathOperator{\Char}{Char}

\newcommand{\1}{1\!\!1}
\newcommand{\beq}{\begin{equation}}
\newcommand{\ene}{\end{equation}}

\newcommand{\RR}{\mathbb{R}}
\newcommand{\CC}{\mathbb{C}}
\newcommand{\ZZ}{\mathbb{Z}}
\newcommand{\NN}{\mathbb{N}}

\newcommand{\pc}{\textnormal{pc}}
\newcommand{\fc}{\textnormal{fc}}
\newcommand{\spc}{\textnormal{spc}}
\newcommand{\sfc}{\textnormal{sfc}}
\renewcommand{\sc}{\textnormal{sc}}
\newcommand{\tc}{\textnormal{tc}}

\newcommand{\Hf}{\mathfrak{H}}

\DeclareMathOperator{\EoM}{EoM}
\DeclareMathOperator{\fEoM}{fEoM}

\DeclareMathOperator{\Sol}{Sol} 
\DeclareMathOperator{\WF}{WF}
\newcommand{\pol}{\textnormal{pol}}

\newcommand{\Alg}{\mathsf{Alg}}
\newcommand{\grAlg}{\mathsf{grAlg}}

\newcommand{\Met}{\mathsf{Met}}
\newcommand{\HVB}{\mathsf{HVB}}

\newcommand{\Sympl}{\mathsf{Sympl}}
\newcommand{\GreenHyp}{\mathsf{GreenHyp}}
\newcommand{\fGreenHyp}{\mathsf{fGreenHyp}}
\newcommand{\GlobHypGreen}{\mathsf{GlobHypGreen}}

\newcommand{\fGlobHypGreen}{\mathsf{fGlobHypGreen}}

\newcommand{\Loc}{\mathsf{Loc}}

\newtheorem{thm}{Theorem}[section]
\newtheorem{lemma}[thm]{Lemma}
\newtheorem{cor}[thm]{Corollary}

\newtheorem{rem}[thm]{Remark}

\theoremstyle{definition}
\newtheorem{defn}[thm]{Definition}

\newdateformat{daymonthyear}{\THEDAY \, \monthname[\THEMONTH] \THEYEAR}

\newcommand{\Bf}{\mathscr{B}}

\newcommand{\Qf}{\mathscr{Q}}

\newcommand{\Sf}{\mathscr{S}}

\newcommand{\Xf}{\mathscr{X}}
\newcommand{\Yf}{\mathscr{Y}}
\newcommand{\Zf}{\mathscr{Z}}

\newcommand{\Xc}{\mathcal{X}}
\newcommand{\Zc}{\mathcal{Z}}

\newcommand{\DD}{\mathscr{D}}

\newcommand{\adj}[1]{\leftidx{^{t}}{#1}{}}

\newcommand{\hadj}[1]{\leftidx{^{\dagger}}{#1}{}}
\newcommand{\sadj}[1]{\leftidx{^{\star}}{#1}{}}

\newcommand{\dlangle}{\langle\!\langle}
\newcommand{\drangle}{\rangle\!\rangle}

\newcommand{\Ac}{\mathcal{A}}
\newcommand{\Cc}{\mathcal{C}}

\newcommand{\Nc}{\mathcal{N}}
\newcommand{\Pc}{\mathcal{P}}
\newcommand{\Rc}{\mathcal{R}}
\newcommand{\Uc}{\mathcal{U}}
\newcommand{\Vc}{\mathcal{V}}

\newcommand{\Had}{\mathrm{Had}}

\newcommand{\rfhgho}{RFHGHO}

\newcommand{\ogth}{{\mathfrak o}}
\newcommand{\tgth}{{\mathfrak t}}

\newcommand{\Mb}{{\boldsymbol{M}}}

\newcommand{\twopt}{\textnormal{2pt}}
\newcommand{\knl}{\textnormal{knl}}

\begin{document}

\title{Hadamard states for decomposable Green-hyperbolic operators}
\author{Christopher J. Fewster{\orcidicon{0000-0001-8915-5321}}\thanks{chris.fewster@york.ac.uk}\\[6pt]  
	\small Department of Mathematics, University of York, Heslington, York YO10 5DD, United Kingdom.\\[4pt]
	\small York Centre for Quantum Technologies, University of York, Heslington, York YO10 5DD, United Kingdom.
} 
 
\date{\daymonthyear\today}

\maketitle
 
\begin{abstract} 
Hadamard states were originally introduced for quantised Klein--Gordon fields and occupy a central position in the theory of quantum fields on curved spacetimes. Subsequently they have been developed for other linear theories, such as the Dirac, Proca and Maxwell fields, but the particular features of each require slightly different treatments. 

This paper gives a generalised definition of Hadamard states for linear bosonic and fermionic theories encompassing a range of theories that are described by Green-hyperbolic operators with `decomposable' Pauli--Jordan propagators, including theories whose bicharacteristic curves are not necessarily determined by the spacetime metric. The new definition reduces to previous definitions for normally hyperbolic and Dirac-type operators. We develop the theory of Hadamard states in detail, showing that our definition propagates under the equation of motion, and is also stable under pullbacks and suitable pushforwards. 
There is an equivalent formulation in terms of Hilbert space valued distributions, and the generalised Hadamard condition on $2$-point functions constrains the singular behaviour of all $n$-point functions. For locally covariant theories, the Hadamard states form a covariant state space. It is also shown how Hadamard states may be combined through tensor products or reduced by partial tracing while preserving the Hadamard property. As a particular application it is shown that state updates resulting from nonselective measurements preserve the Hadamard condition. 

The treatment we give was partly inspired by a recent work of Moretti, Murro and Volpe (MMV) [{Ann.\ H.\ Poincar\'{e} \textbf{24}, 3055--3111 (2023)}] on the neutral Proca field. Among the other applications, we revisit the neutral Proca field and prove a complete equivalence between the MMV definition of Hadamard states and an older work of Fewster and Pfenning [{J.\ Math.\ Phys.\
	\textbf{44}, 4480--4513 (2003)}]. 
\end{abstract}
 

\section{Introduction}
\label{sec:Intro}

A striking difference between quantum field theory (QFT) and quantum mechanics is that the same algebra of observables can have many inequivalent Hilbert space representations. An equally striking difference between QFT on curved spacetimes and QFT in Minkowski spacetime is that 
there is generally no physically preferred state that can be used to identify a particularly natural representation -- a long-standing folk-theorem that can be formulated as a rigorous no-go result~\cite{FewVer:dynloc_theory,Fewster_artofthestate:2018}. 
 
Lacking a single preferred state, attention is directed instead to the question of what class of states can be considered physically relevant. For linear Klein--Gordon fields in general globally hyperbolic spacetimes, the consensus answer is the class of Hadamard states, which were first identified in attempts to renormalise the stress-energy tensor~\cite{AdlerLiebermanNg:1977}, with key properties established in~\cite{FullingSweenyWald:1978,FullingNarcowichWald}, and a full formal definition eventually given in~\cite{KayWald-PhysRep} in terms of the Hadamard expansion.   
In outline, the two-point function $W(x,y)=\langle \phi(x)\phi(y)\rangle$ of a Hadamard state in four spacetime dimensions should take the form
\begin{equation}\label{eq:Had_outline}
	W(x,y) = \frac{U(x,y)}{4\pi^2\sigma_+(x,y)} + V(x,y)\log (\sigma_+(x,y)/\ell^2) + \text{smooth};
\end{equation}
see~\cite{DecaniniFolacci:2008} for general spacetime dimensions. Here, $U$ and $V$ are specific geometrically determined smooth functions, $\ell>0$ is an arbitrary length scale and $\sigma_+$ is a particular distributional regularisation of the 
signed squared geodesic separation between points $x$ and $y$, which can be defined within geodesically convex normal neighbourhoods.
A subtle point, which was missed for many years, is whether there is any dependence on the choice of these neighbourhoods; this was addressed and resolved in~\cite{Moretti:2021} with a slightly tightened definition.

One sees from~\eqref{eq:Had_outline} that all Hadamard states have the same singular structure and the specificity of any individual state is encapsulated within the smooth terms in~\eqref{eq:Had_outline}. Moreover, the geometric nature of the singular terms permits covariant regularisation of the stress-energy tensor and related observables by point-splitting.

In a major development thirty years ago, Radzikowski showed how~\eqref{eq:Had_outline} could be replaced by a wavefront set condition on the two-point function~\cite{Radzikowski_ulocal1996}. (It transpires that this reformulation is equivalent to the corrected definition of~\cite{Moretti:2021}). We recall that the wavefront set~\cite{Hoer_FIOi:1971,Hormander1} describes both positional and directional information about distributional singularities, and  plays an essential role in the general theory of propagation of singularities~\cite{DuiHoer_FIOii:1972}. Radzikowski replaced~\eqref{eq:Had_outline} by the condition
\begin{equation}\label{eq:Had_WF}
	\WF(W) = \{(x,k;x',-k')\in T^*(M\times M): (x,k)\sim (x',k'),~k\in \Nc^+ \},
\end{equation}
where $\Nc^+$ is the bundle of future-directed nonzero null covectors, and the relation
$(x,k)\sim (x',k')$ holds if and only if the two points lie on a common bicharacteristic strip in the cotangent bundle $T^*M$ -- the precise definition is recalled in Section~\ref{sec:Hadamard}.

The microlocal form~\eqref{eq:Had_WF} of the Hadamard condition initiated, and continues to stimulate, a wave of remarkable progress in the theory of QFT in curved spacetimes. Among other things, it has provided the basis for a rigorous perturbative construction of interacting fields in curved spacetime~\cite{BrFr2000,Ho&Wa01,Hollands:2008,FreRej_BVqft:2012,Rejzner_book,Duetsch_book}.
It was also a spur to reconsider the locality and covariance of QFT across different spacetimes, resulting in an understanding that QFTs can be expressed in terms of functors from categories of spacetimes to suitable target categories~\cite{BrFrVe03,FewVer:dynloc_theory,FewVerch_aqftincst:2015}. 
Other physical applications include general results
on quantum energy inequalities~\cite{Fews00}, the strong cosmic censorship conjecture~\cite{HollandsWaldZahn:2020,JuarezAubry:2024}, the chronology protection conjecture~\cite{KayRadWald:1997}, studies of states of low energy~\cite{Olbermann}, 
and the microlocal formulation of adiabatic states~\cite{JunkerSchrohe}. In addition
the existence and Hadamard nature of various states of interest on specific (particularly black hole) spacetimes~\cite{DappiaggiMorettiPinamonti:2011,Sanders:2015,Klein:2023} has been established, and the asymptotic symmetries of asymptotically flat spacetimes have provided another route to construct Hadamard states~\cite{Moretti:2008},
as have direct series constructions~\cite{Lewandowski:2022} and pseudodifferential techniques~\cite{GerardOulghaziWrochna:2017,VasyWrochna:2018,IslamStrohmaier:2020,FewsterStrohmaier:2025}.  Any Hadamard state determines a Feynman propagator in a simple way. However, the reverse construction of a Hadamard state from a Feynman propagator involves a subtlety that is  resolved in~\cite{FewsterStrohmaier:2025}.

The condition~\eqref{eq:Had_WF} can be simplified to the requirement that
\begin{equation}\label{eq:Had_SV}
	\WF(W)\subset \Nc^+\times \Nc^-,
\end{equation}
where $\Nc^-$ is the bundle of past-pointing nonzero null covectors --
an observation made in~\cite{SahlmannVerch:2000RMP}. Even further, it was shown in~\cite{StVeWo:2002} that a vector state $\psi$ in a Hilbert space representation of the theory
is Hadamard if and only if
\begin{equation}\label{eq:Had_SVW}
	\WF(\Phi(\cdot)\psi)\subset \Nc^-,
\end{equation}
where $\Phi$ is the scalar quantum field in the representation and 
$f\mapsto \Phi(f)\psi$ is understood as a vector-valued distribution taking values in the Hilbert space. This boils down the Hadamard condition to the idea that -- in the high frequency limit -- physical two-point functions are positive-frequency in their first argument and negative frequency in their second, and fields create negative frequency excitations of physical states.

In fact,~\eqref{eq:Had_SV} implies much more detailed information both on $\WF(W)$ and on the higher $n$-point functions of $\omega$. A result of Sanders~\cite{Sanders:2010} shows that $\omega$ is Hadamard if and only if it satisfies the microlocal spectrum condition, which constrains the wavefront sets of all $n$-point functions and delineates the class of states relevant to the perturbative construction of interacting theories~\cite{Rejzner_book}. General reviews on the microlocal techniques in QFT on curved spacetimes and related matters can be found in~\cite{Strohmaier:2009, KhavkineMoretti:2015, Fewster_artofthestate:2018,EncycMP_QFTCST_Kay:2025}. 
In particular, see~\cite{Fewster_artofthestate:2018} for commentary on why the Hadamard condition fails for some apparently natural constructions~\cite{AAS,FewsterVerch-SJ} and discussion of results showing that the Hadamard condition can be motivated by the demand that certain Wick products have finite fluctuations~\cite{FewsterVerch-NecHad}.

The purpose of this paper is to give a detailed, unified and generalised account of the Hadamard condition for a wide class of bosonic and fermionic theories described by suitable Green hyperbolic operators. To explain the motivation for this work, we first describe the existing theory of Hadamard states for more general fields before turning to Green-hyperbolic operators.  

The Hadamard condition has been formulated using microlocal techniques for other theories of interest including the Dirac~\cite{Hollands:2001,Sanders_dirac:2010}, Proca (massive spin-$1$)~\cite{Few&Pfen03}, Maxwell equations~\cite{Few&Pfen03} and linearised gravity~\cite{Hunt:2012thesis}. A general class of normally hyperbolic 
and related Dirac theories was treated in~\cite{SahlmannVerch:2000RMP}; here, we recall that a
normally hyperbolic operator on sections of a vector bundle is a second order partial differential operator, with principal part $g^{\mu\nu}\partial_\mu \partial_\nu\delta^{A}_{\phantom{A}B}$ relative to local coordinates and frame components. See the sources just cited and references therein for the corresponding definitions via Hadamard series. 

These microlocal treatments are, in most cases, somewhat indirect.\footnote{An exception is the treatment of the Dirac equation in~\cite{Sanders_dirac:2010} where the analogue of~\eqref{eq:Had_SVW} is employed.} They proceed by relating
the theory at hand to an auxiliary normally hyperbolic equation, such as the spinorial Klein--Gordon equation, the $1$-form Klein--Gordon equation, the $1$-form wave equation
or the de Donder gauge-fixed equation for linearised gravity. One says that a state of the QFT is Hadamard if its two-point function can be related to a distributional bisolution to the auxiliary equation that is of `Hadamard form', in particular satisfying a microlocal condition of the form~\eqref{eq:Had_WF}. For instance, the mass $m>0$ Proca theory was studied by Fewster and Pfenning (FP) in~\cite{Few&Pfen03}. According to FP, a Proca state is Hadamard if there exists a Hadamard-form bisolution $H$ to the $1$-form Klein--Gordon equation so that the Proca $2$-point function is $W =(1\otimes D)H$, where $D$ is given in terms of the exterior derivative $\dd$ and codifferential $\delta$ by $D=1-m^{-2}\dd\delta$.

Recently, Moretti, Murro and Volpe (MMV)~\cite{MorettiMurroVolpe:2023} have 
proposed a direct definition
of the Hadamard condition in the form~\eqref{eq:Had_WF} for the Proca two-point function, rather than the indirect FP definition. 
Unfortunately, it transpires that~\cite{MorettiMurroVolpe:2023} 
contains a gap that was identified and fixed in a companion paper to this one~\cite{Fewster:2025b}. Nevertheless,~\cite{MorettiMurroVolpe:2023} has provided some of the inspiration for this paper, which will give direct definitions of the Hadamard condition for a wide range of operators. We will also return to the relationship between the FP and MMV definitions of Hadamard states for the Proca theory and demonstrate a full equivalence between them.

Green-hyperbolic operators were introduced by B\"ar~\cite{Baer:2015}, who isolated the essential algebraic structure common to many partial differential operators arising in physics, and showed how much of the analytic theory can be extracted from the algebraic data. The idea is to study those partial differential operators that (together with their formal adjoints) admit advanced and retarded Green operators. Green-hyperbolic operators with suitable additional properties can be quantised as free quantum fields of Bose or Fermi type~\cite{BaerGinoux:2012} as we will describe in Section~\ref{sec:GreenHyperbolic}.

The advanced and retarded Green operators of a Green-hyperbolic operator are required to map test functions $f$ (or, in general, test sections of a vector bundle) to solutions of the 
inhomogeneous equation that are supported in the causal past or future of the support of $f$. This is certainly 
satisfied by normally hyperbolic operators; however, Green-hyperbolicity is much broader, and encompasses the Dirac operator (which is first order), the Proca operator (which has second order terms beyond those coming from the metric) and many others. Further examples include
operators that are normally hyperbolic with respect to a `slow metric', whose cones of future/past-directed causal vectors are subsets of the cone of spacetime future/past-directed causal vectors. Dually, the cone of causal covectors of the slow metric is a superset of the cone of causal covectors for the spacetime metric. This type of behaviour can also be seen in the theory of birefringence for electrodynamics in media (see~\cite{FewsterPfiefferSiemssen:2018,Strohmaier:2021} and references therein). Similarly, suitable convex combinations of normally hyperbolic operators with respect to different metrics are Green-hyperbolic but not normally hyperbolic~\cite{MMVparacausal:2023}.
Meanwhile, Lorentz-violating quantum field theories also result in effective null cones differing from  those of the spacetime geometry~\cite{ArderucioCostaBonderJuarezAubry:2023}. The idea can be pushed further -- see, for example~\cite{Fewster:2023} for the existence of
some nonlocal operators that obey a modified version of Green-hyperbolicity.

Green-hyperbolic operators can be combined in various ways to construct other Green-hyperbolic operators~\cite{Baer:2015}. For example, the direct sum of two Green-hyperbolic operators is Green-hyperbolic. 
This simple fact makes the class of Green-hyperbolic operators a natural source of examples for the theory of measurement for QFT in flat and curved spacetimes introduced in~\cite{FewVer_QFLM:2018} -- see~\cite{EncycMP_Measurement_in_QFT_FewsterVerch2025} for a short presentation and a survey of more general questions concerning measurement for QFT. In this framework, one measures a `system QFT' by coupling it to another `probe QFT'. 
The dynamics of the coupling can be described by a scattering map on the theory of the `uncoupled combination' of system and probe. If both the system and probe are obtained from Green-hyperbolic operators $P$ and $Q$, then their uncoupled combination is obtained from the direct sum $P\oplus Q$. The measurement framework, further developed in~\cite{BostelmannFewsterRuep:2020,FewsterJubbRuep:2022}, applies to general QFTs described within algebraic QFT. It provides a description of measurement schemes for system observables and also the updates to system states resulting from selective and nonselective measurements, while avoiding pathologies. Although QFTs based on Green-hyperbolic operators have provided a good source of examples, it has not yet been shown how the state updates interact with the theory of Hadamard states. One reason for this is that the theory of Hadamard states has not been developed for Green-hyperbolic operators, beyond the normally hyperbolic class and the examples related to that in various ways. Given that the theory of Hadamard states is tightly linked to aspects of the PDE theory, like the principal symbol of the field equation and the related propagation of singularities, while the theory of Green-hyperbolicity deliberately avoids these concepts, it is not immediately clear how to proceed. In fact, the characteristic set of the field equation cannot
have any bearing on the matter, because the Proca field operator turns out to be everywhere characteristic!

The present paper has two goals. First, we will present a generalised definition of what a Hadamard state should be for a class of Green-hyperbolic operators going well beyond the normally hyperbolic class, for both bosonic and fermionic statistics. This not only unifies the discussion of Hadamard states for standard linear theories, but also opens up the possibility of
describing more general measurement interactions involving combinations of theories in our class. Second, we illustrate the advantages of this new framework by presenting applications to measurement theory and the Proca field. 

Our approach proceeds by 
analogy with the standard theory of the Klein--Gordon field, in which the wavefront set of the difference of advanced and retarded Green operators (the Pauli--Jordan propagator) has a wavefront set contained in 
$(\Nc^+\times\Nc^-)\cup (\Nc^-\times\Nc^+)$. Extrapolating, we will focus on Green-hyperbolic operators $P$ whose  advanced-minus-retarded Green operator $E_P$ has a wavefront set obeying
\begin{equation}\label{eq:decomp}
	\WF(E_P)\subset (\Vc^+\times\Vc^-)\cup (\Vc^-\times\Vc^+)
\end{equation}
for some conic set $\Vc^+$ in the spacetime cotangent bundle $T^*M$ that does not intersect its opposite $\Vc^-=-\Vc^+$ (see Section~\ref{sec:Hadamard} for the precise definition). No relation is assumed between $\Vc^\pm$ and the cones $\Nc^\pm$ of null (or causal) covectors. 
Such operators $P$ will be described as $\Vc^\pm$-decomposable. States with two-point functions obeying the obvious modification of~\eqref{eq:Had_SV}, namely
\begin{equation}
	\WF(W)\subset \Vc^+\times \Vc^-,
\end{equation}
will be described as $\Vc^+$-Hadamard; one can equally proceed from the corresponding modification of~\eqref{eq:Had_SVW}. Evidently, \eqref{eq:Had_SV} for Klein--Gordon fields corresponds to the $\Nc^+$-Hadamard condition in this terminology. 

Starting from these definitions, we will show 
how various standard properties of Hadamard states have analogues for our generalised definition.
In particular, differences of Hadamard two-point functions are smooth, and the Hadamard property propagates under the equation of motion (Theorem~\ref{thm:Hadamard_props}); there is an equivalent Hilbert space formulation mirroring~\eqref{eq:Had_SVW} (Theorem~\ref{thm:HilbertHadamard}) from which one can see that the existence of a single Hadamard state implies the existence of many more (Corollary~\ref{cor:manyHadamard}). We also show that
the Hadamard form is preserved under suitable pullbacks and pushforwards, and, for similar reasons, is stable under scattering processes (Theorem~\ref{thm:Hadamard_scattering}). For locally covariant theories~\cite{BrFrVe03,FewVerch_aqftincst:2015} described by Green-hyperbolic operators, this implies that the Hadamard states form a covariant state space (Corollary~\ref{cor:covariant_state_space}). Moreover, we show in Theorem~\ref{thm:Hadamard_tools} that the Hadamard condition interacts well with direct sums and partial traces.  Our final general results are given in Theorem~\ref{thm:truncated} and Corollary~\ref{cor:truncated}, where we adapt a result of Sanders~\cite{Sanders:2010} to show how the Hadamard condition on $2$-point functions actually constrains the singular structure of all $n$-point functions. In a precise sense, all Hadamard states are microlocally quasifree. 

The general theory is illustrated by various examples and applications. First, recent results of Islam and Strohmaier~\cite{IslamStrohmaier:2020}, further developed in~\cite{FewsterStrohmaier:2025}, can be used to show that all normally hyperbolic operators on bundles with positive definite hermitian metrics admit $\Nc^+$-Hadamard states for their bosonic QFTs (Theorems~\ref{thm:Hadamard_norm_hyp}) and similarly that
all Dirac-type operators of definite type~\cite{BaerGinoux:2012} admit
$\Nc^+$-Hadamard states for their fermionic QFTs (Theorem~\ref{thm:Hadamard_Dirac_type}).

Second, our results that the Hadamard condition is preserved by tensor products, scattering and partial traces combine to prove that state updates resulting from nonselective measurements yield Hadamard states, if the system and probe states are initially Hadamard -- see Theorem~\ref{thm:Hadamard_update}. This is a significant contribution to the measurement framework because it shows that these measurement updates do not result in pathological states.
Other applications to the theory of measurements will appear elsewhere.

As a further application, we consider the Hadamard states of the Proca field in section~\ref{sec:Proca}. There, it is shown that the 
states considered by MMV (Moretti, Murro and Volpe~\cite{MorettiMurroVolpe:2023}) are precisely the quasifree $\Nc^+$-Hadamard states in the sense of this paper. Although the gap in MMV is fixed in~\cite{Fewster:2025b}, our general results provide an independent proof of the existence of such states on all globally hyperbolic spacetimes (Theorem~\ref{thm:ProcaHadamard}). More significantly, we are able to establish a full equivalence between the $\Nc^+$-Hadamard condition and the FP definition of Hadamard (Theorem~\ref{thm:FPequivalence}). This shows that the MMV definition (setting aside the inessential restriction to quasifree states in MMV) is equivalent to the FP-definition, something that was only partly shown in MMV. What is necessary here is to construct a Hadamard form bisolution for the $1$-form Klein--Gordon operator corresponding to a given Proca $2$-point function. It turns out that there is a simple formula~\eqref{eq:H_from_W} that achieves this, by adding to the Proca $2$-point function a term obtained from an auxiliary scalar Hadamard $2$-point function but with the `wrong' sign. This can be interpreted as the contribution from the ghost associated with the Proca constraint. 
 
The paper is structured as follows. After necessary preliminaries in Section~\ref{sec:preliminaries}, we review in some detail the definition of Green-hyperbolic operators on sections of vector bundles over globally hyperbolic spacetimes, focussing on the subclass that are real and formally hermitian, which possess quantisations as hermitian bosonic quantum fields and (with additional structure) as fermionic fields. Some material here is new, in particular the categorical structures on these operators and the delineation of fermionic real formally hermitian Green-hyperbolic operators. Special cases include the normally hyperbolic operators and Dirac-type operators in the bosonic and fermionic cases respectively.
In Section~\ref{sec:Hadamard} we present the generalised Hadamard condition and establish its properties, as described above, before moving to the application in measurement theory in Section~\ref{sec:FVmeasurements}. As a specific example, we describe the Hadamard states of the neutral Proca field in Section~\ref{sec:Proca}. After a summary in Section~\ref{sec:Conc}, two appendices provide a proof deferred from the main text and a brief account of how the theory of Hadamard states developed for real bose fields extends to complex fields.
 
\section{Preliminaries}\label{sec:preliminaries}

\paragraph{General conventions} The symbol $\subset$ allows for equality. Alternatives labelled with vertically stacked symbols, such as $J^\pm(S)$, are ordered from top to bottom.

\paragraph{Globally hyperbolic spacetimes}

A globally hyperbolic spacetime consists of a smooth paracompact $n$-dimensional manifold $M$ with at most finitely many components, equipped with a smooth Lorentzian metric $g$ on $M$ of signature ${+}{-}\cdots{-}$ and a time-orientation $\tgth$, so that there are no closed $g$-causal curves in $M$ and 
$J^+(K)\cap J^-(K)$ is compact for every compact $K\subset M$, where $J^\pm(S)$ are the causal future/past of a set $S\subset M$. The time-orientation $\tgth$ is expressed as a nontrivial component of the set of nowhere zero $g$-timelike $1$-form fields on $M$. We will study globally hyperbolic spacetimes that are also oriented, that is, equipped with a nontrivial component $\ogth$ of the set of nowhere zero $n$-forms.
The globally hyperbolic spacetimes of dimension $n$ form a category $\Loc$ with objects $\Mb$ and morphisms
$\psi:(M,g,\ogth,\tgth)\to (M',g',\ogth',\tgth')$ given by smooth isometric embeddings $\psi:M\to M'$ with causally convex image $\psi(M)$, and so that $\psi^*\ogth'=\ogth$, $\psi^*\tgth'=\tgth$~\cite{BrFrVe03}.

A closed subset $S$ of globally hyperbolic $M$ is \emph{future/past-compact} if
$S\cap J^\pm(x)$ is compact for all $x\in M$, \emph{strictly past/future-compact} (resp., \emph{spatially compact}) if $S\subset J^\pm(K)$ (resp., $S\subset J(K)$) for some compact $K$, and \emph{temporally compact} if $S$ is both future- and past-compact. The spaces of smooth sections of a bundle $B$ over $M$ that have these types of support will be denoted $\Gamma_\bullet^\infty(B)$, with $\bullet=\fc/\pc/\spc/\sfc/\sc/\tc$ respectively. The smooth sections with compact support are denoted $\Gamma_0^\infty(B)$, those supported in closed set $S$ by $\Gamma_S^\infty(B)$, and those with unrestricted support by $\Gamma^\infty(B)$. 

An open causally convex subset of a globally hyperbolic spacetime with finitely many connected components is called a \emph{region}. An \emph{open (resp., closed) Cauchy slab} in a globally hyperbolic spacetime $M$ is any set of the form $I^+(\Sigma^-)\cap I^-(\Sigma^+)$ (resp., $J^+(\Sigma^-)\cap J^-(\Sigma^+)$), where $\Sigma^\pm$ are smooth spacelike Cauchy surfaces with $\Sigma^\pm\subset I^\pm(\Sigma^\mp)$ and $I^\pm(S)$ denote the chronological future/past of $S$. An open Cauchy slab is a region; the closure of an open Cauchy slab is a closed Cauchy slab; closed Cauchy slabs are temporally compact, and all Cauchy slabs are causally convex. 

\paragraph{Densities} A density of order $\alpha\in\RR$ (or $\alpha$-density) on an $n$-dimensional vector space $V$ is an $\alpha$-homogeneous map $\rho:\bigwedge^n V\setminus\{0\}\to\CC$,
where $\wedge$ denotes the exterior product (see e.g.,~\cite{Duistermaat:2011}). Thus $\rho(\lambda v)=|\lambda|^\alpha \rho(v)$ for
$\lambda\in\RR\setminus\{0\}$, $v\in \bigwedge^n V\setminus\{0\}$. We write $\Omega^\alpha(V)$ for the 
$1$-dimensional complex vector space of all $\alpha$-densities on $V$. Given a smooth $n$-dimensional manifold $X$, $\Omega_X^\alpha$ is the complex line-bundle over $X$ with fibre $\Omega^\alpha(T_xX)$ at $x\in X$. 
If the base manifold is clear from context, we write $\Omega^\alpha$; we also write $\Omega^1$ as $\Omega$.
An example that will appear frequently is the volume density $\rho\in \Gamma^\infty(\Omega)$ induced by the spacetime metric $g$ on Lorentzian spacetime $M$, with coordinate expression $(-\det g_{\bullet\bullet})^{1/2}$. 

\paragraph{Distributions and kernels} Let $B$ be a finite-rank complex vector bundle over manifold $X$ with dual bundle $B^*$ and fibrewise duality pairing $\dlangle\cdot,\cdot\drangle_{B,x}$ (the subscripts will be omitted when they are clear from context). The same
notation can be used without ambiguity for a $\Omega^{\alpha+\beta}$-valued pairing of 
$B^*\otimes\Omega^\alpha$ and $B\otimes\Omega^\beta$ for $\alpha,\beta\in\RR$. The space of 
distributional sections of $B$ is the topological dual
$\DD'(B)=\Gamma_0^\infty(B^*\otimes\Omega)'$ with the weak-$*$ topology and there is a canonical continuous embedding 
$\iota:\Gamma^\infty(B)\to \DD'(B)$ given by
\begin{equation}
	(\iota F)(h) = \int_X \dlangle h,F\drangle  , \qquad 
	F\in\Gamma^\infty(B),~h\in \Gamma_0^\infty(B^*\otimes\Omega).
\end{equation} 
For bundles of the form $B\otimes\Omega^{1/2}$ (i.e., $B$-valued half-densities)
there is a more symmetric formula 
\begin{equation}
	\DD'(B\otimes\Omega^{1/2}) = \Gamma_0^\infty(B^*\otimes\Omega^{1/2})',
\end{equation}
(noting that $(B\otimes\Omega^{1/2})^*=B^*\otimes \Omega^{-1/2}$).
It is often  convenient to keep an explicit half-density bundle factor. For example, if $B_X$ and $B_Y$ are bundles over manifolds $X$ and $Y$, then any continuous linear map
$T:\Gamma_0^\infty(B_X\otimes\Omega_X^{1/2})\to 
\DD'(B_Y\otimes\Omega_Y^{1/2})$ has a distributional kernel $T^\knl\in\DD'((B_Y\boxtimes B_X^*)\otimes\Omega_{Y\times X}^{1/2})$, 
\begin{equation}
	T^\knl(v\otimes u) = (Tu)(v), \qquad u\in \Gamma_0^\infty(B_X\otimes\Omega_X^{1/2}),~ v\in \Gamma_0^\infty(B_Y^*\otimes\Omega_Y^{1/2}).
\end{equation} 
Here, $B_Y\boxtimes B_X^*\to Y\times X$ is the external tensor product.

In some situations there is little ambiguity in writing $T^\knl$ as $T$. Another slight abuse is that we also consider kernel
distributions for continuous linear maps $T:\Gamma_0^\infty(B_X\otimes\Omega_X^{1/2})\to 
\Gamma^\infty(B_Y\otimes\Omega_Y^{1/2})$, writing $T^\knl\in\DD'((B_Y\boxtimes B_X^*)\otimes\Omega^{1/2}_{Y\times X})$ where we would more properly write $(\iota\circ T)^\knl$. In such cases one has the formula
\begin{equation}
	T^\knl(v\otimes u) = \int_X\dlangle v,Tu\drangle, \qquad u\in \Gamma_0^\infty(B_X\otimes\Omega_X^{1/2}),~ v\in \Gamma_0^\infty(B_Y^*\otimes\Omega_Y^{1/2}).
\end{equation} 
It is frequently useful to identify $(B_Y\boxtimes B_X^*)_{y,x}$ with the space of linear maps $\Hom((B_X)_x,(B_Y)_y)$, so that $\upsilon\otimes\xi$ is identified with the map $w\mapsto \xi(w)\upsilon$.

\section{Green-hyperbolic operators}\label{sec:GreenHyperbolic}

\subsection{Hermitian vector bundles over globally hyperbolic spacetimes}
\label{sec:HVBs}

In this paper, a hermitian vector bundle $(B,\Mb,(\cdot,\cdot),\Cc)$ over $\Mb\in\Loc$ will consist  
of a smooth finite-rank complex vector bundle $B\stackrel{\pi}{\to}{M}$, equipped with a smooth (nondegenerate, but not necessarily positive definite) hermitian bundle metric $(\cdot,\cdot)_x$ in each fibre $\pi^{-1}(x)$, antilinear in its first argument, 
and an antiunitary involutive base-point preserving bundle morphism $\Cc$ so that 
$(\Cc v,\Cc w)_x = \overline{(v,w)_x}$.\footnote{In section~\ref{sec:Dirac-type} and Appendix~\ref{sec:complex} we will use bundles of the type described but which do not necessarily possess a complex conjugation. It would be very reasonable to call these `hermitian vector bundles' and to find a different terminology for those with a complex conjugation but we have opted for simplicty of usage.}  The corresponding antilinear involutions $\Cc:\Gamma_\bullet^\infty(B)\to\Gamma_\bullet^\infty(B)$ are continuous for all support types.
There is a sesquilinear pairing on smooth sections of $B$ with compactly intersecting supports given by
\begin{equation}\label{eq:pairing}
	( {f}_1, {f}_2 ) = \int_M ( {f}_1(x), {f}_2(x))_x \dvol_g(x).
\end{equation} 
There is a bilinear bundle metric $\langle\cdot,\cdot\rangle$ so that $\langle v,w \rangle_x = ( \Cc v,w)_x$, and a corresponding bilinear pairing 
$\langle {f}_1, {f}_2 \rangle$ on sections of $B$ with compactly intersecting supports
obtained by replacing $(\cdot,\cdot)_x$ by $\langle\cdot,\cdot\rangle_x$ in~\eqref{eq:pairing}. The duality pairing $\dlangle\cdot,\cdot\drangle_x$ of $B_x^*$ with $B_x$ gives a bilinear pairing
\begin{equation}
	\dlangle F, f\drangle = \int_M \dlangle F(x), f(x)\drangle_x \dvol_g(x)
\end{equation} 
of $F\in \Gamma^\infty(B^*)$, $f\in \Gamma^\infty(B)$ with compactly intersecting supports.

We define a category $\HVB$ of Hermitian vector bundles in which a morphism  
$\beta:(B,\Mb,(\cdot,\cdot),\Cc)\to (B',\Mb',(\cdot,\cdot)',\Cc')$ is a smooth vector bundle map $\beta:B\to B'$ that covers a $\Loc$-morphism $\psi$ between the base manifolds $\Mb$ and $\Mb'$, and which is a fibrewise isometry 
such that $\beta\circ \Cc = \Cc'\circ\beta$. Associated with $\beta$ there is a linear, continuous push-forward $\beta_*:\Gamma_0^\infty(B)\to\Gamma_0^\infty(B')$ so that $\beta_*f$ is the extension by zero of $\beta\circ f\circ \psi^{-1}$. 
 
\paragraph{Adjoints, transposes and duals}
Suppose $S:\Gamma^\infty(B_1)\to\Gamma^\infty(B_2)$ is a linear map, where
 $B_j=(B_j,\Mb_j,(\cdot,\cdot)_j,\Cc_j)$ ($j=1,2$) are hermitian vector bundles. Then the 
\emph{formal hermitian adjoint} $\hadj{S}$ (if it exists) is the unique linear map
$\hadj{S}:\Gamma^\infty(B_2)\to\Gamma^\infty(B_1)$ obeying
\begin{equation} 
		({f}_1, S{f}_2)_2 = (\hadj{S}{f}_1, {f}_2)_1 
\end{equation}
for all smooth sections ${f}_j\in \Gamma^\infty(B)$ so that $f_1$ and $Sf_2$ have compactly intersecting supports (implicitly assuming that $\hadj{S}f_1$ and $f_2$ have compactly intersecting supports). The same definition is used if $S:\Gamma_0^\infty(B_1)\to\Gamma_0^\infty(B_2)$, again seeking $\hadj{S}:\Gamma^\infty(B_2)\to\Gamma^\infty(B_1)$.
Every partial differential operator has a formal hermitian adjoint that is also a partial differential operator. The  push-forward $\beta_*:\Gamma_0^\infty(B)\to\Gamma_0^\infty(B')$ of a $\HVB$-morphism $\beta:B\to B'$  covering $\Loc$-morphism $\psi$ has a formal hermitian adjoint $(\hadj{\beta}_* h)(x)=\beta_x^\dagger h(\psi(x))$, where
$\beta_x^\dagger$ is the fibre-metric adjoint of $\beta_x:B_x\to B'_{\psi(x)}$. We have
$\hadj{\beta}_*\circ \beta_*=\id$ on $\Gamma_0^\infty(B)$, and, for any compact $K\subset \Mb$,
\begin{equation}\label{eq:betastarsupport}
	\beta_*:\Gamma_K^\infty(B)\subset \Gamma_{\psi(K)}^\infty(B'), \qquad 
	\hadj{\beta}_*\Gamma_{J^\pm(\psi(K))}^\infty(B') \subset  \Gamma_{J^\pm(K)}^\infty(B),
\end{equation}
where the second inclusion is due to causal convexity of $\psi(\Mb)$. 

The \emph{formal transpose} is 
$\adj{S}=\Cc_1(\hadj{S})\Cc_2$ and obeys
\begin{equation} 
	\langle {f}_1, S{f}_2\rangle_2 = \langle\adj{S}{f}_1, {f}_2\rangle_1 
\end{equation}
under the same conditions on $f_j$. We also write $\overline{S}:=\Cc_2 S \Cc_1$ 
and say that $S$ is \emph{real} if $S=\overline{S}$.
In the case $B_2=B_1$, $S$ is said to be \emph{formally hermitian} (resp., \emph{formally symmetric}) if $S=\hadj{S}$ (resp., $S=\adj{S}$). For example, the map $\beta_*$ defined by $\HVB$-morphism $\beta:B\to B'$ is real, and also satisfies $\adj{\beta}_*\circ \beta_*=\id$ on $\Gamma_0^\infty(B)$.  Real, formally symmetric operators are formally hermitian.

Finally, if it exists, the \emph{formal dual} $\sadj{S}$ is the linear map
$\sadj{S}:\Gamma^\infty(B_2^*)\to\Gamma^\infty(B_1^*)$ uniquely determined by 
\begin{equation} \dlangle \sadj{S}{F},{f}\drangle_1    = 
	  \dlangle {F},S{f}\drangle_2  
\end{equation}
for all ${F}\in\Gamma^\infty(B_2^*)$ and  ${f}\in\Gamma^\infty(B_1)$ so that
$F$ and $Sf$ (and, implicitly, $\sadj{S}F$ and $f$)
have compactly intersecting supports.  

Note that our terminology and notation for formal adjoint and formal transpose differ from e.g.,~\cite{Baer:2015} and other sources, but seem more natural in the current context.

\paragraph{Musical notation} Let $B$ be a hermitian vector bundle over $\Mb$. By nondegeneracy of the bundle metric, there is a fibrewise linear isomorphism $\sharp:B\to B^*$ with inverse $\flat:B^*\to B$ so that $\dlangle f^\sharp,h\drangle= \langle f,h\rangle$. In particular,  $\sadj{S} f^\sharp = (\adj{S}f)^\sharp$. The dual hermitian vector bundle to $(B,\Mb,(\cdot,\cdot),\Cc)$ is $(B^*,\Mb,(\flat\cdot,\flat\cdot),\sharp\circ \Cc\circ\flat)$. 
Density-weighted versions of $\sharp$ and $\flat$ are introduced so that, for $f\in \Gamma^\infty(B)$, $f^{\sharp,\alpha}\in\Gamma^\infty(B^*\otimes \Omega^\alpha)$ 
is $f^{\sharp,\alpha}=\rho^\alpha f^\sharp$, where $\rho$ is the metric volume density of $\Mb$; the inverse to $f\mapsto f^{\sharp,\alpha}$ is the map $h\mapsto h^{\flat,\alpha}=(\rho^{-\alpha}h)^\flat$.
For a linear map $T:\Gamma^\infty(B_1)\to \Gamma^\infty(B_2)$, the notation $T^{\sharp,\alpha}$ denotes the linear map 
from $\Gamma^\infty(B_1^*\otimes\Omega_1^\alpha)$ to $\Gamma^\infty(B_2^*\otimes\Omega_2^\alpha)$
given by $T^{\sharp,\alpha}h=(T h^{\flat,\alpha})^{\sharp,\alpha}$, i.e.,
\begin{equation}
T^{\sharp,\alpha}=\rho_2^\alpha \circ \sharp \circ T\circ \flat \circ \rho_1^{-\alpha}.
\end{equation}
We will occasionally write $T^{\natural,\alpha}=\rho_2^{\alpha}\circ T\circ\rho_1^{-\alpha}:\Gamma^\infty(B_1\otimes\Omega_1^\alpha)\to\Gamma^\infty(B_2\otimes\Omega_2^\alpha)$.

\subsection{(Semi)-Green-hyperbolic operators and \rfhgho's}\label{sec:sghos}

Let $P:\Gamma^\infty(B)\to\Gamma^\infty(B)$ be a partial differential operator where $(B,\Mb,(\cdot,\cdot),\Cc)$ is a hermitian vector bundle over $\Mb\in\Loc$. If there are linear maps $E_P^\pm:\Gamma_0^\infty(B)\to \Gamma^\infty(B)$ such that the conditions
\begin{enumerate}[G1]
	\item $PE_P^\pm {f}={f}$ 
	\item $E_P^\pm P {f}={f}$
	\item $\supp E_P^\pm {f}\subset J^\pm (\supp {f})$
\end{enumerate} 
hold for all ${f}\in \Gamma_0^\infty(B)$, then we say that $P$ is \emph{semi-Green-hyperbolic} and, as usual, call $E^\pm_P$ the advanced $(-)$ and retarded $(+)$ Green operators for $P$. If $P$ and one (and hence all) of $\adj{P}$, $\hadj{P}$ or $\sadj{P}$ are semi-Green-hyperbolic, then $P$ is \emph{Green-hyperbolic}~\cite{Baer:2015}.   

Two useful facts are that
$\overline{E}_P^\pm =E_{\overline{P}}^\pm$, and that $\adj{E}_P^\pm = E_{\adj{P}}^\mp$.
Therefore, if $P$ is a real, formally hermitian Green-hyperbolic operator (\rfhgho), one has
$\overline{E}_P^\pm=E_{P}^\pm= \adj{(E_{P}^\mp)}$, 
whereas if $P$ is formally hermitian and Green-hyperbolic,
$E_{P}^\pm= \hadj{(E_{P}^\mp)}$. 
We also define the advanced-minus-retarded operator $E_P:\Gamma_0^\infty(B)\to\Gamma^\infty(B)$ $E_P=E_P^--E_P^+$,  
whose range is contained in $\Gamma^\infty_\sc(B)$. For any \rfhgho\ $P$, there is an associated antisymmetric bilinear form
\begin{equation}
	E_P(f,h)=\langle f,E_P h\rangle, \qquad f,h\in\Gamma_0^\infty(B),
\end{equation} 
which is antisymmetric because $\adj{E}_P=-E_P$ and satisfies $E_P(\Cc f,\Cc h)=\overline{E_P(f,h)}$
because $E_P=\overline{E}_P$ and also $E_{P}(Pf,h)=E_{P}(f,Ph)=0$. A number of other general properties of Green-hyperbolic operators will be needed below and are stated here for clarity.
\begin{thm}\label{thm:GreenHypOps}
	The following statements hold for any Green-hyperbolic  $P:\Gamma^\infty(B)\to\Gamma^\infty(B)$:
	\begin{enumerate}[a)]
		\item\label{it:uniqueness} The Green operators $E_P^\pm$ are unique.
		\item $E_P^\pm:\Gamma_0^\infty(B)\to\Gamma^\infty_{\spc/\sfc}(B)$ and 
		$E_P:\Gamma_0^\infty(B)\to\Gamma^\infty_{\sc}(B)$ are continuous linear maps.
		\item For any $f\in\Gamma_0^\infty(B)$, $F=E^\pm_Pf$ is the unique solution to $PF=f$ with past/future compact support.
		\item\label{it:kernels} $\ker P\cap \Gamma_0^\infty(B)=\{0\}$; $\ker E_P=P\Gamma_0^\infty(B)$; $\ker P\cap \Gamma^\infty_\sc(B)=E_P\Gamma_0^\infty(B)$;  $P\Gamma^\infty_\sc(B)=\Gamma^\infty_\sc(B)$.
		\item\label{it:refinedkernel} For any compact $K\subset M$, 
		\begin{equation}
			P\Gamma_K^\infty(B)\subset \ker E_P \cap \Gamma^\infty_K(B)\subset P\Gamma^\infty_{J^+(K)\cap J^-(K)}(B).
		\end{equation}
		\item There is an injective linear isomorphism 
		\begin{align}
			\hat{E}_P:\Gamma_0^\infty(B)/P\Gamma_0^\infty(B)&\to E_P\Gamma_0^\infty(B) \nonumber \\
			[f]& \mapsto  E_Pf
		\end{align} 
		and $\sigma_P(E_Pf,E_Ph):=E_P(f,h)$ defines a 
		symplectic form $\sigma_P$ on $E_P\Gamma_0^\infty(B)$;
		consequently $\hat{\sigma}_P =\sigma_P \circ
		(\hat{E}_P\times\hat{E}_P)$ is a symplectic form on
		$\Gamma_0^\infty(B)/P\Gamma_0^\infty(B)$.
		
		\item\label{it:Cauchy} For any open causally convex $N\subset M$ that contains a Cauchy surface of $\Mb$, one has
		\begin{equation}
			\Gamma_0^\infty(B) = \Gamma_0^\infty(B|_N)+P\Gamma_0^\infty(B).
		\end{equation}
		Explicitly, let $U$ be an open Cauchy slab with $U\subset N$ and let $\chi\in C^\infty(M)$ obey $\chi=1$ on a neighbourhood of $M\setminus J^+(U)$ and $\chi=0$ on a neighbourhood of $M\setminus J^-(U)$. Then $P\chi E_P=[P,\chi]E_P:\Gamma_0^\infty(B)\to\Gamma_0^\infty(B|_U)\subset\Gamma_0^\infty(B|_N)$ and $[P,\chi]E_Pf-f\in P\Gamma_0^\infty(B|_N)$ for all $f\in \Gamma_0^\infty(B)$.
	\end{enumerate}
\end{thm}
\begin{proof}
	All statements except~(e) and~(f) are established in~\cite{Baer:2015}, or are simple consequences of the exact sequence
	\begin{equation}\label{eq:exactseq}
		\begin{tikzcd}
			0 \arrow[r] & \Gamma_0^\infty(B) \arrow[r, "P"]  &\Gamma_0^\infty(B) \arrow[r, "E_P"]  & \Gamma^\infty_{\mathrm{sc}}(B) \arrow[r,"P"]  & \Gamma^\infty_{\mathrm{sc}}(B) \arrow[r]  & 0 ,
		\end{tikzcd}
	\end{equation}
	stated as Theorem 3.22 of~\cite{Baer:2015} and proved as Theorem 3.5 of~\cite{BaerGinoux:2012}.  
	Part~(e) is a slightly refined version of a statement in (d), from which the first inclusion follows. For the second, suppose that $f\in\Gamma^\infty_K(B)$ with $E_Pf=0$. Then $E_P^+ f=E_P^-f$ has support in the compact set $J^+(K)\cap J^-(K)$ and one applies $P$ to obtain $f\in P\Gamma^\infty_{J^+(K)\cap J^-(K)}(B)$. 
	For part~(f), the statement concerning $\hat{E}_P$ also follows from the exact sequence, while the rest follows from Proposition~3.4 of~\cite{BaerGinoux:2012}.
\end{proof}
The symplectic form $\sigma_P$ of part~(f) has the further property that
\begin{equation}
	\sigma_P(\Cc E_P f, \Cc E_Ph) = \sigma_P( E_P \Cc f, E_P\Cc h)
	=E_P(\Cc f,\Cc h)= \overline{E_P(f,h)}=
	\overline{\sigma_P(E_Pf,E_Ph)}
\end{equation}
for all $f,h\in\Gamma_0^\infty(B)$. Consequently
$(E_P\Gamma_0^\infty(B),\sigma_P,\Cc)$ is a 
complexified symplectic space,
as is $(\Gamma_0^\infty(B)/P\Gamma_0^\infty(B),\hat{\Cc},\hat{\sigma}_P)$,
where $\hat{\Cc}[f]=[\Cc f]$. In more detail~\cite{FewVer:dynloc2}, the category of complexified symplectic spaces $\Sympl_\CC$ has objects $(V,\sigma,C)$, where $V$ is a complex vector space, $\sigma:V\times V\to \CC$ is a weakly nondegenerate antisymmetric bilinear form and $C$ is an
antilinear conjugation on $V$ such that $\sigma(Cu,Cv)=\overline{\sigma(u,v)}$. A $\Sympl_\CC$-morphism $T:(V,\sigma,C)\to(V',\sigma',C')$ is a (necessarily injective, due to nondegeneracy) linear map $T:V\to V'$ such that $\sigma'\circ(T\times T)=\sigma$ and which is real in the sense that $C'T=TC$. Composition in $\Sympl_\CC$ is given by composition of linear maps. 

\paragraph{Examples} A \emph{normally hyperbolic operator} is a second-order differential operator $P:\Gamma^\infty(B)\to\Gamma^\infty(B)$ whose principal symbol is given by
$p(x,k)=-g_x^{-1}(k,k)\id_{B_x}$, where $g$ is the spacetime metric. Thus (recalling the correspondence $k_\mu\leftrightarrow -\ii \partial_\mu$) in a local frame for $B$, $P$ is equivalent to a matrix of differential operators $P^{a}_{\phantom{a}b}$ with  
\begin{equation}
	P^{a}_{\phantom{a}b} = g^{\mu\nu}\nabla_\mu \nabla_\nu \delta^{a}_{\phantom{a}b} + U^{a}_{\phantom{a}b} ,
\end{equation}	
where $U^{a}_{\phantom{a}b}$ is a matrix of first order differential operators and $\nabla$ is the Levi--Civita derivative.  
Any normally hyperbolic operator has advanced and retarded Green operators (see Corollary~3.4.3 in~\cite{BarGinouxPfaffle}) and is therefore semi-Green-hyperbolic; as the formal transpose of a normally hyperbolic operator is also normally hyperbolic, these operators are Green-hyperbolic. 

However, not all Green-hyperbolic operators are normally hyperbolic, as shown by the Dirac and Proca field equations (see Sections~\ref{sec:fermionic} and~\ref{sec:Proca} respectively) and, even more starkly, the elliptic operator $-\dd^2/\dd t^2$ on $C^\infty(\RR)$ shows that Green-hyperbolic operators need not be hyperbolic. Another example is given in~\cite{MMVparacausal:2023}: suppose that $g'$ is a time orientable slow metric relative to the spacetime metric $g$ and suppose that $P'$ and $P$ are normally hyperbolic operators with respect to $g'$ and $g$. Let $\chi\in C_0^\infty(M;[0,1])$ and
define a convex combination $P_\chi=(1-\chi)P + \chi P'$. Then there is a metric $g_\chi$ so that
$g_\chi^{-1}=(1-\chi)g^{-1}+\chi (g')^{-1}$, which is a slow metric relative to $g$ (Theorem~2.18 in~\cite{MMVparacausal:2023}) that is also globally hyperbolic. Then $P_\chi$ is normally hyperbolic (and therefore Green-hyperbolic) with respect to $g_\chi$, while it is Green-hyperbolic, but not normally hyperbolic, with respect to the spacetime metric $g$ -- see Theorem~3.14 in~\cite{MMVparacausal:2023}. Further examples can be found in~\cite{Baer:2015}.  
 
\paragraph{The category $\GreenHyp$} By Theorem~\ref{thm:GreenHypOps}~(f), every \rfhgho\ determines an object of $\Sympl_\CC$. We define
the category $\GreenHyp$ of \rfhgho's so that this assignment is functorial. Previously developed categories of Green-hyperbolic operators, such as $\GlobHypGreen$ in~\cite{BaerGinoux:2012} have focussed on intertwining the Green-hyperbolic operators with  pushforward maps -- as will be seen, our approach is more general and has certain advantages.
\begin{defn}
	Objects of $\GreenHyp$ are pairs $(B,P)$ (sometimes abbreviated just to $P$), where $B$ is a hermitian vector bundle over a globally hyperbolic spacetime and $P$ is a \rfhgho\ acting on $\Gamma^\infty(B)$. A morphism $S:(B,P)\to (B',P')$ in $\GreenHyp$ is a linear map $S:\Gamma_0^\infty(B)\to\Gamma_0^\infty(B')$ 
	obeying $S\Cc=\Cc' S$, $SP\Gamma_0^\infty(B)\subset P'\Gamma_0^\infty(B')$, and $E_{P'}(Sf,Sh)=E_{P}(f,h)$ for all $f,h\in\Gamma_0^\infty(B)$. 
\end{defn} 	 
	A $\GreenHyp$-morphism $S:(B,P)\to (B',P')$ determines a
	linear map $\Sol(S):E_P \Gamma_0^\infty(B)\to E_{P'}\Gamma_0^\infty(B')$ by $\Sol(S)E_Pf= E_{P'}Sf$, which is well-defined because
	$E_Pf=0$ implies $f\in P\Gamma_0^\infty(B)$ and hence $Sf\in P'\Gamma_0^\infty(B')=\ker E_{P'}$. The identity $\Sol(S)\Cc=\Cc'\Sol(S)$ follows from $S\Cc=\Cc' S$, and	
	the calculation
	\begin{equation}
		\sigma_{P'}(E_{P'}Sf,E_{P'}Sh) = E_{P'}(Sf,Sh)=E_{P}(f,h) = \sigma_{P}(E_{P}f,E_{P}h)
	\end{equation} 
	shows that $\Sol(S)$ is symplectic and therefore a $\Sympl_\CC$-morphism.
	Consequently, $\Sol(S)$ is injective; we will describe $S$ as \emph{Cauchy} if $\Sol(S)$ is an isomorphism, and say that
	morphisms $S,T:(B,P)\to (B',P')$ are \emph{equivalent}, written $S\sim T$, if $\Sol(S)=\Sol(T)$. 
	A straightforward calculation shows that $\hat{E}_{P'}\hat{S}= \Sol(S)\hat{E}_P$, where  
	\begin{align}
		\hat{S}:\Gamma_0^\infty(B)/P \Gamma_0^\infty(B)&\to \Gamma_0^\infty(B')/P' \Gamma_0^\infty(B')
		\nonumber \\ \hat{S}[f]&= [S f].
	\end{align} 
	Thus $\hat{S}$ is injective, $S$ is Cauchy if and only if $\hat{S}$ is an isomorphism, and morphisms $S,T:(B,P)\to (B',P')$ are \emph{equivalent} if and only if $\hat{S}=\hat{T}$. 
	\begin{thm}\label{thm:sol}
		There is a functor $\Sol:\GreenHyp\to\Sympl_\CC$ given by $\Sol(B,P)=(E_P\Gamma_0^\infty(B),\sigma_P,\Cc)$, where $\sigma_P(E_P f,E_P h)=E_P(f,h)$,
		and $\Sol(S)E_P f=E_{P'}S f$. If $S:(B,P)\to (B',P')$ is a Cauchy $\GreenHyp$-morphism
		then $\Sol(S)$ is a $\Sympl_\CC$-isomorphism. One has 
		\begin{equation}\label{eq:solsum}
			\Sol(B\oplus B',P\oplus P')=(E_P\Gamma_0^\infty(B)\oplus E_{P'}\Gamma_0^\infty(B'),\sigma_P\oplus \sigma_{P'},\Cc\oplus \Cc').
		\end{equation}   
	\end{thm}
	\begin{proof}
		It has already been established that $\Sol(B,P)\in\Sympl_\CC$ for each $(B,P)\in\GreenHyp$,
		and that a $\GreenHyp$-morphism $S:(B,P)\to (B',P')$ 
		determines a $\Sympl_\CC$-morphism $\Sol(S):\Sol(B,P)\to\Sol(B',P')$. It is clear that $\Sol(\id_{(B,P)})=\id_{\Sol(B,P)}$. Given $\GreenHyp$ morphisms $P\stackrel{S}{\to}P' \stackrel{S'}{\to}P''$ then $\Sol(S'S)E_{P}=E_{P''}S'S=\Sol(S')E_{P'}S=\Sol(S')\Sol(S)E_P $ on $\Gamma_0^\infty(B)$,
		so $\Sol$ is a functor. 
		
		Finally, the statement~\eqref{eq:solsum} is immediate from the definitions. 
	\end{proof}
		
	\begin{lemma}\label{lem:GreenHyp0}
		(a) A composition of Cauchy $\GreenHyp$-morphisms is Cauchy. Equivalence of $\GreenHyp$-morphisms is stable under composition, i.e., $ST\sim S'T'$ if $S\sim S'$ and 
		$T\sim T'$. 
		
		(b) If $S:(B,P)\to (B',P')$ and $T:\Gamma_0^\infty(B)\to\Gamma_0^\infty(B')$ is linear then $T:(B,P)\to (B',P')$ with $T\sim S$ if and only if $T-S=P'R$ for (uniquely determined) linear map $R:\Gamma_0^\infty(B)\to \Gamma_0^\infty(B')$ that obeys $R\Cc=\Cc'R$.
		
		(c) If $S:(B,P)\to (B',P')$ has a formal transpose $\adj{S}:\Gamma_\sc^\infty(B')\to\Gamma^\infty(B)$ then $\adj{S}E_{P'}S=E_P$; if $S$ is Cauchy then $\adj{S}E_{P'}h= \hat{E}_{P}\hat{S}^{-1}[h]$ for all $h\in\Gamma_0^\infty(B')$.
	\end{lemma}
	\begin{proof}
		(a) Both statements follow from the trivial identity $\widehat{ST}=\hat{S}\hat{T}$ for composable morphisms $S$, $T$.
		
		(b) Sufficiency: one clearly has $T\Cc=\Cc' T$ and $TP\Gamma_0^\infty(B)\subset P'\Gamma_0^\infty(B')$, and $E_{P'}(Tf,Th)= E_{P'}(Sf,Sh)=E_{P}(f,h)$ for all $f,h\in\Gamma_0^\infty(B)$; finally, $(\hat{T}-\hat{S})[f]=[P'Rf]=0$ for all $f\in\Gamma_0^\infty(B)$, so $T$ defines a $\GreenHyp$-morphism $T:(B,P)\to (B',P')$ with $T\sim S$.
		Necessity: as $P'$ is injective on $\Gamma_0^\infty(B')$ by Theorem~\ref{thm:GreenHypOps}\eqref{it:kernels} and $T\sim S$ entails $(T-S)\Gamma_0^\infty(B)\subset P'\Gamma_0^\infty(B')$, 
		there is a unique map $R:\Gamma_0^\infty(B)\to \Gamma_0^\infty(B')$ so that 
		$T-S=P'R$, which is necessarily linear and satisfies $\Cc' R=R\Cc$. 
		
		(c) We calculate 
		\begin{equation}
			\langle f,\adj{S}E_{P'}Sh\rangle=\langle Sf,E_{P'}Sh\rangle= E_{P'}(Sf,Sh)=E_{P}(f,h)=\langle f,E_P h\rangle
		\end{equation}
		for all $f,h\in\Gamma_0^\infty(B)$. If $S$ is Cauchy and $h\in\Gamma_0^\infty(B')$ then there exists $f\in\Gamma_0^\infty(B)$ so that $[h]=\hat{S}[f]=[Sf]$, whereupon $\adj{S}E_{P'}h=\adj{S}E_{P'}Sf=E_Pf=
		\hat{E}_P[f]=\hat{E}_P \hat{S}^{-1}[h]$. 
	\end{proof}
	
Some basic examples of $\GreenHyp$-morphisms are set out in the next result, which we formulate for ease of use in what follows.
\begin{thm}\label{thm:GreenHyp}
		Suppose $(B,P)$ and $(B',P')$ are objects in $\GreenHyp$ over spacetimes $\Mb$ and $\Mb'$. Then
		\begin{enumerate}[a)]
			\item If $\beta:B\to B'$ is a $\HVB$-morphism such that $P'\beta_*=\beta_*P$, then $\beta_*:(B,P)\to (B',P')$ defines a $\GreenHyp$-morphism. If $\beta$ is a fibrewise isomorphism covering a Cauchy $\Loc$-morphism then $\beta_*$ is Cauchy. 
			\item If $N\subset M$ is open and causally convex then the inclusion map $\iota_N:B|_N\to B$ determines a $\GreenHyp$-morphism $\iota_{N*}:(B|_N,P|_N)\to (B,P)$ which is Cauchy if $N$ contains a Cauchy surface of $\Mb$.
			\item Define 
			$(B\oplus B',P\oplus P')$ with the obvious
			complex conjugation and bundle metric on $B\oplus B'$ defined by taking direct sums,
			noting that $E_{P\oplus P'}^\pm=E_{P}^\pm \oplus E_{P'}^\pm$. Then there are $\GreenHyp$-morphisms 
			\begin{equation}
				(B,P)\stackrel{\iota_*}{\longrightarrow} (B\oplus B',P\oplus P')\stackrel{\iota'_*}{\longleftarrow} (B',P'),
			\end{equation} 
			where $\iota:B\to B\oplus B'$ and $\iota':B'\to B\oplus B'$ are  the canonical injections of
			$B$ and $B'$ into $B\oplus B'$ (which are bundle isometries but not isomorphisms).
			\item If $S:(B,P)\to (B',P')$ is a Cauchy $\GreenHyp$-morphism then there are (many, necessarily equivalent)
			$\GreenHyp$-morphisms $L:(B',P')\to (B,P)$ such that $\hat{L}=\hat{S}^{-1}$.
			Explicitly, take any open Cauchy slab $U\subset M$ and choose $\chi\in C^\infty(M;\RR)$ with $\chi=1$ on a neighbourhood of $M\setminus J^+(U)$ and $\chi=0$ on a neighbourhood of $M\setminus J^-(U)$. Then the required properties are met by $L:\Gamma_0^\infty(B')\to \Gamma_0^\infty(B)$ given by $Lf=[P,\chi]\hat{E}_P \hat{S}^{-1}[f]$. 
			If $S$ has a formal transpose $\adj{S}:\Gamma^\infty_\sc(B')\to \Gamma^\infty(B)$ then
			$L=[P,\chi]\adj{S}E_{P'}$. 
			\item In the case $B'=B$, suppose that $P'$ agrees with $P$ on a 
			causally convex open set $N\subset M$ that contains a Cauchy surface of $\Mb$. 
			Then there exist Cauchy $\GreenHyp$-morphisms $S:(B,P)\to (B,P')$ such that
			$S\Gamma_0^\infty(B)\subset \Gamma_0^\infty(B|_N)$ and
			$\hat{S}[f]_P=[f]_{P'}$ for all $f\in\Gamma_0^\infty(B|_N)$; all such morphisms are equivalent.
			Explicitly, choose 
			an open Cauchy slab $U\subset N$ and $\chi\in C^\infty(M;\RR)$ with $\chi=1$ on a neighbourhood of $M\setminus J^+(U)$ and $\chi=0$ on a neighbourhood of $M\setminus J^-(U)$. Then $S=P\chi E_P=[P,\chi]E_P$ has the stated properties and
			$(S-\id)\Gamma_0^\infty(B)\subset P\Gamma_0^\infty(B)$.
		\end{enumerate}	
\end{thm}
\begin{proof}
	(a) As $\beta$ is a $\HVB$-morphism, one has $\Cc'\beta_*=\beta_*\Cc$. Moreover,
	$\beta_* P\Gamma_0^\infty(B)=P'\beta_*\Gamma_0^\infty(B)\subset P'\Gamma_0^\infty(B')$.
	Taking the formal transpose of $P'\beta_*=\beta_*P$	
	 gives $P\adj{\beta}_*= \adj{\beta}_* P'$, and we compute
	\begin{align} 
		P \adj{\beta}_*E_{P'}^\pm \beta_* &=  \adj{\beta}_* P' E_{P'}^\pm  \beta_* =  \adj{\beta}_* \beta_*=\id\nonumber\\
		\adj{\beta}_*E_{P'}^\pm \beta_* P &= 	\adj{\beta}_* E_{P'}^\pm P\beta_* =\adj{\beta}_* \beta_*=\id
	\end{align}
	on $\Gamma_0^\infty(B)$, so $\adj{\beta}_*E_{P'}^\pm\beta_*$ satisfy G1 and G2 for $P$. 
	Meanwhile, G3 holds because $\adj{\beta}_*E_{P'}^\pm \beta_*\Gamma_K^\infty(B)\subset 
	\adj{\beta}_*E_{P'}^\pm  \Gamma_{\psi(K)}^\infty(B)\subset 
	\adj{\beta}_* \Gamma_{J^\pm(\psi(K))}^\infty(B')\subset \Gamma_{J^\pm( K)}^\infty(B)$ (see~\eqref{eq:betastarsupport}, recalling that $\beta$ covers a $\Loc$-morphism). Thus $\adj{\beta}_*E_{P'}^\pm\beta_*$ are Green operators for $P$ and Theorem~\ref{thm:GreenHypOps}\eqref{it:uniqueness} gives $\adj{\beta}_*E_{P'}^\pm\beta_*=E_{P}^\pm$ and therefore $\adj{\beta}_*E_{P'}\beta_*=E_P$. 
	Finally, if $\psi(M)$ contains a $\Mb'$-Cauchy surface, Theorem~\ref{thm:GreenHypOps}\eqref{it:Cauchy} gives $\Gamma_0^\infty(B')=\Gamma_0^\infty(B'|_{\psi(M)})+P'\Gamma_0^\infty(B')=
	\beta_*\Gamma_0^\infty(B)+P'\Gamma_0^\infty(B')$ as $\beta_*\,\adj{\beta}_*=\id$
	on $\Gamma_0^\infty(B'|_{\psi(M)})$ because $\beta$ is a fibrewise isomorphism.
	Consequently, $\hat{\beta}_*$ is surjective as well as injective and $\beta_*$ is therefore Cauchy.
	
	(b,c) are immediate applications of (a) (see Lemma~3.6 in~\cite{Baer:2015} for a proof that $P|_N$ is Green hyperbolic; reality and formal hermiticity are easily checked).
	
	(d) It is enough to check the explicit example, for which $\Cc L=L\Cc'$ and $LP'=0$ are straightforward.
		Note that $LSf=[P,\chi]\hat{E}_{P}[f]=[P,\chi]E_Pf$, so by
		Theorem~\ref{thm:GreenHypOps}\eqref{it:Cauchy}, one has $LS\Gamma_0^\infty(B)\subset \Gamma_0^\infty(B|_U)\subset\Gamma_0^\infty(B|_N)$ 
		and $LSf=f \mod P\Gamma_0^\infty(B)$, and hence $\hat{L}\hat{S}=\id$; by surjectivity of $\hat{S}$ it follows that $\hat{L}=\hat{S}^{-1}$. Furthermore, 
		\begin{equation}\label{eq:ident}
			E_{P}(LS f,LS h)=E_{P}(f,h)=E_{P'}(Sf,Sh)
		\end{equation}
		for all $f,h\in\Gamma_0^\infty(B)$. As $\hat{S}$ is an isomorphism, for every
		$f'\in\Gamma_0^\infty(B')$ there is $f\in\Gamma_0^\infty(B)$ so that $Sf-f'\in P'\Gamma_0^\infty(B')$ and therefore $Lf'=LSf$. Combining with~\eqref{eq:ident}, one has
		$E_P(Lf',Lh')=E_{P'}(Sf,Sh)=E_{P'}(f',h')$ for all $f',h'\in\Gamma_0^\infty(B')$, 
		so $L:(B',P')\to (B,P)$. 		 
		The last statement follows by Lemma~\ref{lem:GreenHyp0}(c).  
	
	(e) Let $\Mb|_N$ be the region $N$ considered as a spacetime in its own right, with the induced metric and causal structure from $\Mb$. By~(b), there are Cauchy morphisms $\iota_N:(B|_N,P|_N)\to (B,P)$ and $\iota'_{N}:(B|_N,P'|_N)\to (B,P')$ with the same underlying map. By~(d) there is a Cauchy morphism $L_N:(B,P)\to (B|_N,P|_N)$ so that $\hat{L}_N=\hat{\iota}_{N*}^{-1}$, so $S_N=\iota'_N L_N:(B,P)\to(B,P')$ is also Cauchy. Given $\chi$ and $U$ as in the statement, part~(d) gives an explicit expression for $L_N=[P|_N,\chi|_N]\adj{\iota}_{N*}E_P$, and so $S_N=\iota'_{N\*} [P|_N,\chi|_N]\adj{\iota}_{N*}E_P= [P,\chi]E_P$ as $[P,\chi]$ is a differential operator with coefficients supported in $N$. It is straightforward to see that $S_N\Gamma_0^\infty(B)\subset \Gamma_0^\infty(B|_N)$ and $S_N[f]_{P}=[f]_{P'}$ for $f\in\Gamma_0^\infty(B|_N)$. 
	Thus $S=S_N$ has the stated properties. Now consider any $S$ of this type, noting that $\hat{S}$ is fixed uniquely by its action on $\Gamma_0^\infty(B|_N)$ because $\Gamma_0^\infty(B)=\Gamma_0^\infty(B|_N)+P'\Gamma_0^\infty(B)$. Thus all morphisms $S:(B,P)\to (B,P')$ obeying  $\hat{S}[f]_P=[f]_{P'}$ for all $f\in\Gamma_0^\infty(B|_N)$ are equivalent.
	The last statement follows by Theorem~\ref{thm:GreenHypOps}(g).
\end{proof}
The morphisms of Theorem~\ref{thm:GreenHyp}(a) generalise those of the category $\GlobHypGreen$ introduced in~\cite{BaerGinoux:2012}, which however required $\beta$ to be a fibrewise isometric isomorphism (between real vector bundles).

\subsection{Quantisation}\label{sec:Yf_quant}

We now describe how every \rfhgho\ defines a bosonic hermitian QFT. Given the definitions so far, the quantisation process may be described functorially. Let $\Alg$ be the category whose objects are unital complex $*$-algebras, and whose morphisms are unit-preserving injective $*$-homomorphisms. We have already seen that there is a functor
$\Sol:\GreenHyp\to\Sympl_\CC$ turning \rfhgho's into symplectic spaces; by composing this with a  quantisation functor $\Qf_s:\Sympl_\CC\to\Alg$ (see, e.g.,~\cite{FewVer:dynloc2}) we obtain the desired functor from $\GreenHyp$ to $\Alg$. 
The main theorem on quantisation is essentially standard and we state it without detailed proof. 
\begin{thm} \label{thm:RFHGHOquantisation}
	There is a functor $\Yf:\GreenHyp\to \Alg$ with the following properties:
	\begin{enumerate}[a)]
	\item For each $(B,P)\in\GreenHyp$, 
	$\Yf(P)$ is the unital $*$-algebra generated by symbols $\Upsilon_P({f})$ ($f\in\Gamma_0^\infty(B))$, subject to 
	\begin{enumerate}[Y1] 
		\item $f \mapsto \Upsilon_{P}({f})$ is complex linear  
		\item $\Upsilon_{P}({f})^* = \Upsilon_{P}(\Cc f)$  
		\item $\Upsilon_{P}({P} {f}) = 0$,   
		\item $[\Upsilon_{P}({f}),\Upsilon_{P}(h)] = \ii E_{P}({f},{h}) \1_{\Yf(P)}$  
	\end{enumerate}
	for all ${f},h\in\Gamma_0^\infty(B)$, and where $\Cc$ is the antilinear bundle involution of $B$.
	
	\item If $S:P\to P'$ is a $\GreenHyp$-morphism, then the $\Alg$-morphism $\Yf(S):\Yf(P)\to\Yf(P')$ is uniquely determined by
	\begin{equation}
		\Yf(S)\Upsilon_P({f}) = \Upsilon_{P'}(S{f}),\qquad {f}\in\Gamma_0^\infty(B),
	\end{equation}
	and in particular, $\Yf(S)=\Yf(T)$ if and only if $S\sim T$.
	
	\item $\Yf(S)$ is an $\Alg$-isomorphism if $S$ is a Cauchy $\GreenHyp$-morphism. 
	
	\item If $P$ and $P'$ are \rfhgho's on hermitian bundles $B$ and $B'$ over a common $\Mb\in\Loc$, then  $\Yf(P\oplus P')= \Yf(P)\otimes \Yf(P')$ (the algebraic tensor product) under the correspondence
	\begin{equation}\label{eq:Yf_monoidal}
		\Upsilon_{P\oplus P'}({f}\oplus{f}') = 
		\Upsilon_{P}({f})\otimes \1_{\Yf(P')} + \1_{\Yf(P)} \otimes \Upsilon_{P'}({f}'), \qquad 
		{f}\in\Gamma_0^\infty(B),~{f}'\in\Gamma_0^\infty(B').
	\end{equation} 
	\end{enumerate}
\end{thm}
To be more explicit, the required functor is $\Yf=\Qf_s\circ\Sol$. As a vector space
	\begin{equation}
		\Yf(P) = \bigoplus_{n\in\NN_0} \Sol(P)^{\odot n},
	\end{equation}
	where $\odot$ is the symmetrised tensor product, the $*$-operation is $(v_1\odot\cdots\odot v_n)^*=(\Cc v_n\odot\cdots\odot \Cc v_1)$ and the generators are given by $\Upsilon_P(f) = 
	(0,E_Pf,0,\ldots)$.
	The algebra product of $\Yf(P)$ is given explicitly by 
	\begin{equation}\label{eq:starprodbose}
		FH = \odot\circ \exp \left(\tfrac{1}{2}\ii\langle \sigma,\partial\otimes\partial\rangle\right) F\otimes H,
	\end{equation}
	which is a translation of the star product from the functional formalism~\cite{Rejzner_2022} to the current setting. Here,
	$\partial:\Yf(P)\to \Sol(P)\otimes\Yf(P)$ is a linear map such that
	\begin{equation}
		\partial  (v_1 \odot \cdots \odot v_k)  = \sum_{j=1}^k v_j \otimes (v_1\odot\cdots \slashed{v_j}  \cdots \odot v_k), \qquad \partial \1_{\Yf(P)}=0,
	\end{equation}
	where the slash indicates an omission, $\1_{\Yf(P)}=(1,0,\ldots)\in\Yf(P)$ and 
	$\langle \sigma,\cdot\rangle:(\Sol(P)\otimes\Yf(P))^{\otimes 2}\to \Yf(P)^{\otimes 2}$
	is a linear map so that
	\begin{align}
		\langle \sigma,	(v\otimes F) \otimes (w\otimes H) \rangle= \sigma(v,w) F\otimes H.
	\end{align} 
	As $\langle \sigma_P,\partial\otimes\partial\rangle$ is nilpotent on any element of $\Yf(P)\otimes\Yf(P)$, the exponential in~\eqref{eq:starprodbose} can be understood as a power series. It can be checked that this indeed defines an associative algebra product  (cf.\ the argument in Section~6.2.4 of~\cite{Waldmann:PoissonGeometrie} or Proposition II.4 of~\cite{Keller:2010}).
	In particular, one has
	\begin{equation}
	 \Upsilon_P(f)\Upsilon_P(h) = \Upsilon_P(f)\odot \Upsilon_P(h) + \frac{\ii}{2}\sigma(E_Pf,E_Ph)\1_{\Yf(P)},\qquad [\Upsilon_P(f),\Upsilon_P(h) ]=\ii E_P(f,h)\1_{\Yf(P)}.
	\end{equation}
	 Part~(d) of Theorem~\ref{thm:RFHGHOquantisation} is based on the isomorphism
	 of the symmetric tensor algebra over $\Sol(P\oplus P')$ with the algebraic 
	 tensor product of those over $\Sol(P)$ and $\Sol(P')$ (cf.\ III\S 6.6 of~\cite{Bourbaki_algebra1}); a calculation shows that the same map is an isomorphism
	 from $\Yf(P\oplus P')$ to the algebraic tensor product $\Yf(P)\otimes\Yf(P')$ with the stated action on generators when the two algebras are identified under the isomorphism,
	 $\Yf(P\oplus P')=\Yf(P)\otimes\Yf(P')$. One interprets $\Upsilon_P(f)$ as the quantisation of the classical functional $F_f:\Sol(P)\to \CC$ given by $F_f(\phi)=\langle f,\phi\rangle$. 
 
Any functorial assignment of \rfhgho's leads to a functorial assignment of algebras, just by post-composing with the quantisation functor $\Yf$. This is particularly useful when the \rfhgho's arise from a natural equation of motion.
	\begin{defn}\label{def:functorial_EoM}
		A \emph{functorial bundle over $\Loc$} is a functor $\Bf:\Loc\to \HVB$ assigning hermitian vector bundles to globally hyperbolic spacetimes such that $\Bf(\Mb)$ is a bundle over $\Mb$ for each $\Mb\in\Loc$.
		 A \emph{natural differential operator on $\Bf$} is a family $(P_\Mb)_{\Mb\in\Loc}$ of linear partial differential operators $P_\Mb:\Gamma^\infty(\Bf(\Mb))\to \Gamma^\infty(\Bf(\Mb))$ so that
		the identity $P_{\Mb'}\circ \Bf(\psi)_* = \Bf(\psi)_*\circ P_{\Mb}$ holds on $\Gamma_0^\infty(\Bf(\Mb))$		
		for all $\Loc$-morphisms $\psi:\Mb\to\Mb'$. 
		A functor $\EoM:\Loc\to\GreenHyp$ is a \emph{functorial equation of motion} if $\EoM(\Mb)= (\Bf(\Mb),P_\Mb)$, where $(P_\Mb)_{\Mb\in\Loc}$ is a natural differential operator on a functorial   bundle $\Bf:\Loc\to\HVB$ over $\Loc$, and $\EoM(\psi)=\Bf(\psi)_*$ for every $\Loc$-morphism $\psi$.
	\end{defn}
Conversely, if $(P_\Mb)_{\Mb\in\Loc}$ is a natural differential operator on a functorial bundle $\Bf:\Loc\to\HVB$ over $\Loc$ and $(\Bf(\Mb),P_\Mb)$ is a \rfhgho\ for each $\Mb\in\Loc$, then one has a functorial equation of motion $\EoM$ so that $\EoM(\Mb)= (\Bf(\Mb),P_\Mb)$ for all $\Mb\in\Loc$ and $\EoM(\psi)=\Bf(\psi)_*$ for every $\Loc$-morphism $\psi:\Mb\to\Mb'$. Here we use the fact that $\Bf(\psi)_*$ is a $\GreenHyp$-morphism between 
$(\Bf(\Mb),P_\Mb)$ and $(\Bf(\Mb'),P_{\Mb'})$ by Theorem~\ref{thm:GreenHyp}(a). Moreover, if $\Bf(\psi):\Bf(\Mb)\to\Bf(\Mb')$ is a fibrewise isomorphism for every Cauchy morphism $\psi$, then $\EoM$ maps Cauchy $\Loc$-morphisms to Cauchy $\GreenHyp$-morphisms, also by Theorem~\ref{thm:GreenHyp}(a).

\begin{thm}\label{thm:lcqft}
	Suppose that $\EoM:\Loc\to\GreenHyp$ is a functorial equation of motion that takes Cauchy $\Loc$-morphisms to Cauchy $\GreenHyp$-morphisms. Then $\Zf:=\Yf\circ \EoM:\Loc\to\Alg$ defines a locally covariant QFT~\cite{BrFrVe03,FewVerch_aqftincst:2015} across all globally hyperbolic spacetimes, obeying Einstein causality and the timeslice property.
\end{thm}
\begin{proof}
	Let $\Bf$ and $(P_\Mb)_{\Mb\in\Loc}$ be the functorial bundle over $\Loc$ and natural differential operator on $\Bf$ comprising $\EoM$. We also write $Z_\Mb(f)=\Upsilon_{P_\Mb}(f)$ for each $\Mb\in\Loc$, $f\in\Gamma_0^\infty(\Bf(\Mb))$.
	$\Zf$ is certainly a functor from $\Loc$ to $\Alg$ and so defines a locally covariant QFT. Einstein causality is the requirement
	that for any pair of morphisms $\psi_j:\Mb_j\to \Mb$ with causally disjoint images $\psi_j(M_j)$ in $M$,
	the images of the morphisms $\Zf(\psi_j)$ commute elementwise in $\Zf(\Mb)$. 
	Now, for any $f_j\in\Gamma_0^\infty(\Bf(\Mb_j))$, one has $\Zf(\psi_j)(Z_{\Mb_j}(f_j))= Z_{\Mb}(\EoM(\psi_j)f_j)=Z_{\Mb}(\Bf(\psi_j)_*f_j)$. Therefore
	\begin{equation}
		[\Zf(\psi)(Z_{\Mb_1}(f_1)),\Zf(\psi)(Z_{\Mb_2}(f_2))]= \ii E_{P_{\Mb}}(\Bf(\psi_1)_*f_1,\Bf(\psi_2)_*f_2)\1_{\Zf(\Mb)} = 0
	\end{equation} 
	due to the causal disjointness of $\psi$ and the 
	support properties of \rfhgho\ Green operators. Consequently $\Zf(\psi)(\Zf(\Mb_1))$ and 
	$\Zf(\psi)(\Zf(\Mb_2))$ commute elementwise, and Einstein causality is proved.
	
	If $\psi$ is a Cauchy morphism in $\Loc$ then $\EoM(\psi)$ is a Cauchy morphism in $\GreenHyp$, so $\Zf(\psi)=\Yf(\EoM(\psi))$ is an $\Alg$-isomorphism; thus $\Zf$ satisfies the timeslice property.
\end{proof}
Explicitly, the algebra $\Zf(\Mb)=\Yf(\EoM(\Mb))$ is generated by symbols $\Zc_\Mb(f)=\Upsilon_{\EoM(\Mb)}(f)$ ($f\in \Gamma_0^\infty(B)$) obeying the relations inherited from Y1--Y4, and with 
\begin{equation}
	\Zf(\psi)\Zc_\Mb(f) = \Zc_{\Mb'}(\EoM(\psi) f)
\end{equation}
for $\psi:\Mb\to\Mb'$ in $\Loc$.

\subsection{States and $n$-point functions}

Recall that a state on $*$-algebra $\Ac\in\Alg$ is a linear map $\omega:\Ac\to\CC$ such that
$\omega(\1)=1$ and $\omega(A^*A)\ge 0$ for all $A\in\Ac$, i.e., $\omega$ is a positive normalised linear map. The set of states on $\Ac$ is denoted $\Ac^*_{+,1}$ and, as usual, $\omega(A)$ is interpreted as the expectation value of $A=A^*\in\Ac$ in state $\omega$.

Suppose $P$ is a \rfhgho\ with hermitian vector bundle $B$ over globally hyperbolic spacetime $\Mb$, and let $\omega\in\Yf(P)^*_{+,1}$ be any state on $\Yf(P)$. Then we may define $n$-point correlators of the form
$\omega(\Yf_P(f_1)\cdots \Yf_P(f_n))$ for $f_1,\ldots,f_n\in\Gamma_0^\infty(B)$. If $\omega$ is sufficiently
regular then the correlators may be expressed as distributions in various ways. Most immediately, they define distributions in $\Gamma_0^\infty(B^{\boxtimes n})'=\DD'((B^*\otimes\Omega_M)^{\boxtimes n})$.
However, this has the disadvantage that $P$ would not act on the $n$-point functions directly. Instead, 
we will say that $\omega$ has distributional $n$-point functions if, for each $n\in\NN$, there is a distribution $W^{(n)}\in \DD'(B^{\boxtimes n})$ so that 
\begin{equation} \label{eq:npointfns}
	W^{(n)}(f_1^{\sharp,1}\otimes\cdots\otimes f_n^{\sharp,1}) = \omega(\Upsilon_P(f_1)\cdots \Upsilon_P(f_n))
\end{equation}
for $f_1,\ldots,f_n\in\Gamma_0^\infty(B)$, using the musical notation introduced in Section~\ref{sec:HVBs}. The most interesting case will be the two-point function, where we will usually drop the superscript and simply write $W\in\DD'(B\boxtimes B)= \Gamma_0^\infty((B^*\otimes\Omega_M)\boxtimes(B^*\otimes\Omega_M))'$ with
\begin{equation}
	W(f^{\sharp,1}\otimes h^{\sharp,1}) =\omega(\Upsilon_P(f)\Upsilon_P(h)), \qquad
	f,h\in\Gamma_0^\infty(B).
\end{equation}
A similar `two-point distribution' $T^\twopt\in\DD'(B\boxtimes B)$ can be defined for any continuous linear map $T:\Gamma_0^\infty(B)\to\Gamma^\infty(B)$, again using the bilinear bundle metric of $B$, by
\begin{equation}
	T^\twopt(f^{\sharp,1}\otimes h^{\sharp,1})  = T(f,h):=
	\langle f, Th\rangle
	, \qquad
	f,h\in\Gamma_0^\infty(B).
\end{equation}
Particular examples of note include two-point distributions for the retarded and advanced Green operators $E_P^\pm$, and their difference $E_P$. In particular, 
the CCRs entail the identity
\begin{equation}
	W - \adj{W} = \ii E_P^\twopt,
\end{equation}
where $\adj{W}(u\otimes v)=W(v \otimes u)$. 

As we have already described, any continuous linear map $T:\Gamma_0^\infty(B)\to\Gamma^\infty(B)$ also has an associated kernel distribution which does not rely on the bundle metric of $B$. The kernel distributions were defined for bundles with an explicit half-density factor, for which purpose we may regard $B=
(B\otimes\Omega^{-1/2})\otimes\Omega^{1/2}$. In this way $T:\Gamma_0^\infty(B)\to\Gamma^\infty(B)$ also
has kernel distribution
\begin{equation} 
	T^\knl \in \DD'((B\otimes\Omega^{-1/2})\boxtimes (B\otimes\Omega^{-1/2})^*\otimes\Omega^{1/2}_{M\times M}) =
	\DD'(B \boxtimes (B^* \otimes\Omega_{M})),
\end{equation}
which is related to $T^\twopt$ by 
\begin{equation}
	T^\knl(F\otimes h)= T^\twopt(F\otimes h^{\sharp,1}), \qquad F\in\Gamma_0^\infty(B^*\otimes\Omega_M),~
	h\in\Gamma_0^\infty(B),
\end{equation}
or equivalently, $T^\twopt=T^\knl\circ (1\otimes(\flat\circ \rho^{-1}))$, which brings out the fact that the
two-point distribution uses the bundle metric and volume density for its definition in a way that the kernel distribution does not. It will also be useful to consider another relation between two-point and kernel distributions. Namely, each continuous linear map $T:\Gamma_0^\infty(B)\to\Gamma^\infty(B)$ also determines
$T^{\natural,\tfrac{1}{2}}:\Gamma_0^\infty(B\otimes\Omega^{1/2})\to\Gamma^\infty(B\otimes\Omega^{1/2})$, whose
kernel distribution $(T^{\natural,\tfrac{1}{2}})^\knl\in \DD'((B\boxtimes B^*)\otimes\Omega^{1/2}_{M\times M})$ is related to $T^\twopt\in \DD'(B\boxtimes B)$ by
\begin{equation}\label{eq:T2ptTknl}
	(T^{\natural,\tfrac{1}{2}})^\knl(  f^{\sharp,\tfrac{1}{2}}\otimes h^{\sharp,\tfrac{1}{2}}) = 
	\langle f,Th\rangle = T^\twopt(  f^{\sharp,1}\otimes  h^{\sharp,1}), \qquad
	f,h\in\Gamma_0^\infty(B),
\end{equation}
or equivalently, $T^{\twopt}=(T^{\natural,\tfrac{1}{2}})^\knl\circ (\rho^{-1/2}\otimes (\rho^{1/2}\circ\flat\circ\rho^{-1}))$.
We drop the $\twopt$ or $\knl$ superscript where no confusion can arise between the underlying linear map and which of the two associated distributions is intended. As will be seen, $T^\twopt$ and $T^\knl$ have identical wavefront sets but their polarisation sets (though related) are valued in different bundles.

The operator $P$ extends to $\DD'(B)$ as the distributional dual of $\rho (\sadj{P})\rho^{-1}$ on $\Gamma_0^\infty(B^*\otimes \Omega_M)$. In particular, we have $(Pu)(f^{\sharp,1}) = u( \rho(\sadj{P})f^{\sharp})= u( (\adj{P}f)^{\sharp,1})= u( (Pf)^{\sharp,1})$ using formal symmetry of the \rfhgho\ $P$ and the identity $\sadj{P}f^\sharp = (\adj{P}f)^\sharp$, which is derived from the calculation
\begin{equation}
	\dlangle \sadj{P}f^\sharp, h\drangle = \dlangle f^\sharp,Ph\drangle = 
	\langle f, Ph\rangle = \langle \adj{P}f,h\rangle = \dlangle (\adj{P}f)^\sharp,h\drangle
\end{equation}
for arbitrary $f,h\in \Gamma^\infty (B)$ with compactly intersecting supports.  Thus both $W$ and $E_P$ become $P$-bisolutions, in the sense that
\begin{equation}
	(P\otimes 1) W = 0 = (1\otimes P)W, \qquad (P\otimes 1) E_P^\twopt = 0 = (1\otimes P)E_P^\twopt.
\end{equation} 
Similarly, the two-point distributions of the Green operators $E_P^\pm$ obey
\begin{equation}\label{eq:Epm2pteom}
	(P\otimes 1) (E_P^\pm)^\twopt = \id^\twopt = (1\otimes P)(E_P^\pm)^\twopt.
\end{equation}

\section{Fermionic Green-hyperbolic operators}\label{sec:fermionic}

\subsection{Fermionic \rfhgho's}

We now describe a class of \rfhgho's with additional structure that allows for a fermionic quantisation. A fermionic \rfhgho\ consists of a \rfhgho\ $(B,P)$ equipped with a smooth fibrewise bundle isomorphism $R:B\to B$ so that $R=\hadj{R}$ and $R\Cc=-\Cc R$ where $\Cc$ is the antilinear involution of $B$ (consequently also $\adj{R}=-R$) and $PR=RP$.  
Then $E_PR=RE_P$ and 
\begin{equation}\label{eq:fsymm}
	E_P(f,Rh) =  \langle f,E_P Rh\rangle  =    \langle f,RE_P h\rangle  =
	-  \langle Rf,E_P h\rangle   =-E_P(Rf,h) 
	=E_P(h,Rf),
\end{equation}	
so $(f,h)\mapsto \ii E_P(f,Rh)$ is a symmetric bilinear form on $\Gamma_0^\infty(B)$ that vanishes on $P\Gamma_0^\infty(B)$ and obeys
\begin{equation}\label{eq:fC}
\ii E_P(\Cc f,R\Cc h)=-\ii E_P(\Cc f,\Cc Rh)=
-\ii \overline{E_P(f,Rh)}=\overline{\ii E_P(f,Rh)}.
\end{equation}
One may check that $f\cdot h=\langle f,Rh\rangle$ defines an antisymmetric product on $\Gamma_0^\infty(B)$, making contact with the description of fermionic matter models in e.g.,~\cite{BeniniDappiaggi:2015,HackSchenkel:2013}.
We denote the fermionic \rfhgho\ by $(B,P, R)$ or simply $(P,R)$ or $P$ if there is no ambiguity. Fermionic \rfhgho's form the objects of a category $\fGreenHyp$ in which
a morphism $(P,R)\to (P',R')$ is a $\GreenHyp$-morphism
$S:P\to P'$ so that $S R=R'S$. 

In each fibre $B_x$, $R$ has real nonzero eigenvalues with respect to the hermitian bundle metric, and $\Cc$ maps any $\lambda$-eigenvector of $R$ to a $(-\lambda)$-eigenvector. The projection onto the positive spectrum determines a smooth decomposition $B=B_+\oplus B_-$ with respect to which $R$ and $\Cc$ take block forms
\begin{equation}\label{eq:frfhgho_RC}
	R = \begin{pmatrix} R_+ & 0 \\ 0 & R_-\end{pmatrix}, \qquad \Cc = \begin{pmatrix}
		0 & \Cc_{+-} \\ \Cc_{-+} & 0
	\end{pmatrix},
\end{equation}
with $\Cc_{+-}=\Cc_{-+}^{-1}$ and $R_-=-\Cc_{-+}R_+ \Cc_{+-}$. Although the hermitian pairing is diagonal with respect to the decomposition $B=B_+\oplus B_-$, the bilinear pairing derived from it using $\Cc$ is off-diagonal,
\begin{equation} \label{eq:frfhgho_P}
		\left\langle \begin{pmatrix}
			f\\ h
		\end{pmatrix},  \begin{pmatrix}
			f'\\ h'
		\end{pmatrix}\right\rangle =  (\Cc_{+-}h,f')+ (\Cc_{-+}f,h'). 
\end{equation}
The requirements that $PR=RP$ and $P\Cc=\Cc P$ imply that $P$ is also block diagonal, 
\begin{equation}
	P = \begin{pmatrix}
		P_+ & 0 \\ 0 & P_-
	\end{pmatrix}, \qquad P_- = \Cc_{-+} P_+ \Cc_{+-},
\end{equation}
and evidently $P_\pm$ are Green-hyperbolic. 
Thus, a fermionic \rfhgho\ can be seen as a doubled version of a formally hermitian, but not necessarily real, Green-hyperbolic operator $P_+$ on $B_+$, equipped with an operator
$R_+$ whose role is to ensure that $E_P(f,Rh)$ is a symmetric bilinear form.
 
\subsection{Quantisation} 
 
Any fermionic \rfhgho\ can be quantised to yield a bosonic field theory using $\Yf$. However it also has a more appropriate quantisation as a fermionic field theory, which we now describe.

The fermionic analogue of $\Sympl_\CC$ is the category $\Met_\CC$ of complexified real metric vector spaces, with objects $(V,\tau,C)$, where $V$ is a complex vector space, $\tau:V\times V\to\CC$ is a weakly nondegenerate symmetric bilinear form and $C$ is an antilinear involution such that $\tau(Cf,Ch)=\overline{\tau(f,h)}$. Here, `metric' 
is understood in the sense of a (possibly indefinite) metric tensor. A $\Met_\CC$-morphism between
$(V,\tau,C)$ and $(V',\tau',C')$ is a (necessarily injective) linear map $T:V\to V'$ such that $TC=C'T$ and $\tau'(Tv,Tw)=\tau(v,w)$ for $v,w\in V$. Evidently
every $(V,\tau,C)$ admits a sesquilinear form $(v,w)=\tau(Cv,w)$ which satisfies
$(Cv,Cw)=(w,v)= \overline{(v,w)}$; 
$\Met_\CC$ morphisms are isometries with respect to these forms.  There is an analogue of
Theorem~\ref{thm:sol}.
\begin{thm}
	There is a functor $\Sol_F:\fGreenHyp\to\Met_\CC$ so that 
	\begin{equation}
		\Sol_F(P,R,\Cc)=(\Sol(P), \ii\sigma_P\circ (1\otimes R),\Cc),
	\end{equation} 
	whose action on 
	morphisms has the same underlying map as $\Sol$, and maps Cauchy morphisms
	to $\Met_\CC$-isomorphisms. 
	For fermionic \rfhgho's $(P_j,R_j,\Cc_j)$ ($j=1,2$), one has $\Sol_F(P_1\oplus P_2,R_1\oplus R_2,\Cc_1\oplus \Cc_2)= \Sol_F(P_1,R_1,\Cc_1)\oplus \Sol_F(P_2,R_2,\Cc_2)$.
\end{thm}
\begin{proof}
	Define the bilinear form $\tau_{(P,R)}:\Sol(P)\times \Sol(P)\to \CC$ by  $\tau_{(P,R)}(E_P f,E_h)=E_{P}(f,Rh)$, which is well-defined because $P$ and $R$ commute and general properties of $E_P$.
	Then~\eqref{eq:fsymm} and~\eqref{eq:fC} show that $(\Sol(P),\tau_{(P,R)},\Cc)\in\Met_\CC$, with
	$\tau_{(P,R)}=\ii\sigma_P\circ (1\otimes R)$; $\tau_{(P,R)}$ inherits weak nondegeneracy from $\sigma_P$.
	Next,
	let $S:(P,R,\Cc)\to (P',R',\Cc')$ in $\fGreenHyp$. In particular, $S:(P,\Cc)\to (P',\Cc')$ 
	is a $\GreenHyp$ morphism and $SR=R'S$. Consequently, $S\Cc= \Cc'S$ and
	$\tau_{(P',R')}(Sv,Sw)=\ii\sigma_{P'}(Sv,R'Sw)= \ii\sigma_{P'}(Sv,SRw) =\ii \sigma_P(v,Rw)=\tau_{(P,R)}(v,w)$ holds for all
	$v,w\in\Sol(P)$, which shows that $S$ is the underlying map of a $\Met_\CC$-morphism 
	from $\Sol_F(P,R,\Cc)$ to $\Sol_F(P',R',\Cc')$. The remaining functorial properties and other statements follow directly from the analogous properties of $\Sol$.
\end{proof}
 
The fermionic quantised theory is described by graded $*$-algebras. Let $\grAlg$ be the category of graded $*$-algebras over $\ZZ_2=\{0,1\}$ with addition modulo $2$. That is, an object $\Ac\in\grAlg$ is an object of $\Alg$ with a preferred
direct sum decomposition $\Ac=\Ac_0\oplus\Ac_1$ so that the multiplication adds grades modulo $2$, $\Ac_j\Ac_k\subset \Ac_{j+k}$ for $j,k\in\ZZ_2$. A $\grAlg$-morphism $\alpha:\Ac_0\oplus \Ac_1\to\Ac'_0\oplus\Ac'_1$ is an $\Alg$-morphism so that $\alpha\Ac_j\subset \Ac'_j$ for $j\in\ZZ_2$. Elements of $A\in\Ac_j$ are called
homogeneous of grade $|A|=j$. The graded commutator is defined by 
\begin{equation}
	\bbLbrack A, B \bbRbrack =AB-(-1)^{|A|\,|B|}BA
\end{equation}
for homogeneous elements $A,B$, and extended by linearity in each slot.  
The analogue to Theorem~\ref{thm:RFHGHOquantisation} can now be stated.
\begin{thm}\label{thm:fRFHGHOquantisation}
	There is a functor $\Xf:\fGreenHyp\to \grAlg$ with the following properties:
\begin{enumerate}[a)]
	\item If $(P,R)\in\fGreenHyp$ has bundle $B$ and antilinear involution $\Cc$ then $\Xf(P)$ is the unital $*$-algebra generated by symbols $\Xi_P({f})$ ($f\in\Gamma_0^\infty(B))$, subject to 
	\begin{enumerate}[X1] 
		\item $f \mapsto \Xi_{P}({f})$ is complex linear  
		\item $\Xi_{P}({f})^* = \Xi_{P}(\Cc f)$  
		\item $\Xi_{P}({P} {f}) = 0$,   
		\item $\bbLbrack \Xi_{P}({f}),\Xi_{P}(h)\bbRbrack  = \ii E_{P}(f,R{h}) \1_{\Xf(P)}$  
	\end{enumerate}
	for all ${f},h\in\Gamma_0^\infty(B)$, and graded so that each $\Xi_P(f)$ is odd (whereupon the left-hand side of the equality in X4 is an anticommutator); 
	
	\item if $S:P\to P'$ is a $\fGreenHyp$-morphism, then the $\grAlg$-morphism $\Xf(S):\Xf(P)\to\Xf(P')$ is uniquely determined by
	\begin{equation}
		\Xf(S)\Xi_P({f}) = \Xi_{P'}(S{f}),\qquad {f}\in\Gamma_0^\infty(B),
	\end{equation}
	and in particular, $\Xf(S)=\Xf(T)$ if and only if $S\sim T$.
	
	\item  $\Xf(S)$ is an $\grAlg$-isomorphism if $S$ is a Cauchy $\fGreenHyp$-morphism.
	
	\item If $(B,P,R)$ and $(B',P',R')$ are fermionic \rfhgho's then their direct sum 
	$(B\oplus B',P\oplus P',R\oplus R')$ is a fermionic \rfhgho, and
	\begin{equation}
		\Xf(P\oplus P') \cong \Xf(P)\otimes \Xf(P')
	\end{equation}
	under the isomorphism fixed by 
	\begin{equation}
		\Xi_{P\oplus P'}(f\oplus f') \mapsto \Xi_P(f)\otimes \1_{\Xf(P')} + \1_{\Xf(P)}\otimes \Xi_{P'}(f'),
	\end{equation}
	where $\Xf(P)\otimes \Xf(P')$ is the graded tensor product, in which
	$(A\otimes A')(B\otimes B')= (-1)^{|A'|\,|B|} (AB)\otimes (A'B')$ and
	$(A\otimes B)^*=(-1)^{|A|\,|B|}A^*\otimes B^*$ for homogeneous $A,B\in\Xf(P)$, $A',B'\in\Xf(P')$.
	\end{enumerate}
\end{thm}
Explicitly, $\Xf=\Qf_a\circ \Sol_F$, where $\Qf_a:\Met_\CC\to\grAlg$ is a quantisation functor 
defined by analogy with $\Qf_s$ (but which we will not discuss directly). As a vector space,
\begin{equation}
\Xf(P) =  \bigoplus_{n\in\NN_0} \Sol_F(P,R)^{\wedge n},
\end{equation}
where $\wedge$ is the antisymmetrised tensor product, with $*$-operation so that
$(v_1\wedge\cdots\wedge v_n)^*=(\Cc v_n\wedge \cdots \Cc v_1)$ and we write $\1_{\Xf(P)}=(1,0,\ldots)$. Again adopting formulae from the functional formalism, the graded algebra product on $\Xf(P)$ is 
\begin{equation}
	FH = \wedge\circ \exp \left(\tfrac{1}{2}\langle \tau,\partial_L\otimes\partial_R\rangle\right) F\otimes H,
\end{equation} 
where the linear maps $\partial_L,\partial_R:\Xf(P)\to \Sol_F(P,R)\otimes \Xf(P)$ obey (note the numbering and ordering)
\begin{align}
	\partial_L (v_0 \wedge \cdots \wedge v_k) & = \sum_{j=0}^k (-1)^j v_j \otimes (v_0\wedge\cdots \slashed{v_j}  \cdots \wedge v_k), \nonumber \\
	\partial_R (v_k \wedge \cdots\wedge v_0) & = \sum_{j=0}^k (-1)^j v_j  \otimes (v_k\wedge\cdots \slashed{v_j}  \cdots \wedge v_0),
\end{align}
and $\partial_L \1_{\Xf(P)}=\partial_R\1_{\Xf(P)}=0$, while $\tau: (\Sol_F(P,R)\otimes\Xf(P))^{\otimes 2}\to 
\Xf(P)^{\otimes 2}$ is the linear map determined by
\begin{equation}
	\langle \tau,(v\otimes F) \otimes (w\otimes H) \rangle=  \tau(v,w) F\otimes H.
\end{equation} 
To obtain the correct gradings, one should note that  
\begin{equation}\label{eq:taudd}
	\langle \tau,\partial_L\otimes\partial_R\rangle:\Xf(P)_j\otimes\Xf(P)_k\to
	\Xf(P)_{j-1}\otimes\Xf(P)_{k-1}.
\end{equation}
In particular, writing $\Xi_P(f)=E_Pf$, one has 
\begin{equation}
	\Xi_P(f)\Xi_P(h) = \Xi_P(f)\wedge \Xi_P(h) + \frac{1}{2}\tau(E_Pf,E_Ph)\1_{\Xf(P)},\quad \{	\Xi_P(f),\Xi_P(h)\}=  \ii E_P(f,Rh)\1_{\Xf(P)}.
\end{equation} 
A similar argument to the one used in the bosonic case (starting with III \S 7.7 of~\cite{Bourbaki_algebra1}) establishes an isomorphism between $\Xf(P\oplus P')$ and the graded tensor product $\Xf(P)\otimes \Xf(P')$, with the stated action on generators.

It is not guaranteed that $\Xf(P)$ admits states or, equivalently, a completion as a $C^*$-algebra. Indeed a necessary condition for the existence of states is that the sesquilinear form
$(f,h)\mapsto \ii  E_P(\Cc f,Rh)$ determines an inner product on $\Gamma_0^\infty(B)/P\Gamma_0^\infty(B)$, in which case we say that 
$(P,R)$ has \emph{definite type}. The subset of definite-type $\fGreenHyp$ objects will be denoted $\fGreenHyp^+$.  

If $(P,R)$ is a fermionic \rfhgho, the $n$-point functions for states on $\Xf(P)$ are defined
as in~\eqref{eq:npointfns} but replacing $\Upsilon_P$ by $\Xi_P$. In particular, the two-point function obeys 
\begin{equation}
	W+\adj{W} = \ii E_P^\twopt\circ (1\otimes R^{\sharp,1}),
\end{equation}
due to the anticommutation relations.

As in the bosonic case, any functorial assignment of fermionic \rfhgho's results in a
functorial assignment of graded algebras. A fermionic functorial equation of motion 
is a functor $\fEoM:\Loc\to\fGreenHyp$, so that, writing $\fEoM(\Mb)=(\Bf(\Mb),P_\Mb,R_\Mb)$ for $\Mb\in\Loc$, the assignment $\EoM(\Mb)=(\Bf(\Mb),P_\Mb)$ defines a functorial equation of motion (see Definition~\ref{def:functorial_EoM}). Then one has a
functorial assignment of graded algebras $\Xf\circ\fEoM:\Loc\to \grAlg$. There is an obvious analogue of Theorem~\ref{thm:lcqft} with the change that the requirement of commutativity in the definition of Einstein causality is replaced by graded commutativity. 
To enlarge on this a little, note that the operator $R_\Mb$ does not change supports and therefore
	$E_{P_\Mb}(f,R_\Mb h)=0$ if $f,h\in\Gamma_0^\infty(\Bf(\Mb))$ have causally disjoint supports. Therefore the corresponding generators $\Xi_{P_\Mb}(f)$ and $\Xi_{P_\Mb}(h)$ anticommute, indicating that (as is expected for fermi fields) they are not observables. However, even-graded elements $A,B\in \Xf(P_\Mb)$ do commute if the generators appearing in $A$ and $B$ are smeared fields with support in causally disjoint regions $O_A$, $O_B$ in $\Mb$. In this way one sees that by taking the even part $\Xf(\fEoM(\Mb))_0$, we obtain a locally covariant QFT that obeys Einstein causality.  
	
More generally, one might consider functors from categories of globally hyperbolic spacetimes with additional structure such as a spin structure, and, suitably adapted, the above discussion will again yield a locally covariant QFT obeying graded Einstein causality, or Einstein causality on taking even parts. Details are left to the reader.

\subsection{Dirac type operators}\label{sec:Dirac-type}

Just as normally hyperbolic operators form a special class of \rfhgho's,
the operators of Dirac type have a special role in the fermionic theory.
Let $B$ be a finite rank complex vector bundle over $\Mb\in\Loc$ with a hermitian fibre metric but not necessarily a complex conjugation.
($\Mb$ need not admit a spin structure). A first-order formally hermitian partial differential operator
$D$ on $B$ is said to be of \emph{Dirac type} if its principal symbol $\sigma_D$ obeys
\begin{equation}
	\sigma_D(x,k)^2 = g^{-1}_x(k,k)\id_{B_x},
\end{equation}
which entails that $-D^2$ is normally hyperbolic. It
follows that $D=\hadj{D}$ is Green-hyperbolic, with Green operators $E_D^\pm = DE_{D^2}^\pm$. 
As $D$ is first order, its principal symbol is
\begin{equation}
	\sigma_D(x,k) = k_\mu (V_x)^\mu,
\end{equation}
where $V\in \Gamma^\infty(TM\otimes \End (B))$. Choosing a 
connection $\nabla^B$ on $B$ that is compatible with the hermitian form,
the operator $D$ may be written
\begin{equation}
	D = -\frac{\ii}{2} (V^\mu \nabla_\mu^B + \nabla_\mu^{TM\otimes B}  V^\mu) + W ,
\end{equation}
for a unique $W\in \Gamma^\infty(\End(B))$, where $\nabla^B$ is extended to $TM\otimes \End (B)$ compatibly with the Levi--Civita connection $\nabla$ of $\Mb$. While $V$ is fixed uniquely by $D$, the section $W$ depends on $D$ and the chosen connection $\nabla^B$.
%
A short calculation shows that formal hermiticity of $D$ is equivalent to $V=\hadj{V}$ and $W=\hadj{W}$: for $F,H\in\Gamma^\infty(B)$ with compactly intersecting supports, 
\begingroup
\allowdisplaybreaks
\begin{align}\label{eq:Dirac_type_hermiticity} 
	(DF,H) &= \frac{\ii}{2} \int_M  (V^\mu \nabla_\mu^B F + \nabla_\mu^{TM\otimes B}  V^\mu F, H)\dvol + 
	\int_M (WF,H)\,\dvol \nonumber \\
	&= \frac{\ii}{2} \int_M \left[  ( 2\nabla_\mu^{TM\otimes B}  V^\mu F, H) - ((\nabla_\mu^{TM\otimes B} V^\mu)F,H) \right]\dvol + 
	\int_M (F,\hadj{W}H)\,\dvol \nonumber \\
	&= \ii \int_M   \nabla_\mu (  V^\mu F, H) \dvol
	-\frac{\ii}{2} \int_M \left[ (F, \hadj{V}^\mu 2\nabla_\mu^{B} H)
	 + (F,(\nabla_\mu^{TM\otimes B} \hadj{V}^\mu)H) \right]\dvol + 
	\int_M (F,\hadj{W}H)\,\dvol
	\nonumber \\
	&= 
	-\frac{\ii}{2} \int_M   (F, \hadj{V}^\mu \nabla_\mu^{B} H + \nabla_\mu^{TM\otimes B}(\hadj{V}^\mu H))
 \dvol + 
	\int_M (F,\hadj{W}H)\,\dvol,
\end{align}
\endgroup
where we have used the divergence theorem and compact support of $(V^\mu F, H)$ to remove the first term in the penultimate line. Accordingly, 
\begin{equation}
	(DF,H)-(F,DH) = -\frac{\ii}{2} \int_M   (F, \widetilde{V}^\mu \nabla_\mu^{B} H + \nabla_\mu^{TM\otimes B}(\widetilde{V}^\mu H))
	\dvol + 
	\int_M (F,\widetilde{W}H)\,\dvol,
\end{equation}
where $\widetilde{V}=\hadj{V}-V$, $\widetilde{W}=\hadj{W}-W$.  Therefore
formal hermiticity of $D$ is equivalent to $V=\hadj{V}$ and $W=\hadj{W}$.


The hermitian fibre metric induces a unique antilinear isomorphism $B\mapsto B^*$, written $f\mapsto f^+$, so that $\dlangle f^+,h\drangle = (f,h)$. We use the same notation for the inverse of this map. Defining a hermitian fibre metric on $B^*$ so that $(u,u')=(u'{}^+, u^+)$, it is easily seen that $\sadj{D}$ is a Dirac type operator with symbol
$\sigma_{\sadj{D}}(x,k)=-\sadj{\sigma}_D(x,k)$.

Following~\cite{BaerGinoux:2012,IslamStrohmaier:2020}, we say that a Dirac-type operator $D$ has \emph{definite type} if there exists a smooth spacelike Cauchy surface $\Sigma$ for $\Mb$ so that
$(f,h)\mapsto (\sigma_D(n^\flat)f,h)$ is a positive definite sesquilinear form on $B|_\Sigma$ where $n$ is the future-pointing unit normal vector field on $\Sigma$. In this case we say that $\Sigma$ is a witness for the definite type.\footnote{At first sight, this definition is more general than that of B\"ar and Ginoux~\cite{BaerGinoux:2012}, whose definition would require the stronger requirement that $(f,h)\mapsto (\sigma_D(n^\flat)f,h)_{B_x}$ should define a positive definite form (i.e., an inner product) on $B_x$ for every $x\in M$ and all future-pointing timelike vectors $n\in T_xM$ -- see, however, Remark~\ref{rem:BG}.} A useful observation is that $D$ has definite-type if and only if $-\sadj{D}$ does, because $\sigma_{-\sadj{D}}(x,k) = k_\mu(\sadj{V}_x)^\mu$ and hence
\begin{equation}
	(\sigma_{-\sadj{D}}(n^\flat)u,u) = (n_\mu \sadj{V}^\mu u,u)=  
		\overline{(  n_\mu V^\mu u^+,u^+)}\ge 0, \qquad u\in\Gamma_0^\infty(B^*).
\end{equation}	 

The following appears as Lemma 4.1 of~\cite{IslamStrohmaier:2020} and is closely related to Lemma~3.17 of~\cite{BaerGinoux:2012} with antecedents in Lemma~4.3 of~\cite{Sanders_dirac:2010} and~\cite{Dimock:1982}. We give a 
slightly streamlined proof for completeness and to verify conventions.
\begin{lemma}\label{lem:Dirac_positivity}
	If Dirac-type operator $D$ has definite type, witnessed by smooth spacelike Cauchy surface $\Sigma$,
	then there is an inner product on $\Gamma_0^\infty(B)/D\Gamma_0^\infty(B)$ given by 
	\begin{equation}\label{eq:Dirac_positivity}
		([f],[h])\mapsto 
		\ii (f,E_D h)= 
		\int_\Sigma (\sigma_D(n^\flat) E_D f,E_D h)\dd A,
	\end{equation} 
	where $\dd A$ is the induced volume element on $\Sigma$.
\end{lemma}
\begin{rem}\label{rem:BG}
	The proof will show that the equality in~\eqref{eq:Dirac_positivity}
	holds for any smooth spacelike Cauchy surface in place of $\Sigma$.
	By considering the Cauchy formulation for $D$ in terms of the solution restricted to $\Sigma$, and arbitrary data on arbitrary smooth spacelike Cauchy surfaces, it becomes clear that
	definite type in our sense implies (and is hence equivalent to) definite type in the sense of 
	B\"ar and Ginoux~\cite{BaerGinoux:2012}. 
\end{rem}
\begin{proof}[Proof of Lemma~\ref{lem:Dirac_positivity}]
The map in~\eqref{eq:Dirac_positivity} gives a well-defined sesquilinear
form on $\Gamma_0^\infty(B)/D\Gamma_0^\infty(B)$, and the last expression
is a well-defined inner product on this space because $D$ has definite type witnessed by $\Sigma$ and $E_D f$ vanishes identically on $\Sigma$ if and only if $f\in D\Gamma_0^\infty(B)$.

It remains to establish the equality in~\eqref{eq:Dirac_positivity}.
Let $[f],[h]\in \Gamma_0^\infty(B)/D\Gamma_0^\infty(B)$.
Using the timeslice property for $D$ we may consider representatives $f,h\in\Gamma_0^\infty(B)$ with
$\supp f\subset I^-(\Sigma)$, $\supp h\subset I^+(\Sigma)$. Thus
\begin{align}
\int_M (f,E_Dh) \dvol &= \int_{J^-(\Sigma)}   (f,E_D h) \dvol  = \int_{J^-(\Sigma)} (DE_D^+ f,E_D h) \dvol \nonumber \\
&= \ii \int_\Sigma n_\mu (V^\mu E_D^+ f,E_D h)\dd A + 
 \int_{J^-(\Sigma)} (E_D^+ f,DE_D h) \dvol,
\end{align}
by essentially the same calculation used in~\eqref{eq:Dirac_type_hermiticity} to establish formal hermiticity of $D$, 
except that there is now a boundary term at $\Sigma=\partial J^-(\Sigma)$. The second term vanishes and $E_Df=-E_D^+f$ in a neighbourhood of $\Sigma$, and 
the equation in~\eqref{eq:Dirac_positivity} follows. 	
\end{proof}	

We now show that every Dirac-type operator on a bundle $B_+$  
determines a fermionic \rfhgho. Setting $B_-=B^*$, we endow $B=B_+\oplus B_-$ with the direct sum fibre metric and a complex conjugation
\begin{equation}
	\Cc\begin{pmatrix}
		f \\ u
	\end{pmatrix} = \begin{pmatrix} u^+  \\ f^+ \end{pmatrix}, \qquad f\in \Gamma^\infty(B_+),~ u\in \Gamma^\infty(B^*_-),
\end{equation}
which makes $B$ a hermitian vector bundle. Then set $R_+=\id_{B_+}$, $R_-=-\id_{B_-}$, $P_+=D$ and $P_-=\sadj{D}$ and define a \rfhgho\ $P= D\oplus \sadj{D}$ by~\eqref{eq:frfhgho_RC} and~\eqref{eq:frfhgho_P}.
In particular, one has
\begin{equation}\label{eq:EP_for_Diractype}
	E_P = \begin{pmatrix}
		-D & 0 \\ 0 & -\sadj{D}
	\end{pmatrix}
	\begin{pmatrix}
		E_{-{D}^2} & 0 \\ 0 & E_{-\sadj{D}^2}
	\end{pmatrix}.
\end{equation}
Furthermore, if $D$ has definite type then
$(P,R)\in \fGlobHypGreen_+$, i.e., $(P,R)$ has definite type as a fermionic \rfhgho. To see this, we calculate
\begin{align}
	\ii E_P\left(\Cc\begin{pmatrix} f\\ u\end{pmatrix},
	R\begin{pmatrix} f'\\ u'\end{pmatrix}
	\right) &= 
	\ii \left\langle \Cc\begin{pmatrix} f\\ u\end{pmatrix},
	E_P\begin{pmatrix} f'\\ -u'\end{pmatrix}\right\rangle = 
	\ii \left( 
	 \begin{pmatrix} f\\ u\end{pmatrix},
	\begin{pmatrix} E_Df'\\ -E_{\sadj{D}}u'\end{pmatrix}
	\right)\nonumber \\
	& = 
	\ii (f, E_{D} f') + \ii(u, E_{-\sadj{D}} u'),
\end{align}
and note that the right-hand side determines an inner product
on $\Gamma_0^\infty(B)/P\Gamma_0^\infty(B)$. 

In the corresponding quantised theory $\Xf(P)$, we may introduce fields
$\Psi_D$ and $\Psi^+_D$ so that
\begin{equation}
	\Xi_P\begin{pmatrix} f \\ u \end{pmatrix}= \Psi_D(u) + \Psi^+_D(f),\qquad f\in\Gamma_0^\infty(B_+),~ u\in\Gamma_0^\infty(B_+^*).
\end{equation} 
These operators satisfy $\Psi_D(f)^*=\Psi^+_D(f^+)$ for all $f\in\Gamma_0^\infty(B_+)$, the field equations
\begin{equation}
	\Psi_D(\sadj{D}u) = 0 = \Psi^+_D(D f)
\end{equation}
for all $f\in\Gamma_0^\infty(B_+)$, $u\in\Gamma_0^\infty(B_+^*)$
(the formalisation of equations $D\Psi_D=0=\sadj{D}{\Psi}^+_D$) 
and anticommutation relations
\begin{equation}
	\{\Psi_D(u),\Psi_D(v)\} = 0,\quad 
	\{\Psi_D(u),\Psi^+_D(f)\} = \ii \dlangle u,E_D f\drangle\1_{\Xf(P)}, 
	\quad f\in\Gamma_0^\infty(B_+),~ u,v\in\Gamma_0^\infty(B_+^*).
\end{equation}
Note that $\Psi_D$ and $\Psi_D^+$, which may be regarded as sections of $B_+$ and $B_+^*$, are smeared with test sections in $B_+^*$ and $B_+$. This contrasts with the convention adopted in Section~\ref{sec:Yf_quant}, where we smeared fields with test sections of the same bundle, taking advantage of the bilinear pairing.

A particular example of a Dirac type operator is provided by the usual Dirac equation, 
which we briefly describe, referring to~\cite{Sanders_dirac:2010} for all details. Suppose $\Mb\in\Loc$ is equipped with a spin structure and let $DM$ and
$D^*M$ be the bundles of spinors and cospinors over $M$. 
The Dirac adjoint is an antilinear map $u\mapsto u^+$ between $DM$ and $D^*M$, with an inverse denoted by the same symbol. The (indefinite) hermitian bundle metric on $DM$ is defined as $(u,u')=\dlangle u^+,u'\drangle$; similarly the hermitian metric on $D^*M$ is $(v,v')=\dlangle v^+,v'\drangle$, where
as usual the double angle bracket is the duality pairing for the bundle in question. 
Then the spinorial Dirac operator is defined by 
\begin{equation}
	D = -\ii\slashed \nabla + m
\end{equation}
for mass $m\ge 0$, with respect to a suitable Dirac spinor-tensor field
(i.e., $\gamma$-matrices) and is of definite type. Moreover the dual $\adj{D}$ is the cospinorial Dirac operator
\begin{equation}
	\sadj{D} = \ii\slashed \nabla + m
\end{equation}
and the theory is then quantised as described above.

\section{Generalised Hadamard States}\label{sec:Hadamard}

The realisation that the Hadamard condition
for states of the scalar field could be expressed microlocally was a major breakthrough~\cite{Radzikowski_ulocal1996}, which has been achieved for other free fields including the Dirac~\cite{Hollands:2001,Sanders_dirac:2010},  Proca~\cite{Few&Pfen03,MorettiMurroVolpe:2023} and electromagnetic fields~\cite{Few&Pfen03}, and linearised gravity~\cite{Gerard:2023}. One of the main purposes of this paper is to provide an analogue for a class of \rfhgho's. We begin by recalling the definition of the wavefront set and the microlocal form of the Hadamard condition.

\subsection{Wavefront sets}\label{sec:WF}

As standing notation: for any manifold $X$, $\dot{T}^*X$  denotes the cotangent bundle of $X$ with its zero section $0_{T^*X}$ removed; meanwhile, for any subset $\Uc\subset T^*X$, we set $\Uc_0:=\Uc\cup 0_{T^*X}$.

Let $u\in\DD'(M)=\Gamma_0^\infty(\Omega_M)'$ be a distribution and let $(y,l)\in \dot{T}^*M$.
Choose coordinates $X^\alpha$ in a neighbourhood of $y$ and
define smooth functions $\ee_k$ ($k\in T^*_yM$) on the chart domain by $\ee_k(x)=\exp(\ii   X^\alpha(x) k(\partial_{X^\alpha}))$, where $\partial_{X^\alpha}$ are the coordinate basis vector fields. 
Then $(y,l)$ is said to be a 
regular directed point of $u$ if there is a smooth density $\chi$ compactly supported in the coordinate chart, and an open cone $V$ in $\dot{T}_y^*M$ such that $\chi(y)\neq 0$, $l\in V$, and
\begin{equation}\label{eq:regdirndef}
	\sup_{k\in V} (1+\|k\|^N)  |u(\chi e_k)| <\infty
\end{equation}
for all $N\in\NN$, where $\|\cdot\|$ is any norm on $T^*_yM$. 
The \emph{wavefront set} $\WF(u)$ is the set of all $(y,l)\in\dot{T}^*M$ that are not regular directed points of $u$; it is independent of the coordinates used~\cite{Hormander1}.  
If $B$ is a complex finite-rank bundle and $u\in \DD'(B)=\Gamma_0^\infty(B^*\otimes\Omega)'$ then $\WF(u)$ may be defined so that
\begin{equation}
	\WF(u)=\bigcup_{S\in \Gamma^\infty(B^*)}\WF(f\mapsto u(fS));
\end{equation}
near any $y\in M$, it enough to take the union over sections $S$ forming a local basis for $B^*$. 
Many standard results for scalar distributions then extend to the bundle context fairly directly; for instance, 
\begin{equation}\label{eq:WFbounds}
\WF(Pu)\subset \WF(u)\subset \WF(Pu)\cup (\WF(u)\cap\Char P)
\end{equation}
for any $u\in\DD'(B)$ and partial differential operator $P:\Gamma^\infty(B)\to\Gamma^\infty(\tilde{B})$,
 where the \emph{characteristic set} $\Char P\subset \dot{T}^*M$ consists of those $(x,k)$ at which the principal symbol $p(x,k)$ of $P$ has nontrivial kernel. The first inclusion follows easily from the scalar result (see (8.1.11) in~\cite{Hormander1}) and the definition of the wavefront set for bundles, while the second can be obtained by polarisation set methods (see e.g.~Section~3 of~\cite{Fewster:2025b}).
If $\Char P$ is empty, then $\WF(Pu)=\WF(u)$, which in particular implies that 
two-point and kernel distributions for a continuous linear map $T:\Gamma_0^\infty(B)\to \Gamma^\infty(B)$ satisfy
\begin{equation}
	\WF(T^\twopt)=\WF(T^\knl) = \WF((T^{\sharp,\alpha})^\knl)=\WF((T^{\natural,\alpha})^\knl),
\end{equation}
as a result of the identity~\eqref{eq:T2ptTknl}. More examples and exposition of wavefront sets can be found in~\cite{Hormander1,BrouderDangHelein:2014,Strohmaier:2009} and many other references. 
 
Let $B$ be a finite-rank complex vector bundle over a globally hyperbolic spacetime $\Mb\in\Loc$. 
We denote the bundle of nonzero null future/past covectors by $\Nc^\pm$ and also write $\Nc=\Nc^+\cup\Nc^-$ for the bundle of all nonzero null covectors.  If
$\gamma$ is any affinely parametrised null geodesic in $M$, then the curve $c(\lambda)=(\gamma(\lambda),\dot{\gamma}(\lambda)^\flat)$ 
in $T^*M$ is a \emph{bicharacteristic strip} for any normally hyperbolic operator $P$, while $\gamma$ is also called a \emph{bicharacteristic curve}. 
For $(x,k),(y,\ell)\in \Nc$, we write $(x,k)\sim (y,\ell)$ if and only if
the points lie on a common bicharacteristic strip. That is, $(x,k)\sim (y,\ell)$ holds if and only if there is a null geodesic segment $\gamma$ connecting $x$ and $y$, with tangent vectors $k^\sharp$ and $\ell^\sharp$ at the endpoints, which are related by parallel transport along $\gamma$ (understanding $(x,k)\sim (x,\ell)$ if and only if $k=\ell$ is null); in this case we say that the geodesic segment $\gamma$ \emph{witnesses} to the relation $(x,k)\sim (y,\ell)$. With these definitions, set
\begin{align}\label{eq:Rcsets}
	\Rc&= \{(x,k;x',-k')\in \Nc\times\Nc: (x,k)\sim (x',k')\}\\ 
	\Rc^\pm & =\{(x,k;x',-k')\in\Rc:   x\in J^\pm(x') \}.
\end{align}
The Green operators of normally hyperbolic operators have known wavefront sets.
\begin{thm}\label{thm:WFEPs}
	For a normally hyperbolic \rfhgho\ $P$, the wavefront sets of $E_P^\pm$ and $E_P$ are
	\begin{align}
		\WF(E_P^\pm)&= \Rc^\pm\cup\WF(\id^\knl),
		\label{eq:WFEPpm}\\
		\label{eq:WFEP}
		\WF(E_P)&= \Rc ,
	\end{align}
where the wavefront set of $\id^\knl$ is given explicitly by
\begin{equation}\label{eq:WFidknl}
	\WF(\id^\knl) = \{(x,k;x,-k)\in \dot{T}^*(M\times M): (x,k)\in \dot{T}^*M\}.
\end{equation} 
\end{thm} 
This result generalises the scalar case~\cite{DuiHoer_FIOii:1972,Radzikowski_ulocal1996} and is 
proved for bundles as Theorem~A.5 in~\cite{Sanders_dirac:2010} based on a scaling limit structure of the Green operators in terms of Riesz distributions.  An alternative proof, making use of the polarisation set and avoiding Riesz distributions, is given as
Theorem 1.1(c) of~\cite{Fewster:2025b}. 
Closely related statements for Feynman parametrices can be found in~\cite{IslamStrohmaier:2020}. 
The wavefront set~\eqref{eq:WFidknl} is standard; it follows, for example, from the calculation of the polarisation set in section~3 of~\cite{Fewster:2025b}.

\subsection{Decomposable \rfhgho's and the Hadamard condition}
\label{sec:Hadcond}

Let $P$ be a normally hyperbolic \rfhgho\ and suppose $\omega$ is a state on $\Yf_P(M)$ with distributional
$n$-point functions. Then~\cite{Radzikowski_ulocal1996,SahlmannVerch:2000RMP} $\omega$ is said to be Hadamard
if its two-point function $W\in\DD'(B\boxtimes B)$ satisfies the wavefront set condition~\eqref{eq:Had_WF}, which we now write as
\begin{equation}\label{eq:RcHad}
	\WF(W)=\Rc^\Had:=\{(x,k;x',k')\in \Rc: k\in\Nc^+\},
\end{equation}
which generalises the condition for the scalar field~\cite{Radzikowski_ulocal1996}. 

When $P$ is Green-hyperbolic but not normally hyperbolic, we face the complication that Green-hyperbolicity does not constrain the principal symbol to be hyperbolic with null characteristics. For example, the support properties of Green operators only imply that
$\WF(E^\pm)\subset \{(x,k;y,l)\in \dot{T}^*(M\times M): x\in J^\pm(y)\}$. If $g'$ is a `slow metric' whose cone of future-directed causal vectors is contained inside that of $g$, an operator that is normally hyperbolic with respect to $g'$ is Green-hyperbolic but its bicharacteristic curves can be $g$-timelike and the characteristic covectors can be $g$-spacelike.\footnote{If $\omega$ is a positive frequency $g$-causal covector then $\omega(v')\ge 0$ for every vector $v'$ that is future-directed $g'$-causal, and hence also future-directed $g$-causal; we deduce that $\omega$ is also a positive frequency $g'$-causal covector. Thus the cone of (future-directed) $g$-causal covectors is contained in the cone of (future-directed) $g'$-causal covectors.}
An extreme example is provided by the Proca operator $P=-\delta\dd +m^2$ on $1$-forms (see Section~\ref{sec:Proca}), which has $\Char P=\dot{T}^*M$ because its principal symbol acts as $p(x,k)v=g^{-1}(k,v)k - g^{-1}(k,k)v$ 
and clearly has $k\in\ker p(x,k)$ for all $(x,k)\in \dot{T}^*M$. Nonetheless, the Proca operator is Green-hyperbolic.

We will now propose a definition of Hadamard states for \rfhgho's that is sufficiently broad to include the normally hyperbolic case, models such as the Proca field, operators that are normally hyperbolic with respect to slow metrics, and operators formed as direct sums of the above, including situations with multiple characteristic cones that can arise in descriptions of birefringence~\cite{FewsterPfiefferSiemssen:2018}. Hadamard states for the Proca model itself have been studied in~\cite{Few&Pfen03}, and more recently in~\cite{MorettiMurroVolpe:2023} (however, as already mentioned, the latter reference contains a gap owing to an erroneous argument). We will explain the relationship of our definition to these references in Section~\ref{sec:Proca}.
We will show that our generalisation has properties that would be hoped for, and is equipped with a set of tools that are necessary for applications in quantum measurement.  

A starting-point is the observation by Sahlmann and Verch~\cite{SahlmannVerch:2000RMP} that the condition~\eqref{eq:RcHad} can be simplified to
\begin{equation}\label{eq:Hadamard}
	\WF(W)\subset \mathcal{N}^+\times \mathcal{N}^-,
\end{equation}  
because~\eqref{eq:Hadamard} implies~\eqref{eq:RcHad} on using the CCRs $W-\adj{W}=\ii E_P$ and the fact that $\WF(E_P)=\Rc\subset (\Nc^+\times\Nc^-)\cup (\Nc^-\times\Nc^+)$. Abstracting from this example, we now make the following definition. 
\begin{defn}\label{def:Hadamard}
	Let $\Vc^+\subset \dot{T}^*M$ be a conic, relatively closed subset, so that $\mathcal{V}^+\cap \Vc^-=\emptyset$, where $\Vc^-=-\Vc^+$. A semi-Green-hyperbolic operator $P$ on $\Gamma^\infty(B)$ is \emph{$\Vc^\pm$-decomposable} if 
	\begin{equation}
		\WF(E_P) \subset (\Vc^+\times \Vc^-) \cup (\Vc^-\times \Vc^+).
	\end{equation}
	If a \rfhgho\ (resp., fermionic \rfhgho) $P$ on $\Gamma^\infty(B)$ is $\Vc^\pm$-decomposable
	then a state with distributional $n$-point functions on $\Yf(P)$ (resp., $\Xf(P)$) is  \emph{$\Vc^+$-Hadamard} if 
	its distributional two-point function $W\in\DD'(B\boxtimes B)$ satisfies
	\begin{equation}\label{eq:VHadamard}
		\WF(W)\subset \Vc^+\times \Vc^-.
	\end{equation}   
\end{defn}

Note that if $P$ is $\Vc_1^\pm$-decomposable then it is also $\Vc_2^\pm$-decomposable if
$\Vc_2^+$ is a conic relatively closed set with $\Vc_2^+\cap\Vc_2^-=\emptyset$ and $\Vc_1^+\subset\Vc_2^+$. Similarly, if $P$ is both $\Vc_1^\pm$- and $\Vc_2^\pm$-decomposable, then it is also $(\Vc_1^\pm\cap\Vc_2^\pm)$-decomposable. In particular, 
if $P$ is a \rfhgho\ then $E_P$ is a $P$-bisolution, and one has $\WF(E_P^\twopt)\subset \Char_0 P\times \Char_0 P$. Thus, if $P$ is $\Vc^\pm$-decomposable, it is also $(\Vc^\pm\cap\Char P)$-decomposable, using $\Vc^\pm\cap 0=\emptyset$.
Also note that  the definition of $\Vc^+$-Hadamard states is empty in the fermionic case if $P\notin\fGreenHyp^+$, owing to the nonexistence of states on $\Xf(P)$. 

We can give an immediate example.
\begin{thm}\label{thm:Hadamard_norm_hyp}
	Any normally hyperbolic \rfhgho\ $(B,P)$ is $\Nc^\pm$-decomposable. If the hermitian fibre metric on $B$ is positive definite then $(B,P)$ admits $\Nc^+$-Hadamard states.
\end{thm}
\begin{proof}
	Decomposability follows by Theorem~\ref{thm:WFEPs} because
	$\Rc\subset (\Nc^+\times\Nc^-)\times(\Nc^-\times\Nc^+)$, whereupon the $\Nc^+$-Hadamard condition coincides with the standard microlocal definition of Hadamard states. The second part holds by 
	Corollary 5 of~\cite{FewsterStrohmaier:2025}, building on the analysis of Feynman propagators in~\cite{IslamStrohmaier:2020}.
\end{proof}	
A corollary is that \rfhgho's that are normally hyperbolic with respect to a `slow' metric are $\Vc^\pm$-decomposable and admit $\Vc^+$-Hadamard states if the fibre metric is positive definite, where $\Vc^+$ is the cone of future-directed null covectors determined by the slow metric (which is not necessarily contained within the cone of future-directed causal covectors).
A similar result for Dirac-type operators will be given as Theorem~\ref{thm:Hadamard_Dirac_type} below.

We now begin to develop the theory of Hadamard states in this generalised form.
Our first result establishes the basic properties that are familiar from the standard theory of Hadamard states and justify our definition as an appropriate generalisation thereof. Its proof slightly abstracts and adds to arguments in Propositions 4.6 and 4.7 of~\cite{MorettiMurroVolpe:2023}.  
\begin{thm}\label{thm:Hadamard_props}
	Suppose $P$ is a $\Vc^\pm$-decomposable \rfhgho\ (resp., fermionic \rfhgho), where $\Vc^\pm$ are as in Definition~\ref{def:Hadamard}.
	\begin{enumerate}[(a)] 
	\item If $W$ is the $2$-point function of a $\Vc^+$-Hadamard state on the bosonic algebra $\Yf(P)$ (resp., the fermionic algebra $\Xf(P)$) then
	\begin{equation}
		\WF(W) = (\Vc^+\times \Vc^-)\cap \WF(E_P).
	\end{equation} 
	\item If $W_1$ and $W_2$ are $2$-point functions of $\Vc^+$-Hadamard states
	on $\Yf(P)$ (resp., $\Xf(P)$) then $W_1-W_2$ is smooth. 
	
	\item If $U$ is an open Cauchy slab and $\omega$ is a state on $\Yf(P)$ (resp., $\Xf(P)$) whose restriction to $\Yf(P|_U)$ (resp., $\Xf(P|_U)$) is $\Vc^+|_U$-Hadamard, then $\omega$ is $\Vc^+$-Hadamard.  
	
	\item The $1$-point functions of a $\Vc^+$-Hadamard state are smooth.  
	
	\item If $\Yf(P)$ (resp., $\Xf(P)$) admits a $\Vc^+$-Hadamard state then it admits a quasifree $\Vc^+$-Hadamard state.
	\end{enumerate} 
\end{thm}
Part~(a) has a refined version using polarisation sets, namely
\begin{equation}
	\WF_\pol((\rho^{1/2}\otimes\rho^{1/2})W^\twopt) = \WF_\pol(E_{\rho^{1/2}P\rho^{-1/2}}^\knl)|_{\Vc^+\times\Vc^-}\cup 0,
\end{equation}
where the half-density factors $\rho^{1/2}$ arise because of the way that the polarisation set is defined. See Theorem~\ref{thm:Hadpropsaprime}. Regarding part~(e), the definition of a quasifree state is recalled in Section~\ref{sec:truncated}.
\begin{proof} 
	(a) In the bosonic case, we have $W-\adj{W}=\ii E_P$ and $\WF(\adj{W})\subset \Vc^-\times\Vc^+$. Then
	\begin{equation}
		\WF(W)=\WF(\ii E_P+\adj{W})\subset (\WF(E_P)\cup (\Vc^-\times \Vc^+))\cap \WF(W)\subset  \WF(E_P)\cap (\Vc^+\times\Vc^-)
	\end{equation}
	and also
	\begin{align}
		\WF(E_P)\cap (\Vc^+\times\Vc^-) &= \WF(-\ii (W-\adj{W}))\cap (\Vc^+\times\Vc^-) \nonumber \\
		& \subset (\WF(W)\cup \WF(\adj{W}))\cap (\Vc^+\times\Vc^-) = \WF(W).
	\end{align}
	For the fermionic case, we use $W+\adj{W}=\ii E_P\circ (1\otimes R)$ and 
	$\WF(E_P\circ (1\otimes R))=\WF(E_P)$, and otherwise apply the same argument.
	
	(b) For $j=1,2$, write $W_j^+=W_j$ and define $W_j^-=\adj{W}_j$. Then
	$\WF(W_1^\pm - W_2^\pm) \subset \Vc^\pm\times \Vc^\mp$. However,
	$W_j^+-W_j^-=\ii E_P$ and therefore $W_1^+-W_2^+=W_1^--W_2^-$, which implies that
	$\WF(W_1-W_2)\subset \WF(W_1^+-W_2^+)\cap \WF(W_1^--W_2^-)\subset (\Vc^+\times \Vc^- )\times(\Vc^-\cap \Vc^+)=\emptyset$, so $W_1-W_2$ is smooth.	
	For the fermionic case, one has $W_1^+-W_2^+=-(W_1^--W_2^-)$ and the argument goes through.
	
	(e) Any state on $\Yf(P)$ or $\Xf(P)$ has a `liberation', namely the quasifree state  with the same two-point function -- see e.g.,~\cite{Kay:1993}.
	
	We defer the rest of the proof: (c) is a special case of a more general result on pull-backs and push-forwards of states by $\GreenHyp$-morphisms, Theorem~\ref{thm:pullbacks} below. Part~(d) 
	will be proved as Corollary~\ref{cor:smooth1pt}.
\end{proof}
Theorem~\ref{thm:Hadamard_props}(b) permits the construction of normal-ordered quadratic operators such as ${:}\Upsilon_P^2(f){:}$ by point-splitting and regularisation relative to a reference $\Vc^+$-Hadamard state. Meanwhile, part~(c) asserts the preservation of Hadamard form, originally proved for scalar fields by PDE methods in~\cite{FullingSweenyWald:1978}. We will use this property in a deformation argument demonstrating the existence of $\Vc^+$-Hadamard states in one spacetime given their existence in others, in the context of theories specified by functorial equations of motion (see Corollary~\ref{cor:deformation} and Corollary~\ref{cor:bultrastatic}).
As will be seen in section~\ref{sec:truncated}, the $\Vc^+$-Hadamard condition implies constraints on the higher $n$-point functions, as was shown in~\cite{Sanders:2010} for the scalar field. Specifically, all the truncated $n$-point functions for $n\neq 2$ are smooth, and a bound can be given on $\WF(W_n)$
(Corollary~\ref{cor:truncated}).

To make Theorem~\ref{thm:Hadamard_props}(a) more explicit, it is necessary to compute $\WF(E_P)$ for a $\Vc^\pm$-decomposable \rfhgho\ $P$. As the example of a normally hyperbolic operator with respect to a `slow metric' shows, there is no simple general formula for $\WF(E_P)$ beyond the general statement that
\begin{equation}
	\WF(E_P)\subset ((\Vc^+\times \Vc^-)\times(\Vc^-\times \Vc^+))\cap (\Char P\times \Char P)
\end{equation}
already mentioned. However, in some interesting models, the operator $E_P$ takes the form $E_P=QE_KR$, where $K$ is normally hyperbolic and $Q$, $R$ are (pseudo)differential operators. 
This occurs for the Proca equation, as will be seen in Section~\ref{sec:Proca}, and also the Dirac equation.
If both $Q$ and $R$ are elliptic, or, more generally, noncharacteristic at $\Nc$, then the standard wavefront set calculus gives
\begin{equation}
	\WF(E_P) = \WF(E_K) = \Rc;
\end{equation}
otherwise, one can only conclude that $\WF(E_P)\subset \WF(E_K)=\Rc$ -- this situation occurs for both the Dirac and Proca fields. The following results can be applied in such circumstances. We first recall that -- like any normally hyperbolic operator -- $K$ can be put into the standard form
\begin{equation}
	K = \Box^B+V, \qquad \Box^B:=g^{\mu\nu}\nabla_\mu^{T^*M\otimes B}\nabla_\nu^B ,
\end{equation}
where $\nabla^B$ is the \emph{Weitzenb\"ock connection} $\nabla^B$ on $B$,
 $\nabla^{T^*M\otimes B} = \nabla\otimes 1+1\otimes\nabla^B$ is its natural extension to a connection on $T^*M\otimes B$, with $\nabla$ the Levi-Civita derivative as before, and $V\in \Gamma^\infty(\End(B))$. See e.g., Proposition~3.1 of~\cite{BaumKath:1996}, Lemma~1.5.5 of~\cite{BarGinouxPfaffle} or section 1 of~\cite{IslamStrohmaier:2020}.  
\begin{thm}\label{thm:WFQEKR}
	Let $K$ be normally hyperbolic on $\Gamma^\infty(B)$ with associated Weitzenb\"ock connection $\nabla^B$.
	For each $(x,k),(x',k')\in T^*M$ with $(x,k)\sim (x',k')$ and witnessing null geodesic segment $\gamma$, let  
	$\Pi_{x,k}^{x',k'}: B_{x'}\to B_{x}$ be the $\nabla^B$-parallel transport operator from $x'$ to $x$ along $\gamma$. Let $Q:\Gamma^\infty(B)\to\Gamma^\infty(\hat{B})$ and $R:\Gamma^\infty(\tilde{B})\to \Gamma^\infty(B)$ 
	be partial differential operators of finite order with principal symbols $q$ and $r$ respectively,
	where $\hat{B}\to M$ and
	$\tilde{B}\to M$ are finite-rank complex vector bundles. 
	
	If for all $(x,k;x',-k')\in \Rc$ one has
	\begin{equation}\label{eq:WFthm}
		q(x,k)\circ \Pi_{x,k}^{x',k'}\circ r(x',k')\neq 0
	\end{equation}
	then $\WF(QE_{K}R )=\Rc$. This holds in particular if
	$Q=\id$ and $r$ is nonvanishing on $\Nc$, or if $R=\id$ and $q$ is nonvanishing on $\Nc$
	(however $r$ resp., $q$ could have nontrivial kernel on $\Nc$).
\end{thm}
\begin{proof}
	This is essentially proved as Corollary~1.2 of~\cite{Fewster:2025b}, following a computation of the polarisation set~\cite{Dencker:1982} of kernel distributions of Green operators for normally hyperbolic equations. More precisely, under the hypotheses stated here,
	Corollary~1.2 of~\cite{Fewster:2025b}
	 shows that 
	$\WF(\tilde{Q} E_{\tilde{K}} \tilde{R}) = \Rc$, 
	where $\tilde{Z} = \rho^{1/2}Z\rho^{-1/2}$ for any operator $Z$, where $\rho=(-g)^{1/2}$ is the metric density (ref.~\cite{Fewster:2025b} also allows $Q$ and $R$ to be pseudodifferential operators). Because $\tilde{Q} E_{\tilde{K}} \tilde{R} = \widetilde{QE_KR}$ and 
	the wavefront set is unchanged under multiplication by nonvanishing smooth factors, the required result follows.	 
\end{proof}
    
\begin{cor}\label{cor:WFQEKR}
	With the notation and assumptions of Theorem~\ref{thm:WFQEKR}, if a \rfhgho\ (resp., fermionic \rfhgho) $P$ has Green operator
	$E_P=QE_KR$, then $P$ is $\Nc^\pm$-decomposable. If $W$ is the two-point function of a $\Nc^+$-Hadamard state on $\Yf(P)$ (resp., $\Xf(P)$) then its wavefront set is
	\begin{equation}
		\WF(W) = \Rc^\Had.
	\end{equation}
\end{cor}
\begin{proof}
	Combine Theorem~\ref{thm:Hadamard_props}(a) and Theorem~\ref{thm:WFQEKR}.
\end{proof}

We obtain the existence of Hadamard states for definite-type Dirac-type operators as a companion result to Theorem~\ref{thm:Hadamard_norm_hyp}.
\begin{thm}\label{thm:Hadamard_Dirac_type}
	If $D$ is a Dirac-type operator then its
	corresponding fermionic \rfhgho\ $(P,R)$ is $\Nc^\pm$-decomposable;
	if $D$ has definite type then $\Xf(P)$ admits $\Nc^+$-Hadamard states. 
\end{thm}
\begin{proof} 
	$\Nc^\pm$-decomposability follows from Corollary~\ref{cor:WFQEKR} and formula~\eqref{eq:EP_for_Diractype}	for $E_P$, and the
	statement on the existence of Hadamard states is proved as 
	Corollary 11 of~\cite{FewsterStrohmaier:2025}, making use of 
	the Feynman propagators obtained in~\cite{FewsterStrohmaier:2025}. 
\end{proof} 

As a concrete example, the Green operator for the `doubled' Dirac operator
\begin{equation}
	P = \begin{pmatrix}
		-\ii\slashed{\nabla}+m & 0 \\ 0 & 	\ii\slashed{\nabla}+m
	\end{pmatrix}
\end{equation}
in four spacetime dimensions is given by $E_P=Q E_K$, where
\begin{equation}
	Q = \begin{pmatrix}
		\ii\slashed{\nabla}+m & 0 \\ 0 & -\ii\slashed{\nabla}+m
	\end{pmatrix}, \qquad
	K = \begin{pmatrix}
		K_s & 0 \\ 0 & K_c
	\end{pmatrix}
\end{equation}
on $DM\oplus D^*M$, in which $K_s$ and $K_c$ are the spinorial and cospinorial Klein--Gordon operators $(\slashed{\nabla}^2+m^2$ on $DM$ and $D^*M$ respectively. Evidently $K$ is normally hyperbolic, while $Q$ has symbol $q(x,k) = -\slashed{k} \oplus \slashed{k}$, which is characteristic at $\Nc$ though nonzero (recall that $(\det \slashed{v})^2= \det (\slashed{v}^2) = g^{-1}(v,v)^4$ for any covector $v$). We deduce that
$P$ is $\Nc^\pm$-decomposable and that $\Nc^+$-Hadamard states of the Dirac theory have two-point functions (of the `doubled' $\Xi$ fields) with wavefront set $\Rc^\Had$. This coincides (modulo conventions) with the standard microlocal definition of Hadamard states for Dirac fields~\cite{Hollands:2001,Sanders_dirac:2010}. 
The example of the Proca field will be studied in more detail in Section~\ref{sec:Proca}.

\subsection{Hilbert space formulation} 

Let $\Ac\in\Alg$ be any unital $*$-algebra. A Hilbert space representation of $\Ac$ is a triple
$(\Hf,\pi,D)$ where $\Hf$ is a Hilbert space with dense subspace $D\subset \Hf$ and linear map $\pi:\Ac\to \Lin(D)$, so that  (a)
$\pi(AB)\psi=\pi(A)\pi(B)\psi$ for all $A,B\in\Ac$ and $\psi\in D$, (b) $\pi(1)=\id_\Hf|_D$, and (c) for each $A\in\Ac$, $\pi(A)$ has an adjoint $\pi(A)^*$ whose domain contains $D$ and obeys $\pi(A)^*\psi=\pi(A^*)\psi$ for all $\psi\in D$. By the GNS theorem (see section III.2.2 in~\cite{Haag}), every state $\omega$ on $\Ac$ determines a Hilbert space representation of this type, in which there is a unit vector $\varphi_\omega\in D$ such that
$\omega(A)=\ip{\varphi_\omega}{\pi(A)\varphi_\omega}$ and $D=\pi(\Ac)\varphi_\omega$. In this case, the vector $\varphi_\omega$ is called the GNS vector of $\omega$ and we denote the GNS representation by the quadruple $(\Hf,\pi,D,\varphi_\omega)$ referring to $\Hf$ and $D$ as the GNS Hilbert space and GNS domain respectively.

Following~\cite{StVeWo:2002}, the space of $\Hf$-valued distributions, $\DD'(M;\Hf)$, consists of linear and weakly (hence also strongly) continuous maps $u:\Gamma_0^\infty(\Omega)\to \Hf$. The wavefront set $\WF(u)$ can be defined exactly as for scalar-valued distributions, but replacing the absolute value $|u(\chi f)|$ in~\eqref{eq:regdirndef} by 
the $\Hf$ norm, $\|u(\chi f)\|_\Hf$. The definition is extended to $\Hf$-valued distributional sections of a finite-rank bundle $B$ as in Section~\ref{sec:WF}. First, we give some useful results. 
\begin{lemma}\label{lem:Hilbert_tools}
	Let $V$ and $W$ be Hilbert space valued distributions in $\DD'(B;\Hf)$, 
	and let $\Cc$ be a continuous antilinear map of $B^*\otimes\Omega$. 
	
	(a) If
	$U:\Gamma_0^\infty(B^*\otimes\Omega)\to\Hf$ is linear and obeys  
	$\|U(f)\|_\Hf\le \|V(f)\|_\Hf$ for all $f\in\Gamma_0^\infty(B^*\otimes\Omega)$
	then $U\in\DD'(B;\Hf)$ and $\WF(U)\subset \WF(V)$.
	
	(b) If
	$U:\Gamma_0^\infty(B^*\otimes\Omega)\to\Hf$ is linear and obeys
	$\|U(f)\|_\Hf^2\le \|V(f)\|_\Hf\|W(f)\|_\Hf$ for all $f\in\Gamma_0^\infty(B^*\otimes\Omega)$, then 
	 $U\in\DD'(B;\Hf)$ and
	 $\WF(U)\subset \WF(V)\cap\WF(W)$.
	 
	(c) For any fixed $\psi\in\Hf$, the maps $f\mapsto \ip{\psi}{V(f)}_\Hf$ and
	$f\mapsto\ip{V(\Cc f)}{\psi}_\Hf$ are distributions in $\DD'(B)$ with 
	$\WF(\ip{\psi}{V(\cdot)}_\Hf)\subset \WF(V)$, $\WF(\ip{V(\Cc \cdot)}{\psi}_\Hf)\subset -\WF(V)$.
	
	 (d) $Z(f,h)=\ip{V(\Cc f)}{W(h)}_\Hf$ defines
	$Z\in\DD'(B\boxtimes B)$ with $\WF(Z)\subset (-\WF_0(V))\times\WF_0(W)$.
\end{lemma}
\begin{proof}
	(a,b) First, the bounds on $U$ ensure that for any compact $K\subset M$,
	there is a constant $C_K$ and seminorm $|\cdot|_K$ of $\Gamma_K^\infty(B^*\otimes \Omega)$ so that $\|U(f)\|\le C_K |f|_K$ for all $f\in\Gamma_K^\infty(B^*\otimes \Omega)$, giving $U\in\DD'(B;\Hf)$ in 
	either case. Next, suppose $(x,k)$ is a regular directed point for $V$ and
	let $\chi$ be any smooth density compactly supported in a coordinate chart domain about $x$, with associated coordinate phase functions $\ee_k$.
	Then $\|V(S\chi \ee_l)\|$ decays rapidly as $l\to\infty$ in a conic neighbourhood of $k$, while
	$\|W(S\chi \ee_l)\|$ is polynomially bounded in $l$ by considering the seminorms on $\Gamma_0^\infty(B^*\otimes \Omega)$. The bounds on $U$ imply that $(x,k)$ is 
	a regular directed point for $U$, using for (b) the fact that rapid decay is preserved under multiplication by a polynomially bounded function. Taking complements, $\WF(U)\subset\WF(V)$ which proves (a), and also (b) on allowing the the roles of $V$ and $W$ to be interchanged.
	
	Parts~(c,d) are straightforward generalisations to the bundle context of parts (ii) and~(iii) of Proposition 3.2 in~\cite{FewVer-Passivity}
	(see also Proposition 2.2 in~\cite{StVeWo:2002} for (d)), and can be derived from the scalar versions by considering local frames.
\end{proof}

Theorem~6.1 in~\cite{StVeWo:2002} and the corresponding result Theorem~4.11 in~\cite{Sanders_dirac:2010} have the following generalisation -- we give the proof for completeness.
\begin{thm}\label{thm:HilbertHadamard}
	Let $P$ be a $\Vc^\pm$-decomposable \rfhgho\ and let $(\Hf,\pi,D)$ be a Hilbert space representation of $\Yf(P)$. For any normalised $\psi\in D$, define a state $\omega_\psi$ on $\Yf(P)$ by $\omega_\psi(A)=\ip{\psi}{\pi(A)\psi}_\Hf$. Suppose $\omega_\psi$ has distributional $n$-point functions. Then $\pi(\Upsilon_P(\cdot))\psi\in\DD'(B;\Hf)$, and
	$\omega_\psi$ is $\Vc^+$-Hadamard if and only if 
	\begin{equation}\label{eq:HilbertHadamard}
		\WF(\pi(\Upsilon_P(\cdot))\psi)\subset \Vc^-.
	\end{equation} 
	In particular, if $\omega$ is a $\Vc^+$-Hadamard state of $\Yf(P)$ then~\eqref{eq:HilbertHadamard} holds in the GNS representation of $\omega$, with $\psi$ given by the GNS vector of $\omega$. These results are also valid for fermionic \rfhgho's, replacing
	$\Yf(P)$ by $\Xf(P)$ and $\Upsilon_P$ by $\Xi_P$.
\end{thm}
\begin{proof} 	
	As $\|\pi(\Upsilon_P(f))\psi\|^2$ is a linear combination of $n$-point functions with first argument $\Cc f$ and last argument $f$, the continuity of the distributional $n$-point functions and antilinear map $\Cc$ imply that $\pi(\Upsilon_P(\cdot))\psi\in\DD'(B;\Hf)$. 
	The two-point function $W_\psi$ of $\omega_\psi$ is
	\begin{align}
		W_\psi( f^{\sharp,1} \otimes h^{\sharp,1}) &=\omega_\psi(\Upsilon_P(f)\Upsilon_P(h))= 
	\ip{\pi(\Upsilon_P(\Cc f))\psi}{\pi(\Upsilon_P(h))\psi}.
	\end{align} 
	Condition~\eqref{eq:HilbertHadamard} immediately implies  $\WF(W_\psi)\subset \Vc^+_0 \times \Vc_0^-$ by Lemma~\ref{lem:Hilbert_tools}(d), and therefore $\WF(\adj{W}_\psi)\subset \Vc^-_0 \times \Vc_0^+$ does not intersect $\WF(W_\psi)$ (recall that the zero section is excluded from wavefront sets). But as $W_\psi-\adj{W}_\psi=\ii E_P$, we have $\WF(W_\psi)\subset (\Vc^+_0 \times \Vc_0^-)\cap \WF(E_P)=\Vc^+\times\Vc^-$, so $\omega_\psi$ is $\Vc^+$-Hadamard. Conversely, if $\WF(W_\psi)\subset \Vc^+\times\Vc^-$ then the estimate
	\begin{equation}
		\|\pi(\Upsilon_P(f)\psi\|^2 =  W_\psi(  (\Cc f)^{\sharp,1}\otimes   f^{\sharp,1})
	\end{equation}
	implies that every $(x,k) \notin\Vc^-$ is a regular directed point for $\pi(\Upsilon_P(\,\cdot\, S))\psi$ 
	where $S\in\Gamma^\infty(B^*)$ is arbitrary, so $\WF(\pi(\Upsilon_P(\,\cdot\, S))\psi)\subset \Vc^-$ for every $S$ and~\eqref{eq:HilbertHadamard} follows. The statement about the GNS representation is a simple application of the main part, because $\omega=\omega_{\varphi_\omega}$, where $\varphi_\omega$ is the GNS vector.
	
	For a fermionic \rfhgho\ $(P,R)$ we use identical arguments except that  $W_\psi+\adj{W}_\psi=\ii E_P\circ (1\otimes R)$ is used in place of  $W_\psi-\adj{W}_\psi=\ii E_P$.
\end{proof} 

We give two applications: first, the proof of Theorem~\ref{thm:Hadamard_props}(d); second, a proof that the Hadamard condition is stable under action by operators in the algebra. In both cases we give the proof for the bosonic case but the argument applies equally to the fermionic case with minor adaptations.
\begin{cor}\label{cor:smooth1pt}
	The $1$-point functions of any $\Vc^+$-Hadamard state are smooth.
\end{cor}
\begin{proof}
	Let $(\Hf,\pi,D,\varphi_\omega)$ be the GNS representation of $\Vc^+$-Hadamard state $\omega$ on $\Yf(P)$. Then 
	\begin{equation}
		\omega(\Upsilon_P(f))=\ip{\varphi_\omega}{\pi(\Upsilon_P(f))\varphi_\omega}_\Hf=
	\ip{\pi(\Upsilon_P(\Cc f))\varphi_\omega}{\varphi_\omega}_\Hf,
	\end{equation} 
	so Lemma~\ref{lem:Hilbert_tools}(c) gives $\WF(\omega(\Upsilon_P(\,\cdot\,)))\subset \Vc^-\cap \Vc^+=\emptyset$.  
\end{proof}
\begin{cor}\label{cor:manyHadamard}
	Let $\omega$ be a $\Vc^+$-Hadamard state. Then, for every unit vector $\psi$ in the GNS domain of $\omega$, the state $\omega_\psi$ is $\Vc^+$-Hadamard.
\end{cor}
\begin{proof} 
	Let $(\Hf,\pi,D,\varphi_\omega)$ be the GNS representation of $\omega$, so $D=\pi(\Yf(P))\varphi_\omega$.
	As $\omega$ has distributional $n$-point functions, so does every $\omega_\psi$ for unit $\psi\in D$. Noting that~\eqref{eq:HilbertHadamard} is stable under linear combination, it is enough to show that if a unit $\psi\in D$ defines a $\Vc^+$-Hadamard state, then so does every nonzero vector $\psi'=\pi(\Upsilon_P(h))\psi$ for $h\in\Gamma_0^\infty(B)$ (normalised by rescaling $h$ if necessary). But
	\begin{equation}
		\pi(\Upsilon_P(f))\psi' = \pi(\Upsilon_P(f)\Upsilon_P(h))\psi = \ii E_P(f,h)\psi + \pi(\Upsilon_P(h)\Upsilon_P(f))\psi,
	\end{equation}
	the first term of which is smooth because
	$\WF(f\mapsto E_P(f,h))\subset \{(x,k)\in T^*M: (x,k;y,0)\in\WF(E_P)~\text{for some $y\in M$}\}=\emptyset$. To complete the proof, the estimate
	\begin{equation}
		\|\pi(\Upsilon_P(h)\Upsilon_P(f))\psi\|^2\le 
		\|\pi(\Upsilon_P(\Cc h)\Upsilon_P(h)\Upsilon_P(f))\psi\|	
		\|\pi(\Upsilon_P(f))\psi\|
	\end{equation}
	 and Lemma~\ref{lem:Hilbert_tools}(b) show that $\WF(\pi(\Upsilon_P(\cdot))\psi')\subset \Vc^-$ -- note that $\pi(A\Upsilon_P(\cdot))\psi\in\DD'(B;\Hf)$ for all $A\in\Yf(P)$ because $\omega_\psi$ has distributional $n$-point functions (cf.~the proof of Theorem~\ref{thm:HilbertHadamard}).
\end{proof}
A simple consequence is that the set of $\Vc^+$-Hadamard states $\Sf_{\Vc^+}(P)$ forms a \emph{state space} for $\Yf(P)$ as defined e.g., in Section~4.5.3 of~\cite{FewVerch_aqftincst:2015}. 
\begin{cor}\label{cor:state_space}
	The set of $\Vc^+$-Hadamard states $\Sf_{\Vc^+}(P)\subset\Yf(P)^*_{+,1}$ is closed under finite convex combinations and operations induced by $\Yf(P)$.
\end{cor}
\begin{proof}
The $2$-point function of a finite convex combination of states is the corresponding convex combination of their $2$-point functions. If $\omega\in \Sf_{\Vc^+}(P)$ with GNS representation $(\Hf,\pi,D,\varphi_\omega)$ and $B\in\Yf(P)$ with
$\omega(B^*B)=1$ the state $\omega_B$ defined by $\omega_B(A)=\omega(B^*AB)$ is represented
by the unit vector $\pi(B)\varphi_\omega\in D$. Thus $\omega_B\in\Sf_{\Vc^+}(P)$ by Corollary~\ref{cor:manyHadamard}.
\end{proof}

\subsection{Pull-backs, push-forwards and covariance}

Let $S:P_1\to P_2$ be a $\GreenHyp$-morphism and suppose that $\omega_2$ is a state on $\Yf(P_2)$. Then
there is a pulled back state $\omega_1=\Yf(S)^*\omega_2$ on $\Yf(P_1)$ so that 
$\omega_1(A)=\omega_2(\Yf(S)A)$. At the level of two-point functions, this implies 
\begin{equation}
	W_1(f^{\sharp,1}\otimes   h^{\sharp,1} ) = W_2 (  (S f)^{\sharp,1}\otimes  ( Sh)^{\sharp,1}),
\end{equation}
i.e., $W_1 = W_2\circ (S^{\sharp,1}\otimes S^{\sharp,1})$ 
using the musical notation for linear maps introduced in Section~\ref{sec:HVBs}.  
The advanced-minus-retarded two-point distributions obey the same relation $E^\twopt_{P_1} = E^\twopt_{P_2}\circ (S^{\sharp,1}\otimes S^{\sharp,1})$. If $S$ is Cauchy, then one can invert $\Yf(S)$ and
push forward any state $\omega_1$ on $\Yf(P_1)$ to a state $\omega_2=(\Yf(S)^{-1})^*\omega_1$ on $\Yf(P_2)$.
Let $L:P_2\to P_1$ be any $\GreenHyp$-morphism obeying $\hat{L}=\hat{S}^{-1}$ (see  Theorem~\ref{thm:GreenHyp}(d)). Then $SLh=h\mod P_2\Gamma_0^\infty(B_2)$ and
\begin{equation}
	\omega_2(\Yf_{P_2}(h)\Yf_{P_2}(h'))=\omega_2(\Yf_{P_2}(SL h)\Yf_{P_2}(SLh'))= 
	\omega_1(\Yf_{P_1}(L h)\Yf_{P_1}(Lh')),
\end{equation}
which gives the equation $W_2 = W_1\circ (L^{\sharp,1}\otimes L^{\sharp,1})$ for the corresponding two-point functions. Analogous constructions are possible in the fermionic case.

We will show that the pullback and pushforward states are Hadamard under suitable conditions, for which we need the following definition.
\begin{defn}
	A linear map $T:\Gamma_0^\infty(B_1)\to \Gamma_0^\infty(B_2)$ will be called \emph{regular} if (a) $T$ is continuous and has a continuous formal transpose $\adj{T}:\Gamma^\infty_\sc(B_2)\to \Gamma^\infty(B_1)$, (b) for any compact $K_1\subset M_1$ there is a compact
	$K_2\subset M_2$ so that $T\Gamma_{K_1}^\infty(B_1)\subset \Gamma^\infty_{K_2}(B_2)$ (this is true in particular if $T$ is properly supported),\footnote{The operator $T$ is said to be \emph{properly supported} (cf.\ Definition 18.1.21 in~\cite{Hormander3}) if to each compact $K_X\subset X$ there is a compact $K_Y\subset Y$ so that $T\Gamma_{K_X}^\infty(B_X)\subset \Gamma_{K_Y}^\infty(B_Y)$ and
		to each compact $K_Y\subset Y$ there is a compact set $K_X\subset X$ so that
		$u|_{K_X}=0$ implies $(Tu)|_{K_Y}=0$ for $u\in\Gamma_0^\infty(B_X)$.
		Equivalently, the projections from the support of $T^\knl$ to $X$ and $Y$ are both proper maps.} and (c)   
	\begin{equation}
		\WF(T^\knl)\cap ((0_{T^*M_2}\times \dot{T}^*M_1) \cup (\dot{T}^*M_2\times 0_{T^*M_1})) = \emptyset.
	\end{equation}
	In particular, a morphism of $\GreenHyp$ will be described as regular if its underlying linear map is regular in this sense and a morphism of $\fGreenHyp$ is regular if its underlying $\GreenHyp$ morphism is regular. 
\end{defn}

The following technical result is
proved in Appendix~\ref{sec:proofWFuoTT}.
\begin{lemma}\label{lem:WFuoTT}
	Suppose that $T_j:\Gamma_0^\infty(B_1)\to \Gamma_0^\infty(B_2)$ ($j=1,2$) are regular. 
	Let $\Uc\subset\dot{T}^*(M_2\times M_2)$ be a closed conic set that does not intersect $(0_{T^*M_2}\times \dot{T}^*M_2) \cup (\dot{T}^*M_2\times 0_{T^*M_2})$.
	If $u\in\DD'(B_2\boxtimes B_2)$ has $\WF(u)\subset \Uc$, then
	$u\circ (T_1^{\sharp,1}\otimes T_2^{\sharp,1})\in\DD'(B_1\boxtimes B_1)$ has wavefront set
	\begin{equation}\label{eq:WFuoTT}
		\WF(u\circ (T_1^{\sharp,1}\otimes T_2^{\sharp,1}))\subset (\WF'(\adj{}T_1^\knl)\times \WF'(\adj{}T_2^\knl)) \bullet \Uc.
	\end{equation}
	Here, for $\Zc\subset T^*(Y\times X)$ and $\Xc\subset T^*X$, we write
	\begin{equation}
		\Zc\bullet \Xc:=\{(y,l)\in T^*Y: \text{$\exists$ $(k,x)\in \Xc$ s.t., $(y,l;x,k)\in \Zc$}\} = 
		\pr_{T^*Y} \left(\Zc\cap (T^*Y\times \Xc)\right),
	\end{equation}
	and $\WF'(\adj{}T^\knl)=\{(x,k;y,l)\in \dot{T}^*(M_1\times M_2):
	(y,-l;x,k) \in \WF(T^\knl)\}$.
\end{lemma} 
With this result in hand, we can establish regularity of a variety of $\GreenHyp$-morphisms.
\begin{lemma}\label{lem:regularGreenHyp} Examples of regular $\GreenHyp$-morphisms include the following.
	\begin{enumerate}[a)]
	\item The $\GreenHyp$-morphisms of Theorem~\ref{thm:GreenHyp}(a--c) are all regular. 
	
	\item Suppose that $S:P_1\to P_2$ is a continuous Cauchy $\GreenHyp$-morphism with a continuous transpose $\adj{S}:\Gamma_\sc^\infty(B_2)\to\Gamma^\infty(B_1)$ and the additional property that, for every spatially compact set $V_2\subset \Mb_2$, there is a spatially compact set $V_1\subset \Mb_1$ so that $\adj{S}\Gamma^\infty_{V_2}(B_2)\subset \Gamma^\infty_{V_1}(B_1)$. Then any $\GreenHyp$-morphisms
	$L=[P_1,\chi]\,\adj{S}E_{P_2}:P_2\to P_1$ of the type given explicitly in Theorem~\ref{thm:GreenHyp}(d) is regular. 
	In particular, these additional conditions hold if $S$ is a Cauchy morphism of the form described in
	Theorem~\ref{thm:GreenHyp}(a,b).
	
	\item The $\GreenHyp$-morphisms of Theorem~\ref{thm:GreenHyp}(e) are regular. 
	\end{enumerate}
\end{lemma}
\begin{proof}
	(a) Parts~(b,c) of Theorem~\ref{thm:GreenHyp} were instances of its part~(a), so we need only address that case. Suppose that  $\beta:B_1\to B_2$ is a hermitian vector bundle morphism covering $\Loc$-morphism $\psi:\Mb_1\to\Mb_2$ such that $P_2\beta_*=\beta_*P_1$, with corresponding continuous $\GreenHyp$-morphism $\beta_*:(B_1,P_1)\to (B_2,P_2)$.  Then
	$\beta_*$ has a continuous adjoint $\adj{\beta}_*:\Gamma^\infty_\sc(B_2)\to \Gamma^\infty(B_1)$ given as $(\adj{\beta}_* h)(x)=\adj{\beta}_x h(\psi(x))$. For any compact $K_1\subset M_1$ one has
	$\beta_*\Gamma^\infty_{K_1}(B_1)\subset \Gamma^\infty_{\psi(K_1)}(B_2)$ and $\psi(K_1)$ is compact. Furthermore, 
	as $\beta_* fS=(\psi_* f)\beta_*S$ for $f\in C_0^\infty(M_1)$, $S\in\Gamma^\infty(B_1)$, it follows that  $\WF(\beta_*^\knl)\subset	\WF(\psi_*^\knl)$.
	Now $\psi_*^\knl$ can be expressed using a distributional pullback. Specifically, let 
	$\rho_2\in\Gamma^\infty(\Omega_{M_2})$ be any smooth nonvanishing density on $M_2$, for example the metric volume density. Then 
	\begin{equation}
		\psi_*^\knl=(1\otimes\rho_1) (1\otimes\psi)^*((\rho_2^{-1/2}\otimes\rho_2^{-1/2})\delta_{M_2}),
	\end{equation} 
	where $\delta_{M_2}\in\DD'(\Omega^{1/2}_{M_2\times M_2})$ is the Dirac distribution on $M_2$
	and $\rho_1=\psi^*\rho_2$. The pullback of distributions exists by the criterion of Theorem 2.5.11${}'$ in~\cite{Hoer_FIOi:1971} because $\psi$ is an embedding, giving the wavefront set bound
	\begin{align}
		\WF(\beta_*^\knl)&\subset	\WF(\psi_*^\knl)\subset (1\otimes\psi)^*\WF(\delta_{M_2}) =
		\{(y,l;x,\adj{}\psi'(x)k):(y,l;\psi(x);k)\in\WF(\delta_{M_2})\} \nonumber\\
		&=
		\{(\psi(x),l;x,-\adj{}\psi'(x)l): x\in M_1,~l\in \dot{T}^*_{\psi(x)}M_2\}.
	\end{align}
	In particular, 
	\begin{equation}
		\WF(\beta_*^\knl)\cap ((0_{T^*M_2}\times T^*M_1)\cup (T^*M_2\times 0_{T^*M_1}))=\emptyset,
	\end{equation}
	so $\beta_*$ is regular.
	
	(b)	The formal transpose of $L$ is $\adj{L}=-E_{P_2}S\,\adj{[P,\chi]}$, noting that
		$\adj{[P,\chi]}:\Gamma^\infty_\sc(B_1)\to \Gamma_0^\infty(B_1)$. Suppose $K_2\subset M_2$ is compact, then $E_{P_2}\Gamma_{K_2}^\infty(B_2)\subset \Gamma_{J(K_2)}^\infty(B_2)$. Hence there is a spatially compact set $V_1$ in $\Mb_1$ so that
		$\adj{S} E_{P_2}\Gamma_{K_2}^\infty(B_2)\subset \Gamma_{V_1}^\infty(B_1)$, so $L\Gamma_{K_2}^\infty(B_2)\subset \Gamma_{\overline{U}\cap V_1}^\infty(B_1)$, which demonstrates the required support property.
		We compute
		\begin{align}\label{eq:LknlWFbd}
			\WF(L^\knl) \subset \WF((\adj{S}E_{P_2})^\knl)&=\WF(E_{P_2}^\knl\circ(S\otimes 1))\nonumber\\
			&\subset (\WF'(\adj{}S^\knl)\times\WF'(\adj{}\id^\knl))\bullet ((\Vc_2^+\times\Vc_2^-)\cup(\Vc_2^-\times\Vc_2^+)) \nonumber\\
			&\subset (\Vc_1^+\times\Vc_2^-)\cup(\Vc_1^-\times\Vc_2^+)
		\end{align}	
		and therefore $\WF(L^\knl)\cap ((T^*M_1\times 0_{T^*M_2} )\cup (0_{T^*M_1}\times T^*M_2))=\emptyset$. Hence $L$ is regular.	 
		
		For the particular case mentioned, in Theorem~\ref{thm:GreenHyp}(a), $\beta_*$ is continuous, with a continuous transpose, and 
		for spatially compact $V_2\subset \Mb_2$, one has $\adj{\beta}_*\Gamma^\infty_{V_2}(B_2)\subset
		\Gamma^\infty_{\psi^{-1}(V_2)}(B_1)$, where $\beta$ covers $\Loc$-morphism $\psi$. As $\psi$ is Cauchy, $\psi^{-1}(V_2)$ is spatially compact: for if $K_2\subset M_2$ is compact and $\Sigma\subset\psi(M_1)$ is a Cauchy surface of $\Mb_2$, then $\psi^{-1}(J(K_2))\subset J(\psi^{-1}(\Sigma\cap J(K_2)))$, which is spatially compact. As Theorem~\ref{thm:GreenHyp}(b) is a special case, we have completed this part of the proof.
		
		(c) The regularity properties can be read off from the $S=\id$ case of~(b).
\end{proof}

\begin{thm}\label{thm:pullbacks}
	Suppose $S:P_1\to P_2$ is a regular $\GreenHyp$-morphism (resp., $\fGreenHyp$-morphism).
	Suppose that $P_2$ is $\Vc^\pm_2$-decomposable and that $\WF'(\adj{}(S^\knl))\bullet\Vc_2^\pm\subset \Vc_1^\pm$, where $\Vc_1^\pm\subset \dot{T}^*M_1$ are conic and relatively closed with $\Vc_1^-=-\Vc_1^+$ and $\Vc_1^+\cap \Vc_1^-=\emptyset$. Then:
\begin{enumerate}[a)]	
	\item $P_1$ is $\Vc_1^\pm$-decomposable and
	every $\Vc_2^+$-Hadamard state $\omega_2$ on $\Yf(P_2)$ pulls back to a $\Vc_1^+$-Hadamard state
	$\omega_1=\Yf(S)^*\omega_2$ on $\Yf(P_1)$; that is, $\Yf(S)^*\Sf_{\Vc_2^+}(P_2)\subset \Sf_{\Vc_1^+}(P_1)$
	(resp., the same holds with $\Yf$ replaced by $\Xf$). 
	
	\item Suppose $S$ is Cauchy and has the additional property that for every spatially compact set $V_2\subset \Mb_2$ there is a spatially compact set $V_1\subset \Mb_1$ so that $\adj{S}\Gamma^\infty_{V_2}(B_2)\subset \Gamma^\infty_{V_1}(B_1)$. (The additional property is satisfied if $S$ arises from Theorem~\ref{thm:GreenHyp}(a,b) -- see Lemma~\ref{lem:regularGreenHyp}(b).) 
	Then every $\Vc_1^+$-Hadamard state $\omega_1$ on $\Yf(P_1)$ 
	pushes forward to a $\Vc_2^+$-Hadamard state $\omega_2=\Yf(S)^{-1}\omega_1$ on $\Yf(P_2)$ which is the unique state so that $\omega_1=\Yf(S)^*\omega_2$; thus, $\Sf_{\Vc_1^+}(P_1)=\Yf(S)^*\Sf_{\Vc_2^+}(P_2)$ (resp., the same holds with $\Yf$ replaced by $\Xf$).
	\end{enumerate}
\end{thm}
\begin{proof}
	(a) 
	Note that $S$ and $S^{\sharp,1}$ have distributional kernels with the same support and wavefront set. Therefore $S^{\sharp,1}$ is regular. Using $E^\twopt_{P_1} = E^\twopt_{P_2}\circ (S^{\sharp,1}\otimes S^{\sharp,1})$
	and $\Vc_2^\pm$-decomposability of $P_2$, we compute
	\begin{align}
		\WF(E_{P_1}^\twopt) &= (\WF'(\adj{}(S^\knl))\times \WF'(\adj{}(S^\knl)))  \bullet\WF(E^\twopt_{P_2}) \nonumber\\
		&\subset (\WF'(\adj{}(S^\knl))\times \WF'(\adj{}(S^\knl)))\bullet
		((\Vc_2^+\times\Vc_2^-)\cup(\Vc_2^-\times\Vc_2^+)) \nonumber\\
		&\subset (\Vc_1^+\times\Vc_1^-)\cup (\Vc_1^-\times\Vc_1^+) ,
	\end{align}
	thus establishing $\Vc_1^\pm$-decomposability of $P_1$. 
	A similar argument applies to $W_1=W_2\circ (S^{\sharp,1}\otimes S^{\sharp,1})$ giving $\WF(W_1)\subset \Vc_1^+\times\Vc_1^-$, so
	$\omega_1$ is $\Vc_1^+$-Hadamard. The argument is identical in the fermionic case. 
	
	(b) We have $W_2=W_1\circ(L^{\sharp,1}\otimes L^{\sharp,1})$ where the continuous linear map
	$L:\Gamma_0^\infty(B_2)\to \Gamma_0^\infty(B_1)$ can be given explicitly
	as in Theorem~\ref{thm:GreenHyp}(d) by 
	$L=[P,\chi]\,\adj{S} E_{P_2}$, with
	$\chi$ constructed as described there. Then 
	$L$ and consequently $L^{\sharp,1}$ are regular by Lemma~\ref{lem:regularGreenHyp}(b),
	and the wavefront set bound~\eqref{eq:LknlWFbd} implies $\WF'(\adj{}(L^\knl))\bullet \Vc_1^\pm\subset \Vc_2^\pm$. 
	We may now employ Lemma~\ref{lem:WFuoTT} to compute $\WF(W_2)=\WF(W_1\circ(L^{\sharp,1}\otimes L^{\sharp,1}))\subset \Vc_2^+\times\Vc_2^-$ and conclude that $\omega_2$ is $\Vc_2^+$-Hadamard. Finally,	
	$\Yf(S)^*\omega_2$ has $n$-point functions $W_1^{(n)}\circ (L^{\sharp,1})^{\otimes n} \circ (S^{\sharp,1})^{\otimes n}=W_1^{(n)}$, where $W_1^{(n)}$ are the $n$-point functions of $\omega_1$, because $\hat{L}=\hat{S}^{-1}$ and consequently $LSf-f\in P_1\Gamma_0^\infty(B_1)$. Thus $\Yf(S)^*\omega_2=\omega_1$. 
	Uniqueness is clear as $\Yf(S)$ is an isomorphism. Again the argument applies directly to the fermionic case.
\end{proof}

\begin{defn}
 Let $\EoM:\Loc\to\GreenHyp$ be a functorially assigned \rfhgho. Suppose that each $\EoM(\Mb)$ is decomposable with respect to cones $\Vc^\pm(\Mb)\subset T^*M$, such that $\psi^*\Vc^+(\Mb')\subset \Vc^+(\Mb)$ for all $\Loc$-morphisms $\psi:\Mb\to\Mb'$. Then we say that $\EoM$ is $\Vc^\pm$-decomposable.
\end{defn}
An immediate consequence of this definition, together with Theorem~\ref{thm:pullbacks} and Corollary~\ref{cor:state_space} is the following.
\begin{cor}\label{cor:covariant_state_space}
	Suppose that $\EoM:\Loc\to\GreenHyp$ is a $\Vc^\pm$-decomposable functorial \rfhgho\  taking Cauchy $\Loc$-morphisms to Cauchy $\GreenHyp$-morphisms, and
	let $\Zf=\Yf\circ\EoM$ be the corresponding locally covariant QFT (see Theorem~\ref{thm:lcqft}). Then the assignment $\Sf(\Mb)=\Sf_{\Vc^+(\Mb)}(\EoM(\Mb))$ for each $\Mb\in\Loc$ and $\Sf(\psi)=\Yf(\psi)^*|_{\Sf(\Mb')}$ for each $\Loc$-morphism $\psi:\Mb\to\Mb'$ defines a covariant state space for $\Zf$ obeying the timeslice condition in the sense defined in Section~4.5.3 of~\cite{FewVerch_aqftincst:2015}.
\end{cor}
\begin{proof}
	Each $\Sf(\Mb)$ is a state space by Corollary~\ref{cor:state_space}. If $\psi:\Mb\to\Mb'$ in $\Loc$ then $\EoM(\psi)$ is regular by Lemma~\ref{lem:regularGreenHyp}(a) and 
	Theorem~\ref{thm:pullbacks}(a) shows that $\Sf(\psi)\Sf(\Mb')\subset\Sf(\Mb)$; it is clear that $\Sf(\id_\Mb)=\id_{\Sf(\Mb)}$ and
	$\Sf(\psi\circ\varphi)=\Sf(\varphi)\circ\Sf(\psi)$ by the functorial properties of $\Zf$ and
	contravariance of the dual. Finally, if $\psi$ is Cauchy then $\EoM(\psi)$ is a Cauchy $\GreenHyp$-morphism (of the type given in Theorem~\ref{thm:GreenHyp}(a)) and Theorem~\ref{thm:pullbacks}(b) 
	shows that $\Sf(\Mb)=\Sf(\psi)\Sf(\Mb')$. These are the conditions needed for $\Sf$ to be a covariant state space obeying the timeslice condition given (in a slightly different form) in~\cite{FewVerch_aqftincst:2015}.
\end{proof}
Two further consequences of Theorem~\ref{thm:pullbacks} concern the existence of
Hadamard states.
\begin{cor}\label{cor:deformation}
	Suppose that $\EoM:\Loc\to\GreenHyp$ is a $\Vc^\pm$-decomposable functorial \rfhgho, and
	let $\Zf=\Yf\circ\EoM$ be the corresponding locally covariant QFT (see Theorem~\ref{thm:lcqft}). 
	If $\Zf(\Mb)$ admits a $\Vc^+(\Mb)$-Hadamard state  and $\Mb'\in\Loc$ has Cauchy surfaces that are oriented-diffeomorphic to those of $\Mb$, then $\Zf(\Mb')$ admits $\Vc^+(\Mb')$-Hadamard states.  
\end{cor}
\begin{proof}
	If $\Mb$ and $\Mb'$ have oriented-diffeomorphic Cauchy surfaces then there is
	a chain of Cauchy morphisms linking $\Mb$ to $\Mb'$ via an interpolating 
	spacetime $\Mb'''\in\Loc$
	\begin{equation}\label{eq:Cauchy_chain}
		\Mb\leftarrow \Mb'' \rightarrow \Mb''' \leftarrow \Mb''''\rightarrow \Mb',
	\end{equation}
	where $\Mb''$ (resp., $\Mb''''$) is a suitable future (resp., past) region
	of $\Mb$ (resp., $\Mb'$) regarded as a spacetime in its own right 
	(see e.g., Proposition~2.4 in~\cite{FewVer:dynloc_theory}). Under $\EoM$, the morphisms in~\eqref{eq:Cauchy_chain} are mapped to Cauchy $\GreenHyp$-morphisms of the type described in Theorem~\ref{thm:GreenHyp}(b), to which Lemma~\ref{lem:regularGreenHyp}(b) therefore applies. Using Theorem~\ref{thm:pullbacks} repeatedly, alternating between parts~(a) and~(b) to pull back and push forward states, we deduce that $\Zf(\Mb')$ admits at least one $\Vc^+(\Mb)$-Hadamard state (and hence many such states by Corollary~\ref{cor:manyHadamard}).
\end{proof}
Consequently, it is sufficient to prove the existence of $\Vc^+(\Mb)$-Hadamard states for a sufficiently broad class of spacetimes. For example, we will say that $\Mb\in\Loc$ is an ultrastatic spacetime with bounded spatial geometry if it has underlying manifold $M=\RR\times\Sigma$ with metric $g=1\oplus -h$ for some complete Riemannian metric $h$ on $\Sigma$ of bounded geometry, i.e., the Riemann tensor of $h$ and all its covariant derivatives are bounded with respect to the tensorial norms determined by $h$.
\begin{cor}\label{cor:bultrastatic}
	Every $\Mb\in\Loc$ has smooth spacelike Cauchy surfaces oriented-diffeomorphic to those of an ultrastatic spacetime $\Mb'\in\Loc$ that is ultrastatic with bounded spatial geometry. Consequently, under the assumptions of Corollary~\ref{cor:deformation}, if $\Zf(\Mb)$ admits at least one $\Vc^+(\Mb)$-Hadamard state for every ultrastatic spacetime $\Mb\in\Loc$ with bounded spatial geometry then $\Zf(\Mb)$ admits (many) $\Vc^+(\Mb)$-Hadamard states for all $\Mb\in\Loc$. 
\end{cor}
\begin{proof}
	Let $\Mb\in\Loc$ and let $\Sigma\subset M$ be any smooth spacelike Cauchy surface with the inherited orientation $\nu$ (which is a top form on $\Sigma$). Then $\Sigma$ admits a complete Riemannian metric $h$ with bounded geometry~\cite{Greene:1978} -- if necessary we choose such a metric for each connected component of $\Sigma$. Let $\Mb'$ be the ultrastatic spacetime $\RR\times \Sigma$ with metric $\dd t\otimes \dd t -h$, time-orientation so that $\dd t$ is future-pointing, and spacetime orientation so that $\dd t\wedge \nu$ is positively oriented. Global hyperbolicity holds by Proposition~5.2 in~\cite{Kay1978}, so $\Mb'\in\Loc$ has the required properties. The consequence is immediate by Corollary~\ref{cor:deformation} and Corollary~\ref{cor:manyHadamard}.
\end{proof}

\subsection{Scattering morphisms}\label{sec:scattering}
 
Suppose that $P$ and $Q$ are \rfhgho's on $M$ that agree outside a temporally compact set $U\subset M$. Then $M^\pm=M\setminus J^\mp(U)$ are open causally convex subsets on which $P$ and $Q$ agree. Applying Theorem~\ref{thm:GreenHyp}(e), there are Cauchy $\GreenHyp$-morphisms
$\lambda_{QP}^\pm:(B,P)\to (B,Q)$ called \emph{response maps} so that $\lambda_{QP}^\pm\Gamma_0^\infty(B)\subset \Gamma_0^\infty(B|_{M^\pm})$ and $\hat{\lambda}_{QP}^\pm[f]_P=[f]_Q$ for $f\in\Gamma_0^\infty(B|_{M^\pm})$. There 
are also Cauchy $\GreenHyp$-morphisms $\lambda^\pm_{PQ}:(B,Q)\to (B,P)$ with analogous properties. Therefore $\lambda^-_{PQ}\lambda^+_{QP}$ is a $\GreenHyp$-endomorphism of $(B,P)$,
and any equivalent morphism is called a \emph{scattering morphism} for $Q$ relative to $P$ with \emph{coupling zone} $U$. Scattering morphisms of this type play an important role in the theory of measurement~\cite{FewVer_QFLM:2018} -- see Section~\ref{sec:FVmeasurements}. 
\begin{thm}\label{thm:scattering_morphism}
	Under the conditions just stated, there is a scattering morphism $\vartheta$ for $Q$ relative to $P$ with coupling zone $U$ such that $\vartheta f=f-(Q-P)E_Qf=f-(Q-P)E_Q^-f$ for $f\in\Gamma_0^\infty(B|_{M^+})$. Furthermore $\lambda^+_{QP}\sim  \lambda^-_{QP}\vartheta$,
	i.e.,  $\hat{\vartheta}=(\hat{\lambda}^-_{QP})^{-1}\hat{\lambda}^+_{QP}$.  
\end{thm}
\begin{proof}
	Using the freedom to change $\lambda^-_{PQ}$ and $\lambda^+_{QP}$ for equivalent morphisms, we may assume they are both of the explicit form given in Theorem~\ref{thm:GreenHyp}(e).
	Thus $\lambda^+_{QP}f=f \mod P\Gamma_0^\infty(B)$ and $\lambda^-_{PQ}f=f \mod Q\Gamma_0^\infty(B)$. Since $P$ and $Q$ are injective on $\Gamma_0^\infty(B)$ (see Theorem~\ref{thm:GreenHypOps}\eqref{it:kernels}) there are unique linear maps $\rho^+_{QP}, \rho^-_{PQ}:\Gamma_0^\infty(B)\to \Gamma_0^\infty(B)$ such that $\lambda^+_{QP}=\id + P\rho^+_{QP}$, $\lambda^-_{PQ}=\id + Q\rho^-_{PQ}$, both of which commute with the complex involution.
	
	In particular, $\lambda^-_{PQ}\lambda^+_{QP}=\lambda^+_{QP} + Q\rho^{-}_{PQ} \lambda^+_{QP}$ entailing
	\begin{equation}
		E_Q^- \lambda^+_{QP} = E_Q^-\lambda^-_{PQ}\lambda^+_{QP} -\rho^{-}_{PQ} \lambda^+_{QP}
	\end{equation}
	and consequently
	\begin{equation}\label{eq:EQlambda}
		E_Q \lambda^+_{QP} = E_Q^-\lambda^-_{PQ}\lambda^+_{QP} -\rho^{-}_{PQ} \lambda^+_{QP} - E_Q^+\lambda^+_{QP}.
	\end{equation}
	Noting that $Q$ and $P$ agree on the first and last terms on the right-hand side of~\eqref{eq:EQlambda}, which are supported in $M^-$ and $M^+$ respectively, one has
	\begin{equation}
		(Q-P)E_Q \lambda^+_{QP} = -(Q-P)\rho^{-}_{PQ} \lambda^+_{QP}
	\end{equation}
	and by applying $Q$ to~\eqref{eq:EQlambda} and rearranging,
	\begin{align}
		\lambda^-_{PQ}\lambda^+_{QP} &=QE_Q^+\lambda^+_{QP} + Q\rho^{-}_{PQ} \lambda^+_{QP} = 
		\lambda^+_{QP}- (Q-P)E_Q \lambda^+_{QP} + P\rho^{-}_{PQ} \lambda^+_{QP}\nonumber \\
		&=\id - (Q-P)E_Q \lambda^+_{QP}  + P(\rho^+_{QP}+\rho^{-}_{PQ} \lambda^+_{QP}),
	\end{align} 
	by definition of $\rho_{QP}^+$.
	By Lemma~\ref{lem:GreenHyp0}(b), $\lambda^-_{PQ}\lambda^+_{QP}\sim \vartheta:=\id - (Q-P)E_Q \lambda^+_{QP}$, 
	giving $\lambda^+_{QP}\sim \lambda^-_{QP}\vartheta $ because
	$\lambda_{QP}^-\lambda_{PQ}^-\sim \id$. 
	Therefore $\hat{\vartheta}=(\hat{\lambda}^-_{QP})^{-1}\hat{\lambda}^+_{QP}$.
	It remains to compute $\vartheta f$ for $f\in\Gamma_0^\infty(B|_{M^+})$. Here, we use the fact that $\lambda^+_{QP}f-f=P\rho^+_{QP}f\in\ker E_P\cap \Gamma^\infty_0(B|_{M^+})$, moreover, $\rho^+_{QP}f\in \Gamma^\infty_0(B|_{M^+})$ by Theorem~\ref{thm:GreenHypOps}\eqref{it:refinedkernel} and causal convexity of $M^+$. Consequently, $\lambda^+_{QP}f-f=P\rho^+_{QP}f\in Q\Gamma_0^\infty(B)$ and so $E_Q\lambda^+_{QP}f=E_Qf$. Thus
	$\vartheta f=f-(Q-P)E_Qf$ for $f\in\Gamma_0^\infty(B|_{M^+})$, and as $Q$ and $P$ agree on the support of $E^+_Qf$, one has $\vartheta f=f-(Q-P)E_Q^-f$.
\end{proof}
If $V\subset M$ is temporally compact with $U\subset V$, then $\vartheta$ is evidently also a scattering morphism for $Q$ relative to $P$ with coupling zone $V$ with the properties of Theorem~\ref{thm:scattering_morphism}. In this sense, $\vartheta$ is independent of the coupling zone.  

The main result of this subsection is that the Hadamard condition is preserved under scattering.
Recall from Section~\ref{sec:GreenHyperbolic} that when two \rfhgho's $P$ and $Q$ agree outside a temporally compact set there is an associated scattering $\GreenHyp$-automorphism $\vartheta$ of $P$ that describes the scattering of $Q$ relative to $P$ and which then lifts to an automorphism $\Yf(\vartheta)$ of the QFT algebra $\Yf(P)$.  
\begin{thm}\label{thm:Hadamard_scattering}
	Let $P$ and $Q$ be \rfhgho's on hermitian vector bundle $B$ over $\Mb\in\Loc$ and suppose that $P$ and $Q$ agree outside a temporally compact set $U\subset \Mb$. Let $\vartheta$ be a scattering morphism of $Q$ relative to $P$ and define the scattering automorphism of
	$\Yf(P)$ by $\Theta=\Yf(\vartheta)$. If $P$ and $Q$ are $\Vc_P^\pm$- and $\Vc_Q^\pm$-decomposable and $\Vc_P^+|_{M\setminus U}= \Vc_Q^+|_{M\setminus U}$ then
	$\Theta^*$ is an isomorphism on the set of $\Vc_P^+$-Hadamard states.
\end{thm}
\begin{proof}
	From Section~\ref{sec:GreenHyperbolic}, $\vartheta\sim\lambda^-_{PQ}\lambda^+_{QP}$
	where $\lambda^-_{PQ}$ and $\lambda^+_{QP}$ are $\GreenHyp$-morphisms of the type given by
	Theorem~\ref{thm:GreenHyp}(e) and are Cauchy and regular. Then
	$\Theta =\Yf(\lambda^-_{PQ})\Yf(\lambda^+_{QP})$ by Theorem~\ref{thm:RFHGHOquantisation}(b)
	and so $\Theta^*$ maps $\Vc_P^+$-Hadamard states to $\Vc_P^+$-Hadamard states. As
	$\Theta^{-1}=\Yf(\lambda^+_{PQ})\Yf(\lambda^-_{QP})$, the same argument shows that
	$\Theta^*$ is an isomorphism of the set of $\Vc_P^+$-Hadamard states.
\end{proof}

\subsection{Tensor products and partial traces}\label{sec:tensor_products_partial_traces}

We now prove a number of useful properties of Hadamard states, which will be used in Section~\ref{sec:FVmeasurements} in an application to the theory of measurement. 
Parts~(a) and~(b) of the following result show how the Hadamard condition extends to composite theories, while~(c) relates to partial traces.
\begin{thm}\label{thm:Hadamard_tools}
	Let $P$ and $Q$ be \rfhgho{}s, and suppose that $P$ is $\Vc^\pm_P$-decomposable while $Q$ is $\Vc^\pm_Q$-decomposable, and $\Vc_P^+\cap \Vc_Q^-=\emptyset$. Then:
	\begin{enumerate}[(a)] 
	\item \label{it:Hadoplus} $P\oplus Q$ is $\Vc_{P\oplus Q}^\pm$-decomposable, where $\Vc_{P\oplus Q}^\pm=\Vc_P^\pm\cup \Vc_Q^\pm$; 
	\item if $\omega$ is a $\Vc_P^+$-Hadamard state on $\Yf(P)$ and $\sigma$ is a $\Vc_Q^+$-Hadamard state on $\Yf(Q)$, then $\omega\otimes\sigma$ is a $\Vc_{P\oplus Q}^+$-Hadamard state on $\Yf(P\oplus Q)$; 
	\item \label{it:Hadpartial} if $\varphi$ is a $\Vc_{P\oplus Q}^+$-Hadamard state on $\Yf(P\oplus Q)=\Yf(P)\otimes\Yf(Q)$, then
	the partial traces $\varphi_P\in\Yf(P)^*_{+,1}$ and $\varphi_Q\in\Yf(Q)^*_{+,1}$ given by $\varphi_P(A)=\varphi(A\otimes \1_{\Yf(Q)})$, $\varphi_Q(B)=\varphi(\1_{\Yf(P)}\otimes B)$, are respectively $\Vc_P^+$-Hadamard and $\Vc_Q^+$-Hadamard. 
	\end{enumerate}
\end{thm} 
Inspection of the proof shows that the result adapts immediately to the case where
$P$ and $Q$ are fermionic \rfhgho's; part~(a) only requires that $P$ and $Q$ are semi-Green-hyperbolic.
\begin{proof}
	(a) First note that $\Vc_{P\oplus Q}^+\cap \Vc_{P\oplus Q}^-=\emptyset$ because the same is true of $\Vc_P^\pm$ and $\Vc_Q^\pm$, and $\Vc_P^+\cap \Vc_Q^-=\emptyset$.
	As $E_{P\oplus Q}=E_{P}\oplus E_Q$, we have
	\begin{equation}
		\WF'(E_{P\oplus Q})\subset \WF'(E_P)\cup \WF'(E_Q)\subset 
		\bigcup_{\sharp=\pm}((\Vc_P^\sharp\times\Vc_P^\sharp)\cup (\Vc_Q^\sharp\times\Vc_Q^\sharp))\subset \bigcup_{\sharp=\pm}(\Vc_{P\oplus Q}^\sharp\times\Vc_{P\oplus Q}^\sharp),
	\end{equation}
	which is the decomposability condition.
	
	(b) The $n$-point functions of $\omega\otimes \sigma$ on $\Yf(P\oplus Q)$ may be written as sums of products of $r$-point functions of $\omega$ and $\sigma$ for $1\le r\le n$, and are easily seen to be distributions. To simplify notation we write every $n$-point function of $\omega$ as $W_\omega$, distinguishing them by their domains, and similarly for the $n$-point functions of $\sigma$. Then 
	the $2$-point function of $\omega\otimes\sigma$ is given by
	\begin{equation}
		W_{\omega\otimes\sigma}(({f}\oplus{h})\otimes
		({f'}\oplus{h'})) = 
		W_\omega({f}\otimes{f'})+W_\sigma({h}\otimes{h'}) + W_\omega({f}) W_\sigma(h') +
		W_\omega(f') W_\sigma(h) . 
	\end{equation}
	As the $1$-point functions are smooth by Theorem~\ref{thm:Hadamard_props}(d), we have
	\begin{equation}
		\WF(W_{\omega\otimes\sigma})\subset \WF(W_{\omega})\cup \WF(W_{\sigma})\subset (\Vc_P^+\times\Vc_P^-) \cup (\Vc_Q^+\times\Vc_Q^-)\subset \Vc_{P\oplus Q}^+\times\Vc_{P\oplus Q}^-,
	\end{equation}
	so $\omega\otimes\sigma$ is $\Vc_{P\oplus Q}^+$-Hadamard.
	
	(c) It is enough to show that $\varphi_P$ is $\Vc_P^+$-Hadamard.
	As $\Upsilon_P({f})\otimes\1_{\Yf_Q(M)}=\Upsilon_{P\oplus Q}({f}\oplus 0)$, it is clear that $\varphi_P$ has distributional $n$-point functions. In particular, the two-point function $W_P$ of $\varphi_P$ is
	given by 
	\begin{equation}
		W_P({f}\otimes{h}) = W(({f}\oplus 0)\otimes({h}\oplus 0)),
	\end{equation}
	where $W$ is the two-point function of $\varphi$. 
	It follows that $\WF(W_P)\subset \WF(W)\subset 
	\Vc_{P\oplus Q}^+\times \Vc_{P\oplus Q}^-$. Now $P$ is $\Vc_{P\oplus Q}^\pm$-decomposable (as well as $\Vc_P^\pm$-decomposable), so we have shown that $W_P$ is 
	$\Vc_{P\oplus Q}^+$-Hadamard. Using Theorem~\ref{thm:Hadamard_props}(a) we have
	$\WF(W)\subset (\Vc_{P\oplus Q}^+\times \Vc_{P\oplus Q}^-)\cap \WF(E_P)\subset \Vc_P^+\times \Vc_P^-$ using the $\Vc_P^\pm$-decomposability of $E_P$. Therefore $W_P$ is $\Vc_P^+$-Hadamard.	
	\end{proof}

\subsection{Higher $n$-point functions}\label{sec:truncated}

Let $\omega$ be a $\Vc^+$-Hadamard state on $\Yf(P)$ for $\Vc^\pm$-decomposable \rfhgho\ $P$ and define the $n$-point functions
$W^{(n)}\in\DD'(B^{\boxtimes n})$ as described above. The truncated $n$-point functions $W^{(n)}_{T}$ are defined recursively by
\begin{equation}
	W^{(n)}(F_I) = \sum_{\Pc\in\mathscr{P}(I)} \prod_{J\in \Pc} W^{(|J|)}_T(F_J) , \qquad W^{(1)}_{T}=W^{(1)},
\end{equation}
where $I$ is a finite ordered set, $\mathscr{P}(I)$ is the set of partitions of $I$ into nonempty ordered subsets and $F_J$ is the ordered list of $F_j\in\Gamma_0^\infty(B^*\otimes\Omega)$ for $j\in J$. A useful identity (see \cite{BratteliRobinson_vol2}, section 5.2.3) is, for any fixed $j\in I$,
\begin{equation}\label{eq:Tsimple}
	W^{(|I|)}(F_{I}) = \sum_{j\in J\subset I} W^{(|J|)}_T(F_J)W^{(|I\setminus J|)}(F_{I\setminus J})
\end{equation}
(i.e., one sums over all $J$ obeying $j\in J\subset I$ for the given $j$ and $I$) with $W^{(0)}=1$ by convention. A quasifree state is one whose truncated $n$-point functions vanish except for $n= 2$.

Truncated $n$-point functions can be defined for \emph{even states} on $\Xf(P)$, i.e., those whose $n$-point functions all vanish for all odd $n$. They are defined recursively by 
\begin{equation}
	W^{(n)}(F_I) = \sum_{\Pc\in\mathscr{P}(I)} \varepsilon(\Pc) \prod_{J\in \Pc} W^{(|J|)}_T(F_J),  
\end{equation}
where  
$\varepsilon(\Pc)$ is the signature of the permutation of $I$ that results in $(J_1,\ldots J_{|\Pc|})$, where $\Pc=\{J_1,\ldots,J_{|\Pc|}\}$ (the ordering of the $J_k$'s is irrelevant). 
It is sufficient to sum over only partitions into sets of even length. The analogue of~\eqref{eq:Tsimple}
is 
\begin{equation}\label{eq:TsimpleCAR}
	W^{(|I|)}(F_{I}) = \sum_{j\in J\subset I}  \varepsilon(J,I\setminus J) W^{(|J|)}_T(F_J)W^{(|I\setminus J|)}(F_{I\setminus J}),
\end{equation}
with $\varepsilon(J,I\setminus J)$ equal to the signature of the permutation from $I$ to $(J,I\setminus J)$, 
and $W^{(0)}=1$ as before. The quasifree states on $\Xf(P)$ are those even states whose truncated $n$-point functions vanish for $n\neq 2$.

The class of states with Hadamard two-point functions and smooth truncated $n$-point functions for $n\neq 2$ was identified as a plausible class of physical states in~\cite{Kay_Leipzig:1992}, where the question was raised as to how close this was to the Hadamard condition. 
This is answered by a result of Sanders~\cite{Sanders:2010}, which we now adapt and generalise, slightly streamlining in the process.
\begin{thm}\label{thm:truncated}
	Let $\omega$ be a $\Vc^+$-Hadamard state on $\Yf(P)$ or an even state on $\Xf(P)$. Then all its truncated $n$-point functions for $n\neq 2$ are smooth, and $\WF(W^{(2)}_{T})=\WF(W^{(2)})$. 
\end{thm}
\begin{proof}
	We give the argument for the bosonic case first. Smoothness of the $1$-point functions was proved in Corollary~\ref{cor:smooth1pt} from which it follows that $W^{(2)}_{T}-W^{(2)}$ is smooth, so we can restrict to showing smoothness of $W^{(n)}_{T}$ for $n\ge 3$. 
	
	For some $n\ge 3$, suppose that $W^{(k)}_{T}$ is smooth for $2<k<n$, which is trivially true for $n=3$. Suppose that $W^{(n)}_{T}$ is not smooth, in which case there must be
	$(\xu,\ku)\in\WF(W^{(n)}_{T})$ with $k_j\neq 0$ for some $1\le j\le n$. Considering~\eqref{eq:Tsimple}, we deduce that
	$(\xu,\ku)\in\WF(R_{n,j})$ where
	\begin{equation}\label{eq:Rnj_def}
		R_{n,j}(F_I)= W^{(n)}(F_I) - \sum_{i<j} W^{(2)}(F_i,F_j) W^{(n-2)}(F_{I\setminus\{i,j\}}) -  
		\sum_{i>j} W^{(2)}(F_j,F_i) W^{(n-2)}(F_{I\setminus\{i,j\}})
	\end{equation}
	because all other terms are smooth with respect to $x_j$,
	and $W^{(2)}_{T}-W^{(2)}$ is smooth.
	Now, the commutation relations and $\ii E_P=W^{(2)}-\adj{W}^{(2)}$ imply that
	\begin{equation}
		R_{n,j} = R_{n,1}\circ \tau_{1j} = R_{n,n}\circ\tau_{jn},
	\end{equation}
	where $\tau_{ij}$ transposes the $i$'th and $j$'th arguments. Consequently, 
	$\tau_{1j}(\xu,\ku)\in \WF(R_{n,1})$ and 
	$\tau_{jn}(\xu,\ku)\in \WF(R_{n,n})$.
	
	We now employ Cauchy--Schwarz estimates, splitting off the first and last arguments, to deduce that
	$\WF(W^{(n)})\subset \Vc_0^+\times T^*M^{\times (n-1)}$ and
	$\WF(W^{(n)})\subset T^*M^{\times (n-1)}\times \Vc_0^-$. It follows from~\eqref{eq:Rnj_def} and
	$\WF(W^{(2)})\subset\Vc^+\times\Vc^-$ that $\WF(R_{n,1})\subset
	\Vc_0^+\times T^*M^{\times (n-1)}$ and
	$\WF(R_{n,n})\subset
	T^*M^{\times (n-1)}\times\Vc_0^-$. Therefore $(x_j,k_j)\in \Vc_0^+\cap\Vc_0^-=\{0\}$, contradicting $k_j\neq 0$.
	Thus $W^{(n)}_{T}$ is smooth and by induction this holds for all $n\ge 3$.
	
	The argument for the fermionic case is similar, with $R_{n,j}$ 
	modified by the insertion of a factor $(-1)^{i+j-1}$ in both sums, and the use
	of the anticommutation relations to obtain
	$R_{n,j}=(-1)^{j-1}R_{n,1}\circ \tau_{1j} = (-1)^{n-j} R_{n,n}\circ \tau_{jn}$.
\end{proof}
The next result is obtained by employing the formulae~\eqref{eq:Tsimple} and~\eqref{eq:TsimpleCAR}, together with the standard estimates for wavefront sets of sums and tensor products of distributions.
\begin{cor}\label{cor:truncated}
	Under the same assumptions as Theorem~\ref{thm:truncated}, 
	\begin{align}\label{eq:WFWnbd}
		\WF(W^{(n)})&\subset \bigcup_{Q} \{(x_1,k_1;\ldots;x_n,k_n)\in T^*M: 
		\textnormal{$k_j=0$ for $j\notin \bigcup Q$,}
		\nonumber \\
		&\qquad\qquad
		(x_{\min(r)},k_{\min(r)};x_{\max(r)},k_{\max(r)})\in\WF_0(W^{(2)})~\textnormal{for each $r\in Q$}	\},
	\end{align}
	where $Q$ runs over all sets of disjoint two-element subsets of $\{1,\ldots,n\}$. 
\end{cor}
The bounding set in~\eqref{eq:WFWnbd} is the standard bound for quasifree Hadamard states~\cite{BruFreKoe}. 
In this sense, all $\Vc^+$-Hadamard states are `microlocally quasifree'. 

\subsection{Polarisation sets}

To close this section, we briefly consider the polarisation set of 
$\Vc^+$-Hadamard two-point functions. The polarisation set was introduced by Dencker~\cite{Dencker:1982} and the conventions and notation used here are as in the companion paper~\cite{Fewster:2025b}. An
elementary application of the general theory provides a refinement of Theorem~\ref{thm:Hadamard_props}(a).
\begin{thm} \label{thm:Hadpropsaprime}
	Suppose $P$ is a $\Vc^\pm$-decomposable \rfhgho, where $\Vc^\pm$ are as in Definition~\ref{def:Hadamard}. If $W^\twopt$ is the $2$-point function of a $\Vc^+$-Hadamard state on the bosonic algebra $\Yf(P)$, then
	\begin{equation}
		\WF_\pol((\rho^{1/2}\otimes\rho^{1/2})W^\twopt) = (1\otimes\flat) \WF_\pol(E_{\rho^{1/2}P\rho^{-1/2}}^\knl)|_{\Vc^+\times \Vc^-}\cup 0.
	\end{equation} 
	If $(P,R)$ is a fermionic \rfhgho\ and $W^\twopt$ is the $2$-point function of a $\Vc^+$-Hadamard state on the fermionic algebra $\Xf(P)$, then
	\begin{equation}
		\WF_\pol((\rho^{1/2}\otimes\rho^{1/2})W^\twopt) = (1\otimes (R\circ\flat)) \WF_\pol(E_{\rho^{1/2}P\rho^{-1/2}}^\knl)|_{\Vc^+\times \Vc^-}\cup 0.
	\end{equation}
\end{thm}
\begin{proof}
	In the bosonic case, one first calculates that 
	\begin{align}\label{eq:CCRtwopt}
		(\rho^{1/2}\otimes\rho^{1/2})(W^\twopt-\adj{W}^\twopt)&=\ii   E_P^\knl\circ (\rho^{1/2}\otimes (\flat\circ\rho^{-1/2})) \nonumber  \\
		&= \ii (\rho^{1/2}E_P\rho^{-1/2})^\knl \circ (1\otimes (\rho^{1/2}\circ\flat\circ\rho^{-1/2})) \nonumber \\
		&= \ii (1\otimes (\rho^{1/2}\circ\flat\circ\rho^{-1/2})) E_{\rho^{1/2}P\rho^{-1/2}}^\knl,
	\end{align}
	where we use the fact that $\rho^{1/2}E_P\rho^{-1/2}= E_{\rho^{1/2}P\rho^{-1/2}}$ as operators. Thus
	\begin{equation}
		\WF_\pol((\rho^{1/2}\otimes\rho^{1/2})(W^\twopt-\adj{W}^\twopt)) = (1\otimes\flat) \WF_\pol(E_{\rho^{1/2}P\rho^{-1/2}}^\knl),
	\end{equation}
	where we have used a standard result on polarisation sets (see section~2 of~\cite{Dencker:1982} or Lemma~3.1 of~\cite{Fewster:2025b}) and the fact that $\flat$ is an isomorphism. 
	As the wavefront sets of $W^\twopt$ and $\adj{W}^\twopt$ intersect trivially, their polarisation sets sum to that on the right-hand side, and the restriction to $\Vc^+\times\Vc^-$ therefore gives the polarisation set of $(\rho^{1/2}\otimes\rho^{1/2}) W^\twopt$ modulo the zero section.  
	
	For the fermionic case,~\eqref{eq:CCRtwopt} is replaced by
	\begin{equation}
		(\rho^{1/2}\otimes\rho^{1/2})(W^\twopt+\adj{W}^\twopt) = \ii (1 \otimes 
		(\rho^{1/2}\circ R \circ \flat\circ\rho^{-1/2})') E_{\rho^{1/2}P\rho^{-1/2}}^\knl
	\end{equation}
	due to the anticommutation relations. Noting that the symbols of $(\rho^{1/2}\circ R \circ \flat\circ\rho^{-1/2})'$ and
	$R$ coincide up to sign (recall that $R=-\adj{R}$), the result follows.
\end{proof}

\section{An application to measurement schemes}\label{sec:FVmeasurements}

As an application of the theory developed so far, we consider the measurement of local observables
of QFTs based on \rfhgho's according to the scheme introduced in~\cite{FewVer_QFLM:2018} and further developed in~\cite{BostelmannFewsterRuep:2020,FewsterJubbRuep:2022}. In particular, we will show that the state updates resulting from nonselective measurements map Hadamard states to Hadamard states.\footnote{As this manuscript was being completed, J.~Mandrysch informed me that he and M.~Navascu\'es have obtained a similar result in a particular case corresponding to finite rank perturbations (to appear in a revised version of~\cite{MandryschNavascues:2024}).}

To start, we describe the measurement framework of~\cite{FewVer_QFLM:2018} in a simplified setting where all the QFTs involved are obtained from \rfhgho's. Suppose that $P$ is a \rfhgho\ on bundle $B_P$ over $\Mb\in\Loc$, and let $\Yf(P)$ be the algebra of observables for the QFT of interest -- the \emph{system}. Let $Q$ be a \rfhgho\ on bundle $B_Q$ over $\Mb$ and regard $\Yf(Q)$ as a \emph{probe} theory that will be used to measure the system. The system and probe can be combined as an \emph{uncoupled combination} $\Yf(P)\otimes\Yf(Q)\cong \Yf(P\oplus Q)$. We also consider a \emph{coupled combination} described by 
a QFT with algebra of observables $\Yf(T)$, where $T$ is also a \rfhgho\ on $B_P\oplus B_Q$ that
agrees with $P\oplus Q$ outside a temporally compact set $S$ (typically either a compact subset or a closed Cauchy slab), thus defining a scattering automorphism $\Theta$ on $\Yf(P\oplus Q)$ are defined as in section~\ref{sec:scattering}. The overall idea of~\cite{FewVer_QFLM:2018} is that one measures observables of the system by measuring probe observables after the coupling has ceased. Suppose that the system and probe are prepared at early times in the states
$\omega$ and $\sigma$ respectively and that a nonselective measurement of a probe observable is made. 
Tracing out the probe, the result of the measurement is that the system state should be updated from $\omega$ to a state $\omega^{\textnormal{n.s.}}$ on $\Yf(P)$, so that
\begin{equation}\label{eq:omegans}
	\omega^{\textnormal{n.s.}}(A) = (\omega\otimes\sigma)(\Theta(A\otimes \1_{\Yf(Q)}))
\end{equation}
as shown in section~3.3 of~\cite{FewVer_QFLM:2018}. With this introduction, we pass to the main result of this section.
\begin{thm}\label{thm:Hadamard_update}
	Let $P$, $Q$ and $T$ be \rfhgho{}s acting on sections of bundles $B_P$, $B_Q$ and $B_P\oplus B_Q$ respectively, so that $T$ and $P\oplus Q$ agree outside a temporally compact set $U$ with scattering morphism $\vartheta$ on $(B_P\oplus B_Q,P\oplus Q)$. 
	 Suppose that $P$ and $Q$ are, respectively, $\Vc_P^\pm$- and $\Vc_Q^\pm$-decomposable with $\Vc_{P}^+\cap \Vc_Q^-=\emptyset$, and suppose that $T$ is $\Vc_T^\pm$-decomposable, where $\Vc_T$ satisfies $\Vc_T^+|_{M\setminus U}=\Vc_{P\oplus Q}^+|_{M\setminus U}$. Suppose $\omega\in\Yf(P)^*_{+,1}$ and $\sigma\in\Yf(Q)^*_{+,1}$ are 
	$\Vc_P^+$-Hadamard and $\Vc_Q^+$-Hadamard, respectively, and let $\omega^{\textnormal{n.s.}}$ be the nonselective update of $\omega$ from a measurement using $\Yf(T)$ as a coupled combination of the system $\Yf(P)$ and probe $\Yf(Q)$ (with scattering map $\Theta=\Yf(\vartheta)$) and probe preparation state $\sigma$, so that $\omega^{\textnormal{n.s.}}$ is given by~\eqref{eq:omegans}
	Then  $\omega^{\textnormal{n.s.}}$ is a $\Vc_P^+$-Hadamard state on $\Yf_P(M)$.  
\end{thm}
\begin{proof}
	The state $\omega\otimes\sigma$ is $\Vc^+_{P\oplus Q}$-Hadamard by Theorem~\ref{thm:Hadamard_tools}\ref{it:Hadoplus} and therefore 
	$\Theta^*(\omega\otimes\sigma)$ is $\Vc^+_{P\oplus Q}$-Hadamard by Theorem~\ref{thm:Hadamard_scattering}.
	As $\omega^{\text{n.s.}}$ is a partial trace of $\Theta^*(\omega\otimes\sigma)$, the result follows from 
	Theorem~\ref{thm:Hadamard_tools}\ref{it:Hadpartial}.  
\end{proof}

\section{Example: the neutral Proca field}\label{sec:Proca}

Now specialise to $n=4$ dimensions. For any $\Mb\in\Loc$, let $\Lambda^pM$ be the bundle of complex-valued $p$-forms on $M$. The Hodge dual $\star_\Mb:\Lambda^pM\to \Lambda^{4-p}M$ is defined by
\begin{equation}
	\omega\wedge\star_\Mb \eta = \frac{1}{p!}\omega_{\alpha_1\cdots\alpha_p}\eta^{\alpha_1\cdots\alpha_p}\vol,
\end{equation}
where $\vol$ is the metric volume $4$-form. Then $(f,h):=(-1)^p\star_\Mb^{-1}(\overline{f}\wedge\star_\Mb h)$ defines a hermitian bundle metric on $\Lambda^pM$ (note the sign)
with integrated version 
\begin{equation}
	(f,h)=(-1)^p\int_\Mb \overline{f}\wedge\star_\Mb h,
\end{equation}
and one also defines the codifferential $\delta_\Mb$  
\begin{equation}
	\delta_\Mb\omega = (-1)^{\deg \omega} \star_\Mb^{-1}\dd_\Mb\star_\Mb \omega,
\end{equation}
for any $p$-form field $\omega$, where $\dd_\Mb$ is the usual exterior derivative. Defining a bilinear pairing on $\Lambda^pM$ by $\langle f,h\rangle = (\overline{f},h)$, one has  $\delta_\Mb=-\hadj{\dd_\Mb}=-\adj{\dd_\Mb}$. 

With these sign choices, the 
neutral Proca action with mass $m>0$ is
(suppressing cutoffs)
\begin{equation}
	S_{\textnormal{Proca}}[Z] = -\frac{1}{2}\langle \dd_\Mb Z,\dd_\Mb Z\rangle 
	-\frac{1}{2}m^2\langle Z,Z\rangle,
\end{equation}
for $1$-form field $Z$, 
and takes the same form as the action of the minimally coupled massive scalar field 
\begin{equation}
	S_{\textnormal{scalar}}[\phi] =  -\frac{1}{2}\langle \dd_\Mb \phi,\dd_\Mb \phi\rangle -\frac{1}{2}m^2\langle \phi,\phi\rangle,
\end{equation}
for $0$-form field $\phi$.\footnote{The overall sign of the action, in both cases, is the one that guarantees that the corresponding stress-energy tensor $T_{\mu\nu}=2\delta S/\delta g^{\mu\nu}$ obeys the classical weak energy condition -- see~\cite{Few&Pfen03}
for classical and quantum energy conditions obeyed by the Proca theory.}

The equation of motion derived from $S_{\textnormal{Proca}}$ is 
$P_\Mb Z = 0$, where the
neutral Proca operator on $\Gamma^\infty(\Lambda^1 M)$,  
\begin{equation}
	P_\Mb = -\delta_\Mb\dd_\Mb + m^2
\end{equation}  
is real and formally hermitian. As is well known (see e.g.~\cite{Few&Pfen03}), $P_\Mb$ has Green operators 
\begin{equation}\label{eq:EPMKMDM}
	E_{P_\Mb}^\pm = E_{K_\Mb}^\pm D_\Mb ,
\end{equation}
where $K_\Mb = -(\delta_\Mb\dd_\Mb+\dd_\Mb\delta_\Mb)+m^2$ is the $1$-form Klein--Gordon operator on $\Gamma^\infty(\Lambda^1 M)$ and
\begin{equation}
	D_\Mb = 1-\frac{1}{m^2}\dd_\Mb\delta_\Mb.
\end{equation}
If $u\in\Gamma^\infty(\Lambda^1M)$ satisfies $P_\Mb u=0$, then by applying $\delta_\Mb$ we find that $\delta_\Mb u=0$ and consequently $K_\Mb u=0$; conversely, solutions to $K_\Mb u=0$ with $\delta_\Mb u=0$ also solve $P_\Mb u=0$. Proca solutions therefore form a subspace of the Klein--Gordon solutions, and indeed $D_\Mb$ is a projection onto this subspace. The same holds for distributional solutions. 

For later use, we note the Weitzenb\"ock formula $(K_\Mb A)_\mu =\Box_\Mb A_\mu + (m^2\delta_\mu^{\phantom{\mu}\nu} + R_\mu^{\phantom{\mu}\nu}) A_\nu$, where $\Box_\Mb=g^{\mu\nu}\nabla_\mu\nabla_\nu$ is given by the Levi--Civita connection and $R_{\mu\nu}$ is the corresponding Ricci tensor, $R_{\alpha\beta}=R^{\lambda}_{\phantom{\lambda}\alpha\lambda\beta}$, where our conventions are that $(\nabla_\alpha\nabla_\beta- \nabla_\alpha\nabla_\beta)v^\mu = R_{\alpha\beta\lambda}^{\phantom{\alpha\beta\lambda}\mu} v^\lambda$.
Thus the Weitzenb\"ock connection for $K_\Mb$ (see Section~\ref{sec:Hadcond}) is simply the Levi--Civita connection on $\Lambda^1M$. We also note that $D_\Mb$ has principal symbol with action
\begin{equation}\label{eq:DM_principalsymbol}
	\sigma(D_\Mb)(x,k)v = -m^{-2} g_x^{-1}(k,v)k, \qquad v\in T_x^*M
\end{equation}
and therefore has nontrivial kernel $k^\perp\subset T_x^*M$ for all $(x,k)\in \dot{T}^*M$.
Thus $\Char(D_\Mb)=\dot{T}^*M$, but it should also be noted that $\sigma(D_\Mb)$ is nowhere zero on $\dot{T}^*M$.

Summarising, $\Lambda^1\Mb =(\Lambda^1M,\Mb,(\cdot,\cdot),\overline{\cdot})$ is a Hermitian vector bundle and $\EoM(\Mb):=(\Lambda^1\Mb,P_\Mb)$ is a \rfhgho. Moreover, if $\psi:\Mb\to\Mb'$ in $\Loc$, then $\psi^*P_{\Mb'}=P_{\Mb}\psi^*$, so
by Theorem~\ref{thm:GreenHyp}(a) the pushforward $\psi_*$ on compactly supported $1$-forms determines a $\GreenHyp$-morphism $\EoM(\psi):\EoM(\Mb)\to\EoM(\Mb')$, which respects composition of morphisms and makes $\EoM:\Loc\to \GreenHyp$ into a functor.

The quantised theory $\Zf=\Yf\circ\EoM$ is a locally covariant QFT. Explicitly, each $\Zf(\Mb)$ is generated by 
symbols $\Zc_\Mb(f)$ ($f\in\Gamma_0^\infty(\Lambda^1 M)$) subject to relations
\begin{enumerate}
	\item[Z1] $f\mapsto \Zc_\Mb(f)$ is complex-linear
	\item[Z2] $\Zc_\Mb(f)^*=\Zc_\Mb(\overline{f})$ (hermiticity)
	\item[Z3] $\Zc_\Mb(P_{\Mb}f)=0$ (field equation)
	\item[Z4] $[\Zc_\Mb(f),\Zc_\Mb(h)] = \ii E_{P_\Mb}(f,h)\1_{\Zf(\Mb)}$ (commutation relations)
\end{enumerate}
where $E_{P_\Mb}(f,h) = -\int_\Mb f\wedge \star_\Mb E_{P_\Mb}h$ due to the choice of sign in the bundle metric on $\Lambda^1M$.
The morphism $\Zf(\psi)$ induced by $\psi:\Mb\to\Mb'$ in $\Loc$ is determined uniquely by
\begin{equation}
	\Zf(\psi)\Zc_\Mb(f) = \Zc_{\Mb'}(\psi_*f), \qquad f\in\Gamma_0^\infty(\Lambda^1 M).
\end{equation}
According to our earlier convention, the generators $\Zc_\Mb(f)$ 
are interpreted as quantisations of the functional $F_f(Z)=\langle f,Z\rangle=-\int_\Mb f\wedge \star_\Mb Z$ on the classical solution space $\Sol(P_\Mb)$. As the theory has
a $\ZZ_2$ global gauge symmetry under $\Zc_\Mb(f)\mapsto -\Zc_\Mb(f)$ (even at the functorial level, cf.~\cite{Fewster:gauge}) one is free to interpret the generators
alternatively as quantisations of $Z\mapsto \int_\Mb f\wedge \star_\Mb Z$ if desired. Note that the signs in Z4 and in the formula for $E_{P_\Mb}$ are unchanged under the $\ZZ_2$ symmetry.
  
Owing to~\eqref{eq:EPMKMDM}, we may use Theorem~\ref{thm:WFQEKR} and Corollary~\ref{cor:WFQEKR} to  
describe the Hadamard states of the theory.
\begin{thm}\label{thm:WFEP}
	The wavefront set of $E_{P_\Mb}$ is given by
	\begin{equation}\label{eq:WFEProca}
		\WF(E_{P_\Mb})=\Rc_\Mb.
	\end{equation}	
Consequently, the neutral Proca operator $P_\Mb$ is $\Nc_\Mb^\pm$-decomposable in every $\Mb\in\Loc$,
and the functor $\EoM$ is $\Nc^\pm$-decomposable.
The $\Nc_\Mb^\pm$-Hadamard states of the Proca field have distributional two-point functions satisfying
	\begin{equation}
		\WF(W)=\Rc_\Mb^\Had.  
	\end{equation}  
\end{thm}
\begin{proof}
Eq.~\eqref{eq:WFEProca} is a simple application of Theorem~\ref{thm:WFQEKR},
using the fact that $\sigma(D_\Mb)$ is nonvanishing on $\Nc$, and
has also been given as Theorem~5.1 of~\cite{Fewster:2025b}.   Corollary~\ref{cor:WFQEKR} completes the proof.  
\end{proof}
By Corollary~\ref{cor:deformation} it follows that if $\Zf(\Mb)$ admits $\Nc_\Mb^\pm$-Hadamard states for some $\Mb\in\Loc$ then the same is true for every $\Mb'\in\Loc$ whose Cauchy surfaces are oriented-diffeomorphic to those of $\Mb$. 

The formulation of Hadamard states for the Proca field has been 
studied from the microlocal viewpoint by Fewster and Pfenning (FP)~\cite{Few&Pfen03}
and more recently by Moretti, Murro and Volpe (MMV)~\cite{MorettiMurroVolpe:2023} (note that various sign conventions differ between~\cite{MorettiMurroVolpe:2023} and the present paper). We now briefly explain the relationship between these treatments and what has been presented above. 

\paragraph{FP-Hadamard states} The FP definition of a Hadamard state on $\Zf(\Mb)$ is
slightly indirect. Namely, a state $\omega$ on $\Zf(\Mb)$ is FP-Hadamard (i.e., Hadamard according to~\cite{Few&Pfen03}) if (adjusting to the conventions used here) the two-point function $W$ obeys 
 \begin{equation}\label{eq:FPcondition}
 	W  = H\circ (1\otimes D_\Mb^{\sharp,1}) 
 \end{equation}
for some $K_\Mb$-bisolution $H$ of Hadamard-form. The latter condition can be understood in terms of Hadamard series, but can be stated equivalently by requiring that $H$ has antisymmetric part $\tfrac{1}{2}\ii E^\twopt_{K_\Mb}$ modulo smooth contributions and that $H$ obeys the wavefront set condition $\WF(H)=\Rc_\Mb^\Had\subset \Nc_\Mb^+\times\Nc_\Mb^-$. By nonexpansion of the wavefront set, every 
FP-Hadamard state is $\Nc_\Mb^\pm$-Hadamard -- the reverse implication will be proved below in Theorem~\ref{thm:FPequivalence}.
  
The existence of FP-Hadamard states was proved for ultrastatic spacetimes with compact Cauchy surfaces $\Sigma$ with trivial homology $H_1(\Sigma)$ (though this condition was really imposed to facilitate a parallel treatment of the Maxwell field rather than out of necessity). By deformation arguments this implies existence on general globally hyperbolic spacetimes whose Cauchy surfaces are compact and obey the homology condition. This was proved in Section~IV.E of~\cite{Few&Pfen03} for FP-Hadamard states, so all such spacetimes have FP-Hadamard (and therefore $\Nc_\Mb^\pm$-Hadamard) states.
 
\paragraph{MMV-Hadamard states}
The MMV definition is that a Hadamard state is a quasifree state with two-point function $W$ obeying $\WF(W)=\Rc_\Mb^\Had$. Here, we recall that a state $\omega$ is  \emph{quasifree} if it is completely determined by its two-point function according to
\begin{equation}\label{eq:quasifree}
	\omega(\ee^{\ii\Zc_\Mb(f)}) = \ee^{-W(f,f)/2}, \qquad f\in\Gamma_0^\infty(\Lambda^1\Mb)
\end{equation}
in the sense of formal power series in $f$, thus fixing all $n$-point functions (this is equivalent to the vanishing of all truncated $n$-point functions for $n\neq 2$). In particular, every quasifree FP-Hadamard state is MMV-Hadamard, and every MMV-Hadamard state is $\Nc_\Mb^+$-Hadamard. By Theorem~\ref{thm:WFEP}, every quasifree $\Nc_\Mb^+$-Hadamard state is MMV-Hadamard. Thus we have 
\begin{equation}
	\textnormal{quasifree FP-Hadamard} \implies \textnormal{MMV-Hadamard} \iff \textnormal{quasifree $\Nc_\Mb^+$-Hadamard} .
\end{equation}
MMV also show (Theorem 5.9 of~\cite{MorettiMurroVolpe:2023}) that there is a MMV-Hadamard state on every ultrastatic spacetime with bounded spatial geometry. By Corollary~\ref{cor:bultrastatic}, we deduce the following.
 \begin{thm}\label{thm:ProcaHadamard}
	There exist $\Nc_\Mb^+$/MMV-Hadamard states on the Proca algebra $\Zf(\Mb)$ for every $\Mb\in\Loc$.
\end{thm} 
MMV state a similar result, using M{\o}ller operators and a notion of \emph{paracausal deformation} developed in~\cite{MMVparacausal:2023}, and also state a result showing that MMV-Hadamard is `almost equivalent' to the FP-Hadamard definition. 
 However, there turns out to be a gap in some of the proofs in~\cite{MorettiMurroVolpe:2023} affecting 
\begin{itemize}
	\item the argument that MMV-Hadamard form propagates for quasifree states (Prop 4.7 of~\cite{MorettiMurroVolpe:2023})
	\item the `almost equivalence' between MMV and FP, namely that if $\Zf(\Mb)$ admits any quasifree FP-Hadamard state, then every MMV-Hadamard state on $\Zf(\Mb)$ is also FP-Hadamard (Theorem 6.6 of~\cite{MorettiMurroVolpe:2023})
	\item the stability of MMV-Hadamard states under M{\o}ller isomorphisms (part (3) of Theorem 4.9 in~\cite{MorettiMurroVolpe:2023}); as this is the key tool in reducing the general globally hyperbolic case to ultrastatic spacetimes with bounded spatial geometry, it also affects their result on the existence of MMV-Hadamard states in all globally hyperbolic spacetimes, stated as Theorem 4 and repeated as Theorem 5.10 in~\cite{MorettiMurroVolpe:2023}.
\end{itemize}

The problem, in all cases, is an incorrect argument used to compute $\WF(E_{P_\Mb})$ -- it is erroneously claimed in part~(b) of the proof of Proposition~4.7 in~\cite{MorettiMurroVolpe:2023} that the operator $D_\Mb= 1-m^{-2}\dd_\Mb\delta_\Mb$ (denoted $Q$ there) satisfies $\Char(1\otimes D_\Mb)=T^*M\times 0_{T^*M}$ and therefore does not intersect $\WF(E_{K_\Mb})$ (effectively stating that $D_\Mb$ is elliptic). However, $D_\Mb$ is characteristic everywhere on $\dot{T}^*M$ and in particular
on $\Nc_\Mb$, as we have already mentioned after~\eqref{eq:DM_principalsymbol} , so in fact $\Char (1\otimes D_\Mb)\supset T^*M\times\Nc_\Mb\supset  \WF(E_{P_\Mb})$. This problem is remedied by our Theorem~\ref{thm:WFEP}, which (via Theorem~\ref{thm:WFQEKR}) uses the polarisation set to track the singularities of $E_{K_\Mb}$ in more detail and hence shows that $\WF(E_{P_\Mb})$ has the expected form. With this gap filled, the results of~\cite{MorettiMurroVolpe:2023} should go through. In any case, neither Theorem~\ref{thm:ProcaHadamard} nor Theorem~\ref{thm:FPequivalence} below relies on the results of~\cite{MorettiMurroVolpe:2023} except in relation to the existence of Hadamard states on ultrastatic spacetimes of bounded spatial geometry.

We now show that the various notions of Hadamard states discussed above are all, in fact, the same (except that MMV is only stated for quasifree states).
\begin{thm}\label{thm:FPequivalence}
	For any globally hyperbolic $\Mb$,
	a state $\omega$ on $\Zf(\Mb)$ is FP-Hadamard if and only if it is $\Nc_\Mb^+$-Hadamard. In particular, all MMV-Hadamard states are FP-Hadamard.
\end{thm}
\begin{proof}
	The forward implication has already been mentioned as a consequence of nonexpansion of the wavefront set under differential operators. Suppose that $\omega$ is $\Nc^+_\Mb$-Hadamard, with two-point function $W$. We will construct a Hadamard form bisolution $H\in \DD'(\Lambda^1\Mb\boxtimes \Lambda^1\Mb)$ satisfying~\eqref{eq:FPcondition}. An important observation is that 
	$W$ is itself a $K_\Mb$-bisolution obeying 
	$W\circ  (1\otimes D_\Mb^{\sharp,1} )=W$ and $\WF(W)=\Rc_\Mb^\Had$.
	This is seen because $(1\otimes P_\Mb)W=0$
	implies $(1\otimes\delta_\Mb)W=0$ and hence $W\circ  (1\otimes D_\Mb^{\sharp,1} )= (1\otimes D_\Mb)W=W$ and $(1\otimes K_\Mb)W=0$ -- a manifestation of the fact that $D_\Mb$ projects $K_\Mb$-solutions onto $P_\Mb$-solutions. Similarly we also have $(K_\Mb\otimes 1)W=0$, so $W$ is a $K_\Mb$-bisolution. However $W$ is not of Hadamard form, because it has antisymmetric part $\tfrac{1}{2}\ii E_{P_\Mb}^{\twopt}$ rather than  $\tfrac{1}{2}\ii E_{K_\Mb}^{\twopt}$. 
	
	To obtain the required Hadamard form bisolution $H$, we will construct an additional $K_\Mb$-bisolution with the same wavefront set, that does not alter~\eqref{eq:FPcondition} and has antisymmetric part $\tfrac{1}{2}\ii (E_{K_M}^{\twopt}-E_{P_M}^{\twopt})=\tfrac{1}{2}\ii m^{-2}E_{K_M}^{\twopt}\circ(1\otimes \dd_\Mb \delta_\Mb)$ (a smooth error could be tolerated, but will not be needed).
	This is achieved by invoking an auxiliary scalar field 
	subject to the $0$-form operator $K^{(0)}_\Mb=-\delta_\Mb\dd_\Mb + m^2$ on
	$\Lambda^0M$, i.e., the trivial line bundle $M\times\CC$ with bilinear pairing given by multiplication and the antilinear conjugation given by complex conjugation. 
	The corresponding  quantised theory $\Yf(K^{(0)}_\Mb)$ certainly admits Hadamard states~\cite{FullingNarcowichWald,GerardWrochna:2014} (or by Theorem~\ref{thm:Hadamard_norm_hyp}) and we choose any such state $\omega^{(0)}$ with corresponding two-point function $W^{(0)}$.
	One has $\WF(W^{(0)})=\Rc_\Mb^\Had$ and $W^{(0)}-\adj{W}^{(0)}=\ii E^{\twopt}_{K_\Mb^{(0)}}$. Because $\dd_\Mb K_\Mb^{(0)}=K_\Mb\dd_\Mb$, one has
	$\dd_\Mb E_{K_\Mb^{(0)}} = E_{K_\Mb}\dd_\Mb$ and consequently
	\begin{equation}
		(\dd_\Mb\otimes\dd_\Mb) E^{\twopt}_{K_\Mb^{(0)}} = - (\dd_\Mb E_{K_\Mb^{(0)}}\delta_\Mb)^{\twopt} = -( E_{K_\Mb}\dd_\Mb\delta_\Mb)^{\twopt},
	\end{equation}
	where the minus sign arises because $\delta_\Mb=-\adj{\dd_\Mb}$.
	It follows that
	\begin{equation}\label{eq:H_from_W}
		H = W - m^{-2} (\dd_\Mb \otimes \dd_\Mb) W^{(0)}
	\end{equation}
	has antisymmetric part $\tfrac{1}{2}\ii E_{K_M}^{\twopt}$, and is a $K_\Mb$-bisolution with
	$\WF(H) = \Rc_\Mb^\Had$ satisfying~\eqref{eq:FPcondition}. Thus $\omega$ is FP-Hadamard.

	The last statement holds because all MMV-Hadamard states are $\Nc_\Mb^+$-Hadamard.
\end{proof}

\section{Conclusion}
\label{sec:Conc}
 
We have given a detailed treatment of Hadamard states for a wide class of bosonic and fermionic theories and demonstrated how many standard properties can be generalised using our approach.  
In addition, we have given an application to measurement theory~\cite{FewVer_QFLM:2018}, by showing that nonselective state updates of Hadamard states are Hadamard, and proved the complete equivalence of the MMV definition of Hadamard states for the Proca field~\cite{MorettiMurroVolpe:2023} to the FP definition~\cite{Few&Pfen03}. 
The general results proved here should form a useful toolbox for further applications. Indeed, our definition has already been employed in a treatment of functionals for perturbative algebraic QFT in~\cite{HawkinsRejznerVisser:2023}. To conclude, we consider some other potential directions.

First, our discussion of bosonic theories has focussed on hermitian quantum field theories.
However, it is straightforward to encompass complex fields as is sketched for bosonic fields in Appendix~\ref{sec:complex}, by `doubling' them to hermitian theories. In the fermionic case, a similar doubling procedure was used for Dirac-type operators in Section~\ref{sec:Dirac-type}.
Second, one could develop an extended algebra of normal ordered products~\cite{Ho&Wa01} which contains the stress-energy tensor and other similar observables; this can then form the basis for Epstein--Glaser renormalisation theory~\cite{BrFr2000,Ho&Wa01,Ho&Wa02,Rejzner_book,Duetsch_book} for models based on decomposable Green-hyperbolic operators. Third, it would be natural to extend our treatment to theories with gauge freedom by means such as~\cite{HackSchenkel:2013} (which abstracts
methods used in e.g.~\cite{Dimock92,Few&Pfen03,FewsterHunt:2013,DappLang:2012,SandDappHack:2012} to the level of Green-hyperbolic operators). Fourth, one could switch from the smooth wavefront set to its analytic cousin, whereupon results such as the timelike tube theorem and Reeh--Schlieder theorem would be expected for analytically decomposable Green hyperbolic operators satisfying the analytic wavefront set condition
$\WF_A(E_P)\subset (\Vc^+\times\Vc^-)\cup(\Vc^-\times\Vc^+)$ using methods from~\cite{StrohmaierWitten:2024,StrohmaierWitten:2024b}. Fifth, in terms of measurement theory, an important question is to understand which selective measurements correspond to update rules that preserve the Hadamard class. 

Finally, we have not discussed the Hadamard parametrix in the general setting, relying entirely on the microlocal formalism. Partly that is because the motivation for constructing a parametrix is most appropriate for a generally covariant theory -- where one can aim for a description reducing to that of the Minkowski space theory at short scales -- and we treat operators that need not have a covariant description in all spacetimes. Nonetheless, as all two-point functions of Hadamard states have smooth differences, finding a parametrix is related to solving a cohomological problem~\cite{BrFrVe03} and it would be interesting to determine for which Green-hyperbolic operators this can be achieved.

%

{\small 
\paragraph{Acknowledgements} 
It is a pleasure to thank Benito Juárez-Aubry, Eli Hawkins, Onirban Islam, Daan Janssen, Christiane Klein, Gandalf Lechner, Valter Moretti, Kasia Rejzner,  Alexander Strohmaier, Rainer Verch, and Berend Visser for useful discussions in relation to this paper.

\paragraph{Funding}
The author's work is partly supported by EPSRC Grant EP/Y000099/1 to the University of York, and additional support during the programme `Quantum and classical fields interacting with geometry' from the Institut Henri Poincaré (UAR 839 CNRS-Sorbonne Université), and LabEx CARMIN (ANR-10-LABX-59-01) is gratefully acknowledged. For the purpose of open access, the author has applied a creative commons attribution (CC BY) licence to any author accepted manuscript version arising.
 
\paragraph{Data availability} No datasets were generated or analysed in this work.

\paragraph{Conflict of interest} The author has no conflict of interest to declare that is relevant to the content of this article.
}
\appendix

\section{Proof of Lemma~\ref{lem:WFuoTT}}
\label{sec:proofWFuoTT}

We start with an analogue of Theorem 2.5.14 in~\cite{Hoer_FIOi:1971}. We recall the notational convention that $\WF_0(u)=\WF(u)\cup 0$.
\begin{lemma}\label{lem:uoT} Let $u\in\DD'(B_2\otimes\Omega^{1/2}_{M_2})$.
	Suppose a continuous linear map $T:\Gamma_0^\infty(B_1^*\otimes\Omega^{1/2}_{M_1})\to\Gamma_0^\infty(B_2^*\otimes\Omega^{1/2}_{M_2})$ has the property that,
	for each compact $K_1\subset M_1$ there is a compact $K_2\subset M_2$
	so that $T\Gamma^\infty_{K_1}(B_1^*)\subset \Gamma^\infty_{K_2}(B_2^*)$.
	If $\WF(T^\knl)\cap (\dot{T}^*M_2\times 0_{T^*M_1})=\emptyset$ then
	$u\circ T\in\DD'(B_1\otimes\Omega^{1/2}_{M_2})$ has wavefront set
	\begin{equation}\label{eq:uoT}
		\WF(u\circ T) \subset \WF'(\adj{}(T^\knl))\bullet \WF_0(u) .
	\end{equation}
	for all $u\in\DD'(B_2\otimes\Omega^{1/2}_{M_2})$, where $\adj{}(T^\knl)\in \DD'((B_1\boxtimes B_2^*)\otimes\Omega^{1/2}_{M_1\times M_2})$ is obtained by reversing arguments, 
	$\adj{}(T^\knl)(h\otimes F)=T^\knl(F\otimes h)$. If, in addition, 
	$\WF(T^\knl)\cap (0_{T^*M_2}\times T^*M_1)=\emptyset$, then~\eqref{eq:uoT} simplifies to
	\begin{equation}\label{eq:simplified_WFuoT_bound}
		\WF(u\circ T) \subset \WF'(\adj{}(T^\knl))\bullet \WF(u).
	\end{equation}
\end{lemma}
\begin{proof}
	It is enough to prove this in the case where 
	$T:\Gamma_0^\infty(\Omega^{1/2}_{M_1})\to\Gamma_0^\infty(\Omega^{1/2}_{M_2})$ as the general statement follows by considering local frames
	for $B_1$ and $B_2$. 
	Let $u\in\DD'(\Omega^{1/2}_{M_2})$. Owing to the support and continuity conditions, one has $u\circ T\in \DD'(\Omega_{M_1}^{1/2})$.
	Informally, we would like to express $(u\circ T)(f)=((u\otimes 1)T^\knl)(1\otimes f)$, which is possible if the distributional product $(u\otimes 1)T^\knl$ exists and one can avoid the problem that the constant function is not a test function on $M_2$. 
	
	The H\"ormander criterion for the existence of the product is
	satisfied because
	$\WF(u\otimes 1)\subset \WF(u)\times 0_{T^*M_1}$, 
	while $\WF(T^\knl)\cap (\dot{T}^*M_2\times 0_{T^*M_1})=\emptyset$, so $\WF(u\otimes 1)+\WF(T^\knl)$  
	does not meet the zero section of $T^*(M_2\times M_1)$. 
	Consequently,
	the distributional product $(u\otimes 1)T^\knl$ exists 
	in $\DD'(\Omega_{M_2}\boxtimes \Omega_{M_1}^{1/2})$
	with
	\begin{equation}\label{eq:u1T}
		\WF((u\otimes 1)T^\knl)\subset 	\WF_0(u\otimes 1)+\WF_0(T^\knl)
		\subset \WF_0(u)\times 0_{T^*M_1}+\WF_0(T^\knl) .
	\end{equation}
	Using the support property of $T$, for any compact $K\subset M_1$ we may choose $\chi\in C_0^\infty(M_2)$ so that
	$(\chi u\circ T)(f)=(u\circ T)(f)$ for all $f\in\Gamma_{K}^\infty(B_1)$, and 
	\begin{equation}
		(u\circ T)(f)= (\chi u\circ T)(f) = ((u\otimes 1)T^\knl)(\chi\otimes f),
	\end{equation}
	which allows the wavefront set of $u\circ T$ to be estimated locally as
	\begin{equation}
		\WF(u\circ T)\subset \pr_{\dot{T}^*M_1} ( (0_{T^*M_2}\times \dot{T}^*M_1)\cap \WF((u\otimes 1)T^\knl)) ,
	\end{equation} 
	and since this bound is independent of $\chi$, it holds globally. 
	Using~\eqref{eq:u1T} we may compute
	\begin{align}
		\WF(u\circ T) &\subset \pr_{\dot{T}^*M_1}( (0_{T^*M_2}\times \dot{T}^*M_1)\cap (\WF_0(u)\times 0_{T^*M_1}+\WF_0(T^\knl))) \nonumber \\
		&=	\pr_{\dot{T}^*M_1} ((\dot{T}^*M_1\times 0_{T^*M_2} )\cap(0_{T^*M_1}\times \WF_0(u)+\WF(\adj{}T^\knl))\nonumber\\
		&= 	\pr_{\dot{T}^*M_1} ((\dot{T}^*M_1\times \WF_0(u))\cap \WF'(\adj{}T^\knl))\nonumber\\
		&\subset  
		\WF'(\adj{}(T^\knl))\bullet \WF_0(u),
	\end{align} 
	thus establishing~\eqref{eq:simplified_WFuoT_bound}. Here, in the second line, 
	we used the fact that the zero section in $\WF_0(T^\knl)$ does not contribute to the overall result, and reversed the order of factors. 
	The last part holds because $\WF'(\adj{}(T^\knl))\bullet  0_{T^*M_2}=\pr_{T^*M_1} \WF(T^\knl)\cap   (0_{T^*M_2}\times T^*M_1)=\emptyset$.
\end{proof}

We will  apply this result to $u\circ (T_1\otimes T_2)$ for $u\in \DD'((B_2\boxtimes B_2)\otimes \Omega^{1/2}_{M_2\times M_2})$ assuming that the $T_j$ both obey both assumptions from Lemma~\ref{lem:uoT}. 
Noting that 
\begin{align}\label{eq:WFT1T2_bd}
	\WF'(\adj{}((T_1\otimes T_2)^\knl))&\subset \{ (x,k;x',k';y,-l;y',-l')\in \dot{T}^*(M_1\times M_1\times M_2\times M_2):
	\nonumber\\
	&\qquad\qquad
	(y,l;x,k)\in \WF_0(T_1^\knl),~(y',l';x',k')\in \WF_0(T_2^\knl)
	\},
\end{align}
one easily checks that $T_1\otimes T_2$ satisfies the support condition and that $\WF((T_1\otimes T_2)^\knl)$ obeys both conditions in Lemma~\ref{lem:uoT}, so
\begin{equation}
	\WF(u\circ (T_1\otimes T_2)) \subset \WF'(\adj{}((T_1\otimes T_2)^\knl))\bullet\WF(u).
\end{equation}
Furthermore,~\eqref{eq:WFT1T2_bd} entails that, for conic $\Gamma\subset \dot{T}^*(M_2\times M_2)$, 
\begin{equation}
	\WF'(\adj{}((T_1\otimes T_2)^\knl))\bullet\Gamma = (\WF_0'(\adj{}(T_1^\knl))\times \WF_0'(\adj{}(T_2^\knl)))\bullet \Gamma,
\end{equation}
which simplifies to
\begin{equation}\label{eq:uoTTsimple}
	\WF'(\adj{}((T_1\otimes T_2)^\knl))\bullet\Gamma = (\WF'(\adj{}(T_1^\knl))\times \WF'(\adj{}(T_2^\knl)))\bullet \Gamma
\end{equation}
if $\Gamma\cap ((0_{T^*M_2}\times \dot{T}^* M_2 ) \cup 
(\dot{T}^*M_2\times 0_{T^*M_2}))=\emptyset$.

\begin{proof}[Proof of Lemma~\ref{lem:WFuoTT}] 
	We will apply~\eqref{eq:uoTTsimple} to $T^{\sharp,1}_j$ in place of $T_j$ and $B_j\otimes\Omega_j^{-1/2}$ in place of $B_j$.  As
	the kernel distributions of $T^{\sharp,1}_j$ and $T_j$ have the same supports and wavefront sets, it follows that the $T^{\sharp,1}_j$ 
	satisfy the support condition  
	and~\eqref{eq:WFuoTT} is immediate from~\eqref{eq:uoTTsimple}. 
\end{proof}

\section{Complex fields}\label{sec:complex} 

Let $\Mb\in\Loc$ and suppose $B$ is a finite-rank complex vector bundle over $\Mb$ equipped with a hermitian fibre metric, but not necessarily with a complex conjugation. The hermitian metric gives an antilinear isomorphism of $B$ to $B^*$ (or sections thereof) written $f\mapsto f^\star$, so that $\dlangle f^\star,h\drangle = (f,h)$. We use the same notation for the inverse isomorphism, and define a hermitian fibre metric on $B^*$ by
$(u,v)_{B^*}=(v^\star,u^\star)_{B}$. Next, we endow $B\oplus B^*$ with the direct sum hermitian fibre metric and a complex conjugation $\Cc$ given by
\begin{equation}
	\Cc \begin{pmatrix}
		f\\ u
	\end{pmatrix} = \begin{pmatrix}
	u^\star \\ f^\star
	\end{pmatrix},
\end{equation}
which makes $B\oplus B^*$ into a hermitian vector bundle in $\HVB$.
If $P$ is a formally hermitian Green hyperbolic operator (FHGHO) on $\Gamma^\infty(B)$, it is easily shown that $\sadj{P}$ is formally hermitian on $\Gamma^\infty(B^*)$ and $P\oplus \sadj{P}$ is a \rfhgho\ on $\Gamma^\infty(B\oplus B^*)$. 

We quantise the theory as $\Zf(P)=\Yf(P\oplus \sadj{P})$, also introducing 
complex fields
\begin{equation}
	\Phi_P(u) = \Upsilon_{P\oplus\sadj{P}}(0\oplus u), \qquad
	\Phi^\star_P(f) = \Upsilon_{P\oplus\sadj{P}}(f\oplus 0)
\end{equation}
for $u\in\Gamma_0^\infty(B^*)$, $f\in\Gamma_0^\infty(B)$, which clearly generate $\Zf(P)$. It is easily seen that the following relations are satisfied: 
	\begin{enumerate}[C1]
		\item $u\mapsto \Phi_{P }(u)$ and $f\mapsto \Phi^\star_{P }(f)$ are complex linear
		\item $\Phi^\star_{P}(f) = \Phi_{P}(f^\star)^*$
		\item $\Phi_{P }(\sadj{P}   u) = 0 = \Phi_{P}^\star(P f)$  
		\item commutation relations   
		\begin{equation}
			[\Phi_{P}(u),\Phi^\star_{P}(f)] = \ii \dlangle u,E_{P}f\drangle\1_{\Zf(P)},
			\qquad 
			[\Phi_{P}(u),\Phi_{P}(v)] = [\Phi^\star_{P}(f),\Phi^\star_{P}(h)] =0
		\end{equation}
	\end{enumerate}
	for all $f\in\Gamma_0^\infty(B)$, $u,v\in \Gamma_0^\infty(B^*)$.	
Equivalently, we could construct $\Zf(P)$ using 
the generators $\Phi_P(u)$ and $\Phi^\star_P(f)$ subject to 
relations C1--C4.
The generator $\Phi_P(u)$ may be interpreted as a quantisation of
classical linear functional $\phi\mapsto \dlangle u,\phi\drangle$ on $\Sol(P)$, while $\Phi_P^\star(f)$ quantises
the antilinear functional $\phi\mapsto \dlangle \phi^\star,f\drangle$.

For each $z\in\mathrm{U}(1)=\{z\in\CC:|z|=1\}$, the map $G_P(z):f\oplus u\mapsto 
\overline{z}f\oplus z u$ is a $\GreenHyp$-automorphism of $P\oplus \sadj{P}$ and therefore furnishes a group of automorphisms of $\Zf(P)$ given by 
$\gamma_P(z)=\Yf(G_P(z))$. The action on the complex fields is
\begin{equation}
	\gamma_P(z)\Phi_P(u) = z\Phi_P(u), \qquad 
	\gamma_P(z)\Phi^\star_P(f) = \overline{z}\Phi^\star_P(f)
\end{equation}
for $z\in\mathrm{U}(1)$, $f\in\Gamma_0^\infty(B)$, $u\in\Gamma_0^\infty(B^*)$.

Specialising, if $B$ is in fact a  hermitian vector bundle, with antilinear involution denoted $f\mapsto \overline{f}$, then there is also a bilinear pairing on $B$ and a linear isomorphism $f\mapsto f^T= \overline{f}^\star$ from $B$ to $B^*$, so that $\dlangle f^T,h\drangle = \langle f,h\rangle$. The inverse isomorphism is written with the same notation. In this case, one can define $\Psi_P(f):=\Phi_P(f^T)$ and $\overline{\Psi}_P(f):=\Phi_P^*(f)$ which generate $\Zf(P)$ as $f$ runs over $f\in\Gamma_0^\infty(B)$, and obey
\begin{enumerate}[C1$'$]
	\item $f\mapsto \Psi_{P }(f)$ and $f\mapsto \overline{\Psi}_{P }(f)$ are complex linear
	\item $\overline{\Psi}_{P}(f) = \Psi_{P}(\overline{f})^*$
	\item $\Psi_{P }(\overline{P}f) = 0 = \overline{\Psi}_P(P f)$  
	\item commutation relations   
	\begin{equation}
		[\Psi_{P}(f),\overline{\Psi}_{P}(h)] = \ii \langle f,E_{P}h\rangle\1_{\Zf(P)},
		\qquad 
		[\Psi_{P}(f),\Psi_{P}(h)] = [\overline{\Psi}_{P}(f),\overline{\Psi}_{P}(h)] =0
	\end{equation}
\end{enumerate}
for all $f,h\in\Gamma_0^\infty(B)$, where $\overline{P}f=\overline{P\overline{f}}$. This provides an
alternative presentation of $\Zf(P)$ in these circumstances.
The element $\Psi_P(f)$ represents the quantisation of the functional $\phi\mapsto \langle f,\phi\rangle$ on $\Sol(P)$.

Specialising further, any \rfhgho\ $P$ on $B$ (i.e., $P=\overline{P}$) may be quantised both as a hermitian bosonic field theory $\Yf(P)$ and as a complex bosonic field theory $\Zf(P)$. In this situation, there is an isomorphism between 
$\Zf(P)=\Yf(P\oplus \sadj{P})$ and  $\Yf(P)\otimes\Yf(P)=\Yf(P\oplus P)$ arising from the
Cauchy $\GreenHyp$-morphism from $P\oplus \sadj{P}$ to $P\oplus P$
given by  
\begin{equation}
	f\oplus u\mapsto \frac{1}{\sqrt{2}}\left[
(u^T+f)\oplus \ii(u^T-f)	\right];
\end{equation} 
here, one uses the
easily checked identity $\sadj{P}f^T = (Pf)^T$. In this situation the complex theory is two copies of the hermitian theory, as expected. The correspondence can be realised by writing 
\begin{equation}\label{eq:field_id_complex_rfhgho}
	\Phi_P(u) = \frac{1}{\sqrt{2}}\left(\Upsilon_P(u^T)\otimes \1_{\Yf(P)} +\ii \1_{\Yf(P)}\otimes \Upsilon_P(u^T)\right),
	\qquad u\in\Gamma_0^\infty(B^*),
\end{equation}
or equivalently,
\begin{equation}\label{eq:field_id_complex_rfhgho_psi}
	\Psi_P(f) = \frac{1}{\sqrt{2}}\left(\Upsilon_P(f)\otimes \1_{\Yf(P)} +\ii \1_{\Yf(P)}\otimes \Upsilon_P(f)\right), \qquad f\in\Gamma_0^\infty(B).
\end{equation}

Returning to a general FHGHO $P$, and dropping the assumption that $B$ possesses an antilinear involution, the two-point function of a state $\omega$ on $\Zf(P)=\Yf(P\oplus \sadj{P})$ is 
a distribution in $\DD'((B\oplus B^*)\boxtimes (B\oplus B^*))$ that combines the two-point functions
\begin{align}
	W_{\Phi\Phi}(u^{\natural,1}\otimes v^{\natural,1}) &= \omega(\Phi_P(u)\Phi_P(v)), \qquad
	W_{\Phi^\star\Phi^\star}(f^{\natural,1}\otimes h^{\natural,1}) = \omega(\Phi^\star_P(f)\Phi^\star_P(h)) \\
	W_{\Phi\Phi^\star}(u^{\natural,1}\otimes h^{\natural,1}) &= \omega(\Phi_P(u)\Phi^\star_P(h)), \qquad
	W_{\Phi^\star\Phi}(f^{\natural,1}\otimes v^{\natural,1}) = \omega(\Phi^\star_P(f)\Phi_P(v))
\end{align}
for $f,h\in\Gamma_0^\infty(B)$, $u,v\in \Gamma_0^\infty(B^*)$. Note that $W_{\Phi\Phi}$ 
and $W_{\Phi^\star\Phi^\star}$ are symmetric by virtue of the commutation relations C4. 
The state $\omega$ is described as quasifree if it is quasifree as a state on $\Yf(P\oplus\sadj{P})$, which can be expressed as the identity
\begin{equation}
	\omega(\ee^{\ii(\Phi_P(u)+\Phi^\star_P(f))}) = \exp \left[-\frac{1}{2}\left(
	W_{\Phi\Phi}(u\otimes u) + W_{\Phi\Phi^\star}(u\otimes f) + W_{\Phi^\star\Phi}(f\otimes u)+
	W_{\Phi^\star\Phi^\star}(f\otimes f)
	\right) \right]
\end{equation}
as formal double series in $f\in\Gamma_0^\infty(B)$, $u\in \Gamma_0^\infty(B^*)$. 

A state $\omega$ on $\Zf(P)$ is said to be \emph{gauge-invariant} if $\gamma_P(z)^*\omega=\omega$ for all $z\in\mathrm{U}(1)$. For a quasifree state $\omega$, gauge-invariance is equivalent to 
$W_{\Phi\Phi}=W_{\Phi^\star\Phi^\star}=0$. A particular example of this occurs if $P$ is a \rfhgho: if $\omega$ is a quasifree state on $\Yf(P)$ then the state $\omega\otimes\omega$
on $\Yf(P)\otimes\Yf(P)$ determines a state on $\Zf(P)$ (via the isomorphism described above) that is quasifree and gauge invariant. To see that $W_{\Phi\Phi}=W_{\Phi^\star\Phi^\star}=0$, we use a calculation based on~\eqref{eq:field_id_complex_rfhgho}, 
\begin{align} 
	2\Phi_P(f^T)\Phi_P(h^T) &= \Upsilon_{P}(f)\Upsilon_{P}(h)\otimes\1_{\Yf(P)} - \1_{\Yf(P)}\otimes \Upsilon_{P}(f)\Upsilon_{P}(h) \nonumber \\
	&\qquad +\ii\left(
	\Upsilon_{P}(f)\otimes\Upsilon_{P}(h) + \Upsilon_{P}(h)\otimes\Upsilon_{P}(f) 
	\right)
\end{align}
and the fact that quasifree states have vanishing $1$-point functions. The analogous calculation for $2\Phi^\star_P(f)\Phi^\star_P(h)$ reverses the sign before the $\ii$. We summarise this discussion as follows.
\begin{thm}\label{thm:complex_qf}
	Let $P$ be a $\Vc^\pm$-decomposable \rfhgho\ and suppose that $\Yf(P)$ admits a quasifree state $\omega$. Then $\omega\otimes\omega$ induces a gauge-invariant quasifree state on $\Zf(P)$ via the isomorphism $\Zf(P)\cong\Yf(P)\otimes\Yf(P)$ described above.  
\end{thm}

Now let $P$ be a general FGHO, and suppose that $P$ is $\Vc^\pm$-decomposable, $\WF(E_P^\knl)\subset (\Vc^+\times\Vc^-)\cup(\Vc^-\times\Vc^+)$. Noting that
$E_{\sadj{P}}=-\sadj{E}_P$ and $\sadj{E}_{P}^\knl (f\otimes u)=
E_P^\knl(u\otimes f)$, we have $\WF(\sadj{E}_{\sadj{P}})=\adj{\WF}(E_P^\knl)$ and therefore
$\sadj{P}$ is also $\Vc^\pm$-decomposable. By Theorem~\ref{thm:Hadamard_tools}(a), $P\oplus\sadj{P}$ is 
$\Vc^\pm$-decomposable. Conversely, $\Vc^\pm$-decomposability of
$P\oplus\sadj{P}$ implies that of $P$ because
$\WF(E_P)\subset \WF(E_P\oplus E_{\sadj{P}}) = \WF(E_{P\oplus\sadj{P}})$.
Thus decomposability of $P\oplus \sadj{P}$ and $P$ are equivalent.  With this in mind, 
we say that a state on $\Zf(P)$ is $\Vc^+$-Hadamard if and only if its equivalent state on  $\Yf(P\oplus\sadj{P})$ is $\Vc^+$-Hadamard. The properties of $\Vc^+$-Hadamard states of the complex theory may be read off from the theory developed in Section~\ref{sec:Hadamard} applied to $\Yf(P\oplus\sadj{P})$. In particular, if $\omega$ is $\Vc^+$-Hadamard then
the wavefront sets of $W_{\Phi\Phi^\star}$ and $W_{\Phi^\star\Phi}$ are contained in $\Vc^+\times\Vc^-$, while $W_{\Phi\Phi}$ and $W_{\Phi^\star\Phi^\star}$ are smooth because
the fact that they are symmetric implies $\WF(W_{\Phi\Phi})\subset (\Vc^+\times\Vc^-)\cap(\Vc^-\times\Vc^+)=\emptyset$, and similarly for $W_{\Phi^\star\Phi^\star}$.

If $P$ is a \rfhgho\ we can also make use of the properties of Hadamard states on $\Yf(P)$.
\begin{thm}\label{thm:complex_Had}
	Let $P$ be a $\Vc^\pm$-decomposable \rfhgho\ and suppose that $\Yf(P)$ admits a $\Vc^+$-Hadamard state $\omega$. Then $\omega\otimes\omega$ is a $\Vc^+$-Hadamard on  $\Yf(P)\otimes\Yf(P)$ and if $\omega$ is also quasifree, then $\omega\otimes\omega$ is also quasifree and gauge-invariant. In particular, $\Zf(P)$ admits a quasifree,  gauge-invariant $\Vc^+$-Hadamard state.
\end{thm}
\begin{proof}
	Theorem~\ref{thm:Hadamard_tools}(a,b) show that $P\oplus P$ is $\Vc^\pm$-decomposable and $\omega\otimes\omega$ is $\Vc^+$-Hadamard.
	The second part follows by Theorem~\ref{thm:complex_qf}, and the last statement
	holds because $\Yf(P)$ admits a quasifree $\Vc^+$-Hadamard state by Theorem~\ref{thm:Hadamard_props}(e).
\end{proof}

For example, let $P_\Mb=-\delta_\Mb\dd_\Mb+m^2$ be the Proca field operator on $\Mb\in\Loc$, as studied in Section~\ref{sec:Proca}, which is a \rfhgho. The complex theory is $\Zf(P_\Mb)=\Yf(P_\Mb\oplus\sadj{P}_\Mb)\cong\Yf(P_\Mb)\otimes\Yf(P_\Mb)$ (note that $\Yf(P_\Mb)$ was denoted $\Zf(\Mb)$ in Section~\ref{sec:Proca}). By Theorem~\ref{thm:ProcaHadamard}, $\Yf(P_\Mb)$ admits  $\Nc^+$-Hadamard states, so we may conclude that the complex Proca field admits quasifree gauge-invariant $\Nc^+$-Hadamard states on any $\Mb\in\Loc$. 

\end{document}